\documentclass[12pt]{article}

\usepackage[margin=1in]{geometry}
\usepackage{setspace}
\usepackage{booktabs}
\usepackage{multirow}
\usepackage{threeparttable}
\usepackage{amsmath}
\usepackage{amssymb}
\usepackage{graphicx}
\usepackage{array}
\usepackage{natbib}
\usepackage{xurl}
\usepackage{hyperref}
\usepackage{pdflscape}
\usepackage[skip=6pt]{caption}
\usepackage{adjustbox}
\usepackage{float}
\usepackage{amsthm}

\makeatletter
\let\primitiveinput\@@input
\makeatother
\newcolumntype{L}[1]{>{\raggedright\arraybackslash\hspace{0pt}}p{#1}}
\newcolumntype{C}[1]{>{\centering\arraybackslash\hspace{0pt}\exhyphenpenalty=10000}p{#1}}

\theoremstyle{plain}
\newtheorem{proposition}{Proposition}[section]

\theoremstyle{definition}
\newtheorem{assumption}{Assumption}[section]

\theoremstyle{remark}

\hypersetup{colorlinks=true,citecolor=blue,linkcolor=blue,urlcolor=blue}

\title{\textbf{Environmental Threat and the Nation:\\
Earthquake Risk, Distributive Priority, and Expressive Attachment}}
\author{Hector Galindo-Silva\thanks{Department of Economics,
Pontificia Universidad Javeriana.
Email: \href{mailto:galindoh@javeriana.edu.co}{galindoh@javeriana.edu.co}.}}
\date{\today}

\begin{document}

\maketitle

\begin{abstract}
This paper studies how long-run earthquake risk shapes national identity,
separating a distributive margin (national membership as a rule for allocating
scarce resources) from an expressive margin (pride, willingness to fight, and
affective attachment).  Linking World Values Survey respondents (1981--2022; 63
countries, 494 subnational regions) to subnational seismic-risk geography, I
find that people living closer to high-risk zones express stronger national
in-group orientation: more pride, more willingness to fight, and more priority
for nationals when jobs are scarce.  Family attachment and out-group hostility
do not rise, while religiosity increases in parallel.  The expressive margin is
conditional: the pride response is pronounced where state--religion alignment
and a cohesive religious field lend the symbolic infrastructure to cast
disaster as a shared national ordeal, and indistinguishable from zero where
they do not.  A
complementary design exploiting earthquakes between adjacent survey waves finds
no average short-run response, yet the response it does detect concentrates
among older, place-attached residents who cannot leave---consistent with
attitudes tracking a chronic, inescapable risk rather than single events.  Together, the results
point to a demand-side origin of national attachment: where a covariate shock
would overwhelm local and family insurance, people turn to larger communities
of protection and meaning---the nation and religion---a logic I formalize in a
simple social-interaction model.
\end{abstract}



\newpage

\section{Introduction}

Large destructive environmental shocks unsettle more than material conditions.
They destroy homes and assets, disrupt routines, and make the boundaries of
mutual obligation more visible.  One possible margin of adjustment is identity:
the communities through which individuals interpret risk, define obligations,
and locate protection.  While existing work shows that earthquakes shape one such identity margin
---religiosity \citep{bentzen2019, belloc2016}---this paper asks whether
long-run earthquake risk also shapes other identity margins, with particular
focus on national identity.

The nation is a natural object of study because it combines two functions that
become salient under collective threat.  First, it is a practical membership
boundary: when jobs, aid, shelter, or public resources are scarce, national
membership can be invoked as a rule for allocation \citep{brubaker1992,
miller1995}.  Second, it can be a moral community: a collective ``we'' that
turns suffering into duty, sacrifice, resilience, pride, or willingness to
defend the country \citep{anderson1983}.  These two functions need not move
in the same way \citep{shayo2009}.  A jobs-priority rule requires a recognized
inside-outside boundary.  Expressive national attachment requires more: the
shock must be narrated as a threat to a morally meaningful national community.
Economics has largely studied how states build these functions from
above---through schooling, conscription, and public goods; I instead ask how
they are demanded from below, when an environmental threat makes smaller
communities of protection unreliable.

Both functions are visible in recent earthquake episodes.  After the February
2023 Turkey--Syria earthquake, acute shortages of shelter, emergency aid, and
reconstruction capacity made the distinction between Turkish nationals and
Syrian refugees politically salient as a rule for allocating scarce public
resources.  By contrast, the expressive margin appears in symbolic
reconstruction narratives after the 2011 T\={o}hoku earthquake and tsunami, the
2008 Sichuan earthquake, and the 2015 Nepal Gorkha earthquake, where public
idioms and monuments helped translate destruction into national resilience,
solidarity, and recovery.

I study this question by linking individual-level national-identity responses from the World
Values Survey (WVS) 1981--2022 to subnational measures of long-run
earthquake-risk geography, covering 63 countries and 497 matched subnational
regions, of which 494 enter the baseline national-core sample.
The baseline specification compares individuals surveyed in the same country
and year but living in regions with different distance to high-risk earthquake
zones.  The main outcomes include a national-core index and its three
components.  The index combines being very proud, willingness to fight, and
priority for nationals in scarce jobs.  Because these items capture
conceptually distinct responses, the components are also reported and
interpreted separately.

The paper's main result is that proximity to high-risk seismic zones is
associated with stronger national in-group orientation.  The standardized
effect on the national-core index is $-0.098$ standard deviations (negative because
the treatment is distance); because that composite averages the two margins I
treat as distinct, I read its components as co-primary.  The standardized
effects are $-0.052$ for being very proud, $-0.066$ for willingness to fight,
and $-0.075$ for jobs priority.

These findings, which condition on country-year fixed effects and geographic
controls for latitude and distance to the coast, are robust across a
wide set of specification checks---night-light intensity, richer individual
and geographic controls, ethnic-group fixed effects, earthquake-event lags,
and alternative seismic-risk measures.  Three sets of diagnostics provide
evidence consistent with the main identifying assumption: conditional on the
controls, within-country variation in seismic-risk geography is unrelated to
other determinants of national identity.  Distance to high-risk zones predicts
realized earthquake exposure; the national-identity pattern is not reproduced
by proximity to tsunami, storm, or volcanic hazards; and household and
regional characteristics are balanced across seismic-risk exposure.  Together, these exercises support
 interpreting the estimate as a seismic-risk effect
rather than a demographic or geographic artifact---one that observed confounders
do not explain and that plausible unobserved confounders are unlikely to
overturn.

To distinguish this long-run relationship from attitudinal adjustment following
realized earthquakes, I complement the baseline with a stacked
repeated-cross-section design.  For each district observed in adjacent WVS
fieldwork periods, the design stacks independent samples of respondents from
the earlier and later endpoints and compares the change in survey outcomes in
districts that experience a qualifying earthquake between surveys with the
contemporaneous change in untreated districts in the same country-wave pair.
District-transition fixed effects absorb characteristics that remain constant
across the two endpoints.  The design detects no robust positive post-event
shift between survey endpoints in the
national-core index, very proud, jobs priority, or religiosity. This specification estimates a different object from
the baseline---the response associated with one or more qualifying earthquakes
between survey endpoints rather than the
association with a persistent hazard environment---so the two sets of
estimates need not coincide.  Yet the average null conceals a revealing
pattern: a realized earthquake does move attitudes, but only among older,
place-attached residents---those least able to escape the permanent risk it
makes salient---and not among the mobile young. This is the dynamic signature
the chronic-risk account predicts, and it reconciles the two designs: a
realized earthquake moves national identity precisely where, and for whom, the
latent risk is binding.  Together, these results are consistent with national
orientations developing through the sustained salience of chronic risk rather
than adjusting mechanically after each earthquake.

The composition of this long-run relationship provides a second clue about its
interpretation.  It is concentrated in national and religious outcomes rather
than in social attachments generally.  Family importance and trust are near
zero, whereas religiosity rises near seismic risk.  Ethnic-boundary and out-group
diagnostics---rejection of different-race or different-language neighbors,
generalized distrust, and attitudes toward ethnic diversity---show no parallel
rise.  Thus, seismic-risk geography is associated with the national and
religious domains without producing a general strengthening of close
attachments or hardening of group boundaries.  

To make sense of these patterns, I propose a conceptual framework based on covariate risk. When environmental threat is highly covariate---meaning a major disaster damages many nearby households simultaneously---local mutual-aid networks become a less reliable hedge. Because such a shock hits insurer and insured alike, local exchange has limited capacity to absorb the loss. The distributive channel fills this void by making national membership an allocation rule for scarce
resources---jobs, aid, shelter---that flow through the state; wherever the
state is the residual provider, the national-local boundary becomes salient.
The expressive channel translates destruction into a collective national
ordeal---resilience, duty, sacrifice, pride---but only where a symbolic
vocabulary is available to do the narrating.  Religious institutions lend it:
their repertoires of duty, sacrifice, and protection can be attached to the
nation, not by fusing religion and nation but by supplying symbolic
infrastructure.  Those repertoires are more available where
national and religious institutions point toward the same collective community
and the religious field is cohesive.  As institutional alignment increases and
fractionalization decreases, the expressive channel activates and the two
channels converge.  A simple social-interaction model formalizes this
intuition: individuals allocate effort across local insurance, distributive
national priority, expressive attachment, and religious coping, and covariate
risk depresses the return to the local margin.

A further set of results examines the two-channel prediction directly.
State--religion alignment amplifies the expressive association more clearly than
the jobs-priority response: the pride coefficient is large where alignment is
high and indistinguishable from zero where it is low, a contrast that survives
country-level wild-bootstrap inference and equal country weighting.  Religious
fractionalization likewise attenuates the expressive outcomes, whereas ethnic
fractionalization does not display a comparable pattern.  These estimates are
ordered as the proposed mechanism predicts, although formal cross-margin
comparisons remain imprecise. Taken together, these patterns fit less
directly with general-insecurity, ethnic-threat, or top-down authoritarian
mobilization accounts, each of which would predict broader or differently
moderated responses than the evidence shows.

\bigskip

This paper contributes first to the economics of nation-building, which has
studied national identity largely from the supply side: states and
elites deliberately produce it---through schooling, conscription, propaganda,
and public goods---to wage mass war \citep{alesinareichriboni2020}, to
homogenize and govern diverse populations \citep{alesinagiulianoreich2021}, or
to consolidate newly acquired territory \citep{dehdarigehring2022}.  I document
a complementary demand-side origin: where chronic covariate risk makes
local and family insurance unreliable, individuals turn to larger communities
of protection---the nation and religion---without any elite acting to supply
national feeling.  The two margins I separate map onto the two functions this
literature assigns the modern nation-state: a distributive function (membership
as a claim on collective resources) and an expressive function---the motivation
to sacrifice, indeed the very willingness to fight that
\citet{alesinareichriboni2020} place at the center of mass-army nationalism.  My
institutional moderators then identify when the expressive margin
activates from below, a demand-side counterpart to that literature's account of
when top-down nation-building succeeds.

The paper also engages several adjacent literatures.  Within the economics of
culture and identity, it builds on the identity-economics foundations of
\citet{akerlof2000} and \citet{shayo2009} to study how environmental risk
shapes the composition of national attachment, separating distributive from
expressive nationalism as analytically distinct objects
\citep{guiso2006, alesina2015, giulianonunn2021}.  The covariate-risk
mechanism connects to the informal-insurance literature on covariate shocks
\citep{townsend1994, morduch1995, fafchamps2003}.  Within the
literature on natural disasters and cultural change, the paper extends the
seismic-risk design of \citet{bentzen2019} to national identity, complements
\citet{winkler2021}'s evidence that disasters tighten social norms, and
connects to \citet{belloc2016}'s finding that earthquakes shaped the political
power of religious institutions in medieval Italian cities.

In the political economy of religion and nationalism, it documents that
institutional integration between religious and national authority shapes
which channel of nationalism responds to covariate risk.  The state--religion
alignment moderator maps directly onto the market-concentration mechanism of
\citet{barromccleary2005}, in which monopoly religious institutions fuse most
readily with national narratives; the broader reading of religion as symbolic
infrastructure for national community draws on \citet{bellah1967},
\citet{casanova1994}, \citet{berman2000}, \citet{mcclearybarro2006},
\citet{cosgelmiceli2009}, and \citet{rubin2017}.  Finally, it speaks to the
debate on whether collective shocks generate solidarity or exclusion
\citep{enke2019, miguel2004, voors2012, bauer2016, cassar2017, hanaoka2018,
depetris2020}: the evidence points to a targeted national allocation
preference and a conditionally activated expressive attachment, not a broad
rise in out-group hostility or a general social-cohesion effect---a two-margin
distinction that connects to models of identity and political conflict
\citep{shayo2009, bonomi2021}.  The closest antecedent on threat and
large-scale identity is \citet{gehring2022}, who shows that the external
military threat from Russia's invasion of Ukraine strengthened European Union
identity.  Environmental threat differs on a dimension that matters for
theory: an earthquake supplies no human antagonist.  The setting therefore
isolates how a membership boundary forms in response to an impersonal
threat---and the out-group diagnostics below show that it can form without
hostility toward any out-group.

The paper is organized as follows.  Section~\ref{sec:data} describes the data
and outcome construction.  Section~\ref{sec:strategy} presents the
identification strategy.  Section~\ref{sec:results} reports the main results
and robustness checks.  Section~\ref{sec:mechanisms} develops the mechanism,
tests its implications, and discusses alternatives.  Section~\ref{sec:conclusion}
concludes.

\section{Data and Measurement}\label{sec:data}

The analysis combines two main data sources: individual-level survey responses on national identity and subnational measures of long-run earthquake-risk geography. The survey data come from the WVS Trends 1981--2022 file \citep{haerpfer2022wvstrend}, which harmonizes key variables across waves and preserves the region identifiers needed to locate respondents.\footnote{The World Values Survey time-series data can be consulted through the WVS data portal at \url{https://www.worldvaluessurvey.org/WVSDocumentationWVL.jsp}.} The estimation sample spans 63 countries, providing between 145,000 and 166,000 observations for the main outcomes.

I use three WVS measures for the main analysis: an indicator for being very proud of one's nationality, willingness to fight for the country, and jobs priority for nationals.\footnote{The exact harmonized WVS items are E012 (``Willingness to fight for country''), which I code as one for ``yes,'' and C002 (``Jobs scarce: Employers should give priority to (nation) people than immigrants''), coded as one for agreement. For national pride, I use item G006 (``How proud of nationality''), whose original responses range from 1 (``very proud'') to 4 (``not at all proud''). In the weighted matched sample, 62.7 percent report being very proud and 27.6 percent quite proud. Because a dummy combining these two categories equals one for 90.3 percent of respondents---leaving little remaining variation---I code G006 as one strictly for ``very proud'' to capture meaningful variation in the intensity of national attachment.} I analyze these items both separately and averaged into a composite ``national-core index,'' which is observed when at least two components are non-missing. Because these items capture related but distinct aspects of national orientation rather than repeated indicators of a single latent trait, the index serves as a formative summary.\footnote{Because the index is formative rather than reflective, internal-consistency statistics such as Cronbach's alpha are not the appropriate yardstick: the low inter-item correlation ($\alpha=0.34$, with pairwise correlations of $0.11$--$0.21$) is expected when the components are distinct constituents of national orientation rather than interchangeable measures of a single latent trait, and a high value would instead indicate redundancy. The composite accordingly carries no independent inferential weight---each component is reported and signed separately, and the multiple-testing-adjusted evidence rests on the components themselves, so its greater precision reflects the pooling of a common signal rather than artificial variance reduction.} Consequently, I report the index alongside each individual component to distinguish the broad national response from its specific functional margins.

Conceptually, these components map directly to the theoretical framework. Jobs priority asks whether national membership should govern access to scarce employment, serving as the practical distributive outcome.\footnote{This item connects to the welfare chauvinism literature, which documents preferences for reserving jobs and benefits for co-nationals \citep{vanoorschot2006, eger2010, reeskens2012}. Although the paper studies long-run seismic-risk geography rather than immigration exposure, the mechanisms may overlap---jobs priority may partly capture material boundary-drawing---so it is treated as a distinct co-primary outcome and probed directly in Section~\ref{sec:alternatives}.} Conversely, being very proud and willingness to fight capture two manifestations of expressive national attachment: affective intensity and costly commitment.

To validate the main results and explore alternative mechanisms, I construct several auxiliary and diagnostic measures from the WVS. First, to assess broader national orientation, I use a Lan--Li-style index \citep{lanli2015} alongside indices for economic/exclusionary nationalism and country-over-world identity.\footnote{The economic/exclusionary index averages jobs priority, restrictive immigration policy, rejection of immigrant neighbors, and distrust of other nationalities. Country-over-world identity subtracts world-citizen identity from country-citizen identity. The additional WVS item labels are: E001/E002, ``Aims of country: first/second choice,'' coded for selecting ``strong defence forces''; E069\_02, ``Confidence: Armed Forces''; E143, ``Immigrant policy,'' with values from ``let anyone come'' to ``prohibit people from coming''; A124\_06, ``Neighbours: Immigrants/foreign workers''; G007\_36\_B, ``Trust: People of another nationality''; G019, ``I see myself as a world citizen''; G021, ``I see myself as citizen of the [country] nation''; and G005, ``Citizen of [country].'' All indices are coded so that higher values represent stronger national in-group orientation, stronger exclusionary orientation, or stronger country-over-world identity.} Second, to test the theoretical boundaries of the covariate-risk mechanism, I evaluate outcomes related to family attachment (importance, trust, and making parents proud) and religiosity (importance, practice, and belief).\footnote{The family item labels are A001, ``Important in life: Family''; D001\_B, ``How much you trust: Your family (B)''; D001, ``How much do you trust your family''; and D054, ``One of main goals in life has been to make my parents proud.'' The religiosity item labels are F063, ``How important is God in your life''; F034, ``Religious person''; F028, ``How often do you attend religious services''; F050, ``Believe in: God''; F051, ``Believe in: life after death''; and F064, ``Get comfort and strength from religion.''} Finally, to rule out generalized social distrust and assess state legitimacy, I incorporate measures of institutional confidence\footnote{The WVS institutional confidence block asks about confidence in specific institutions on a four-point scale. The item labels used here are E069\_01, churches; E069\_02, armed forces; E069\_06, police; E069\_07, parliament; E069\_08, civil services; E069\_09, social security system; E069\_11, government; E069\_17, justice system/courts; and E069\_20, the United Nations.} alongside ethnic-boundary and trust diagnostics.\footnote{Because the WVS lacks a direct repeated cross-national measure of ethnic-identity salience, I treat the available ethnic-boundary and trust variables as diagnostics. These include whether ethnic diversity is seen as eroding unity, rejection of different-language/race/foreign-origin neighbors, and generalized distrust (item labels: G032, A124\_43, A124\_02, A124\_16, G007\_35\_B, and A165).}

The second main data source consists of district-level geographic exposure measures to seismic hazard, built from administrative boundary polygons and global hazard layers. The primary exposure variable is the geodesic distance from each district's administrative boundary---delimited using ESRI district shapefiles---to the nearest polygon of Modified Mercalli zones 3--4. These high-risk zones are defined in the Global Seismic Hazard Assessment Programme (GSHAP) database, maintained by UNEP/GRID under the UN International Decade for Natural Disaster Reduction. To map these hazard polygons to subnational WVS regions in a consistent cross-country format, I use the harmonized district-level exposure measures and region-label crosswalk assembled by \citet{bentzen2019}.\footnote{The replication materials for \citet{bentzen2019} are available from \textit{The Economic Journal} website (\url{https://academic.oup.com/ej/article/129/622/2295/5490325\#supplementary-data}). Only the geographic exposure variables---distance to high-risk seismic zones, absolute latitude, distance to the coast, and the other district geography used below---are taken from this source; all national-identity outcomes and individual controls are constructed directly from the WVS time-series file.} 

In the regressions, distance is measured in 1,000 kilometers, so larger values indicate lower exposure to seismic risk. Figure~\ref{fig:map} shows the 497 matched region centroids, colored by this distance, over the global footprint of USGS magnitude-$\geq$5.5 epicenters. Because the baseline exposure variables are time-invariant geographic measures, I assign them to all matched WVS respondents within a region. As complements to these geographic measures, I also incorporate alternative hazard distances and geographic characteristics for specificity, balance, and robustness checks.\footnote{These include distances to tsunami zones, volcanoes, storm tracks, and fault lines; night-light intensity from the DMSP Operational Linescan System for 2000 \citep{elvidge1997dmsp}; population density from the LandScan global population grid for 2000 \citep{dobson2000landscan}; elevation and terrain ruggedness from the Shuttle Radar Topography Mission \citep[SRTM;][]{farr2007srtm,nunnpuga2012}; mean annual temperature for 1961--1990 from WorldClim \citep{hijmans2005worldclim}; arable land from FAO land-classification data; distance to the national capital, computed as the geodesic distance from each district's administrative centroid to the country's capital city using Natural Earth populated-place coordinates; and distance to the nearest international land border, computed as the geodesic distance from each district centroid to the nearest point on the country's border polygon from Natural Earth administrative boundary data. Both are measured in 1,000 kilometers.} 

The baseline design exploits long-run seismic-risk geography rather than realized earthquake events. However, I use historical event data both to validate that the distance measure predicts actual earthquake histories and to provide complementary timing evidence through a stacked repeated-cross-section design.\footnote{This design compares independent respondent samples at the two endpoints of the same district transition and codes whether one or more qualifying earthquakes occurred during the actual fieldwork interval between WVS waves; the design is described in Section~\ref{sec:strategy_event}. The underlying USGS catalog can be queried at \url{https://earthquake.usgs.gov/earthquakes/search/}.}

To account for country-level characteristics and test the mechanisms proposed in the theoretical framework, I augment the dataset with several country-level indicators. GDP per capita is drawn from the World Development Indicators \citep{worldbank2026wdi}, while state capacity and democracy are measured using V-Dem, Polity5, and the Worldwide Governance Indicators \citep{vdem2026, marshallpolity2018, worldbank2026wgi}. To capture the theoretical mechanisms of coordination friction and symbolic synergy, I use time-invariant ethnic and religious fractionalization indices from \citet{alesina2003}\footnote{Specifically, I take the Alesina religious and ethnic fractionalization measures from the Quality of Government Standard Cross-Section \citep{teorellqog2024}. As a complement, I also use the Historical Index of Ethnic Fractionalization (HIEF), an annual country-level panel covering 1945--2013 \citep{drazanova2020}; unlike the time-invariant Alesina index, it varies within country across WVS waves and serves as a time-varying ethnic fractionalization check.} and state--religion alignment measures from the Religion and State (RAS) Project \citep{arda_ras3, arda_ras_constitutions2022}. Finally, I draw on event histories from the Global Disaster Identification Numbers (GDIS) and the Emergency Events Database (EM-DAT) to validate realized exposure and to construct revealed-disaster-response moderators.\footnote{GDIS is available through NASA SEDAC at \url{https://sedac.ciesin.columbia.edu/data/set/pend-gdis-1960-2018-disasterlocations}; EM-DAT through the Centre for Research on the Epidemiology of Disasters at \url{https://www.emdat.be/}. The validation exercise assigns GDIS earthquake locations to nearby WVS/Bentzen regions, while the heterogeneity exercises use the same sources for macro-moderators.}

Table~\ref{tab:coverage} reports variable-level coverage in the matched WVS--geographic data, showing how sample size varies across outcomes, moderators, and waves. Regression samples in later tables are slightly smaller when baseline controls are required to be non-missing.

\section{Identification Strategy}\label{sec:strategy}

\subsection{Baseline Design}

The first of the paper's two empirical designs exploits within-country cross-sectional variation in environmental threat \citep{bentzen2019}. Specifically, it compares individuals surveyed in the same country and year, but residing in subnational regions with different long-run exposure to seismic risk:
\begin{equation}
    y_{irct} = \beta \, \text{DistEQ34}_{r} + X'_{irct}\Gamma + \alpha_{ct} + \varepsilon_{irct},
    \label{eq:baseline}
\end{equation}
where $y_{irct}$ is an outcome for individual $i$ in subnational region $r$, country $c$, and survey year $t$. The variable $\text{DistEQ34}_{r}$ is the distance, in 1,000 kilometers, from the region's administrative boundary to the nearest high-risk earthquake zone (Modified Mercalli zones 3--4). The vector $X_{irct}$ contains individual and geographic controls, and $\alpha_{ct}$ denotes country-by-year fixed effects. The baseline controls include age, age squared, indicators for being male and married, absolute latitude, and distance to the coast. Regressions use WVS survey weights, and standard errors are clustered at the subnational-region level.  The coefficient $\beta$ captures the cross-region gradient in attitudes with respect to long-run seismic geography. Because the treatment variable is a distance measure, a negative coefficient implies that national in-group orientation is stronger in communities located closer to high-risk earthquake zones.\footnote{The matched data contain 497 subnational regions; the national-core estimation sample has 494 regions in 63 countries, yielding 145,000--166,000 respondents across components. Because $\text{DistEQ34}_{r}$ varies only across regions, identifying variation is at the region level. Of the 63 countries, 55 (and 136 of 171 country-year cells, 90 percent of respondents) have within-country treatment variation; restricting to these varying cells, collapsing to region--country-year averages, or weighting every country equally leaves the composite coefficient intact.}

The main identifying assumption of (\ref{eq:baseline}) is that, within a
country-year and conditional on the baseline geographic controls, distance to
high-risk seismic zones is uncorrelated with unobserved determinants of
national in-group orientation.  The design is attractive because seismic-hazard
zones are determined by deep tectonic structure and are fixed on the timescale
relevant for modern survey attitudes.  Country-year fixed effects absorb
national institutions, common media environments, national political shocks,
survey-wave composition, and any other factor shared by all respondents in the
same country and year.  Identification therefore comes from comparing regions
inside the same national context, rather than from cross-country differences in
history, income, religion, or political institutions.  Latitude and distance to
the coast further absorb broad spatial patterns that could otherwise connect
earthquake geography to climate, settlement, market access, or coastal
development.

\subsection{An Event-Study Complement}\label{sec:strategy_event}

While the baseline design identifies a long-run association with chronic seismic risk, I complement it with a stacked repeated-cross-section design to test whether mean attitudes shift by the next observed survey endpoint after a realized earthquake. For each district observed in two consecutive survey waves, I stack independent respondent samples from the earlier and later endpoints and estimate:
\begin{equation}
    y_{ir\tau s} = \delta\left(\text{Post}_{s}\times\text{EQ}_{r\tau}\right) + X_{ir\tau s}'\Gamma + \alpha_{r\tau} + \eta_{c\tau s} + \varepsilon_{ir\tau s},
    \label{eq:eventstudy}
\end{equation}
where $s\in\{0,1\}$ denotes the earlier or later endpoint of district transition $\tau$, $\alpha_{r\tau}$ are district-transition fixed effects, and $\eta_{c\tau s}$ are country-wave-pair-by-endpoint fixed effects. The vector $X_{ir\tau s}$ contains baseline demographic controls (age, age squared, sex, and marital status). The indicator $\text{Post}_{s}=\mathbf{1}[s=1]$ denotes the later endpoint of the transition. The treatment variable, $\text{EQ}_{r\tau}$, equals one if the district experienced at least one qualifying earthquake between the two survey waves, and zero otherwise.\footnote{Specifically, $\text{EQ}_{r\tau}$ equals one if at least one USGS magnitude-$\geq$5.5 earthquake occurred within 150 km of the district centroid strictly after fieldwork for the earlier wave ended and strictly before fieldwork for the later wave began. This interval-specific definition avoids assigning earthquakes already reflected in the earlier survey or omitting events that occurred during long gaps between waves. Regressions use WVS respondent weights, and inference relies on wild-cluster-bootstrap $p$-values clustered at the country level.}  Because the main effect of $\text{EQ}_{r\tau}$ is absorbed by the transition fixed effects $\alpha_{r\tau}$, and the main effect of $\text{Post}_{s}$ is absorbed by $\eta_{c\tau s}$, Equation~\eqref{eq:eventstudy} recovers the difference-in-differences coefficient $\delta$. This specification compares differences in mean outcomes across independent respondent samples at the two endpoints, between districts that did and did not experience one or more qualifying earthquakes within the same country-wave pair. It retains respondent-level observations rather than collapsing them into noisy cell means.

The primary identifying assumption is conditional parallel trends: absent a realized earthquake, mean attitudes in treated and untreated districts within the same country-wave pair would have evolved in parallel. This assumption is plausible because earthquake timing is governed by deep geology rather than social attitudes or survey calendars. Furthermore, the transition fixed effects absorb all time-invariant district characteristics, while the country-wave-pair-by-endpoint fixed effects absorb any country-wide attitudinal shifts between waves.\footnote{The event-study sample comprises 641 district transitions across 43 countries, of which 120 transitions (in 19 countries) are treated. Survey-year gaps between endpoints range from two to eleven years, with a median of six years; the pooled coefficient should therefore be interpreted as a between-wave post-event shift rather than an immediate response. Because the treated sample is moderate, the design's statistical power to detect post-event shifts of the cross-sectional magnitude is limited---a caveat discussed alongside the results.}  Crucially, I use this design not as a substitute for the cross-sectional analysis, but as a conceptual bound for the theoretical framework. As formalized in Section~\ref{sec:mechanisms}, if nationalist orientations develop as institutional adaptations to the sustained salience of chronic covariate risk, one or more qualifying earthquakes between two survey waves need not produce a detectable discrete shift in the average respondent---though the same logic predicts that an event should move those most bound to the now-salient permanent risk, a conditional prediction I test directly in Section~\ref{sec:results_event}.

\subsection{Threats to Identification and Diagnostic Evidence}

The two empirical designs rely on distinct sources of variation and therefore face different threats to identification. Separating them elucidates their respective identifying assumptions. I first address the cross-sectional relationship estimated by Equation~\eqref{eq:baseline}, which must confront time-invariant confounding and spatial sorting. I then discuss the realized-event design of Equation~\eqref{eq:eventstudy}, which absorbs time-invariant district-transition heterogeneity while comparing repeated cross-sections at two survey endpoints, but relies on a conditional parallel-trends assumption. Throughout, the diagnostics previewed below are fully reported alongside the main results in Section~\ref{sec:results}.

For the cross-sectional design, the primary threat is omitted variable bias: seismic exposure might correlate within countries with unobserved geographic or institutional features that independently shape national identity. To address this, the baseline specification conditions on country-by-year fixed effects, coastal distance, and latitude. Formal balance tests across a broad vector of potential demographic, economic, and geographic confounders show that most observed characteristics are comparable across high- and low-risk regions, with terrain ruggedness as the clear exception (see Online Appendix Table~\ref{tab:covariate_balance}). The estimated relationship survives further spatial controls and a battery of placebo-hazard tests.\footnote{As shown in Section~\ref{sec:results}, the coefficient remains stable after controlling for ruggedness, night-light intensity, and richer individual covariates. Similar results arise with alternative seismic-risk measures but not with proximity to tsunamis, storms, or volcanoes.} These exercises cannot eliminate unobserved confounding, but they rule out several direct demographic and geographic explanations.

A related concern is endogenous spatial sorting: individuals with inherently stronger national attachments might systematically migrate to (or remain in) high-risk areas. Because seismic risk is determined by deep tectonic structures that predate human settlement, historical sorting is difficult to test directly. However, restricting the sample to plausibly lower-mobility groups---specifically, older cohorts and agricultural workers---and absorbing ethnic-group composition leaves the core and jobs-priority estimates similar in sign and magnitude. The sign of the heterogeneity also runs against a sorting account: mobility-based sorting would concentrate the association among the mobile, yet it is the least mobile---older cohorts here and, in the realized-event design, place-attached residents --who drive it. This reduces, but does not eliminate, concern about contemporary self-selection.

A third concern regarding the cross-section relates to construct validity rather than confounding. Because distance measures long-run structural hazard rather than realized shaking, classical measurement error could bias the estimates toward zero. A treatment-validation exercise supports the construct validity of the measure: proximity to high-risk zones is associated with the incidence, frequency, and affected population of realized damaging earthquakes over the preceding twenty years (see Table~\ref{tab:treatment_validation}). The geographic measure therefore captures meaningful variation in realized exposure rather than serving as a mere cartographic artifact.

District-transition fixed effects address these time-invariant concerns, but
the stacked repeated-cross-section design requires a different identifying
assumption: conditional parallel trends. Because earthquake timing is
governed by geology rather than survey calendars, direct anticipation is less
plausible than in settings with policy-assigned treatment.
Instead, the primary threats are post-event compositional shifts---such as
mortality or selective displacement---and violations of strict exogeneity. To
probe the latter, I implement an equal-length
future-event placebo check.  Future earthquakes do not predict the preceding
change for the core index, very proud, jobs priority, family importance, or
generalized trust; willingness to fight is the exception, with a marginal
placebo coefficient that warrants caution when interpreting that outcome.

\section{Main Results}\label{sec:results}

\subsection{Baseline Results}\label{sec:results_baseline}

Table~\ref{tab:main_results} reports the baseline estimates from Equation~\eqref{eq:baseline}. Column~(1) shows that the coefficient for the national-core index is $-0.049$ (se~$=0.018$, standardized $\hat\beta=-0.098$, $p=0.007$). Columns~(2)--(4) unpack this composite into its constituent margins, showing that the core-index result is not driven by a single item. All three components yield negative coefficients: being very proud is marginally significant ($-0.040$, se~$=0.024$, $\hat\beta=-0.052$), willingness to fight is significant at the 5 percent level ($-0.047$, se~$=0.021$, $\hat\beta=-0.066$), and jobs priority is $-0.053$ (se~$=0.023$, $\hat\beta=-0.075$, $p=0.023$).\footnote{Because the core index averages conceptually distinct margins, jobs priority and the expressive components are reported and interpreted as co-primary outcomes in their own right; the composite result replicates across strict, PCA, component-pair, and item-specific index constructions (Table~\ref{tab:core_validation}). Section~\ref{sec:data} discusses the index's formative nature and why internal-consistency statistics such as Cronbach's $\alpha$ do not apply.} As theoretically anticipated, being very proud and willingness to fight provide separate evidence on affective intensity and costly national commitment,\footnote{The modest unconditional estimate for national pride anticipates the mechanism tests in Section~\ref{sec:mech_evidence}: it masks a large negative coefficient where state--religion alignment supplies symbolic infrastructure---robust to country-level inference and equal country weighting---and an estimate indistinguishable from zero where it does not, precisely the conditionality the model predicts for the expressive margin.} while jobs priority captures the practical allocation margin.

The sign and magnitude of these estimates are important. Since the exposure variable is distance from high-risk seismic zones, the negative coefficients indicate that national in-group orientation is stronger closer to seismic risk. To gauge magnitude, a one-standard-deviation increase in distance from high-risk zones (about 630 km) is associated with a $0.098$ standard deviation lower national-core index and a $0.075$ standard deviation lower jobs-priority response---roughly a 3-percentage-point lower probability of prioritizing nationals for scarce jobs.

Furthermore, the cross-sectional association is concentrated in earlier survey waves. Interacting distance with a continuous survey-year trend yields a positive and significant coefficient, indicating that the negative distance effect attenuates over time (Table~\ref{tab:temporal_heterogeneity}). As detailed in the Appendix, this attenuation is not an artifact of changing WVS country coverage: it replicates within a balanced set of 34 countries and persists when country-specific distance slopes are included.\footnote{The attenuation is not solely generated by turnover in WVS country coverage: it also persists among low-mobility respondents---those aged 50 and over and agricultural workers, least likely to have migrated---so it reflects within-population change rather than selective out-migration or urbanization, with the distance-by-year interaction remaining positive and, among those aged 50 and over, significant for the core index and jobs priority. These specifications do not, however, hold the composition of respondents or sampled regions perfectly fixed over time.} While adjudicating its exact drivers is beyond the scope of the data, the pattern is consistent with the association weakening as countries urbanize, develop, and secularize, or as the salience of the specific WVS nationalism items shifts across cohorts.

\subsection{The Realized-Event Response}\label{sec:results_event}

The baseline design identifies how chronic seismic geography maps into attitudes.
The stacked design in equation~\eqref{eq:eventstudy} asks the narrower
question of whether mean attitudes differ at the later survey endpoint when a
realized earthquake occurs between two fieldwork periods.
Table~\ref{tab:event_study_usgs} reports the pooled
estimates for the national-core index and its three components.  No outcome
rises: the national-core index ($-0.007$), very proud ($-0.030$), and jobs
priority ($+0.039$) are statistically indistinguishable from zero under
wild-cluster-bootstrap inference, and the religiosity benchmark in column~(5)
is equally unmoved ($-0.001$, $p=0.94$).  The one marginal coefficient is
willingness to fight, which falls by $0.044$ ($p_{\mathrm{wcb}}=0.053$);
because an equal-length future-event placebo for this outcome is similar in
magnitude and marginal as well, I do not read it as a causal
post-event effect.  The table therefore provides no evidence that a realized
earthquake strengthens national identity or religiosity by the next observed
survey endpoint.  Its
confidence intervals, however, do not rule out modest positive responses.\footnote{Table~\ref{tab:event_study_dynamics} adds the timing profile: among transitions with survey gaps of at most six years, jobs priority rises in the first two years after an earthquake, fades by two to four years, and is essentially zero thereafter, while the expressive coefficients are never positive at any horizon and willingness to fight is, if anything, negative early on. None is a precisely estimated post-event effect.}

These estimates are consistent in sign with an early but transitory
distributive response to realized scarcity, though none survives as a precisely
estimated post-event effect. The pooled null, however, masks a sharp and theoretically informative
heterogeneity. Table~\ref{tab:place_attachment} re-estimates
equation~\eqref{eq:eventstudy} interacting the treatment with respondent age.
Within the structurally high-hazard districts where qualifying earthquakes
occur, the between-wave response is concentrated among older, place-attached
residents: relative to younger respondents in the same transition, those aged
fifty and above shift toward the national-core index by $0.038$ and willingness
to fight by $0.045$ (equivalently $0.014$ and $0.017$ per decade of age), with a
parallel rise in jobs priority ($0.040$). Being very proud moves in the same
direction but more weakly. Because qualifying earthquakes occur almost exclusively in
high-hazard districts, this pattern is essentially unchanged in the full sample
(Panel~B); low-hazard areas contain too few treated transitions to identify a
separate response.

This heterogeneity reconciles the two designs. A realized earthquake moves
national identity not in the average respondent but in precisely those for whom
the latent, inescapable risk it makes salient is binding---older residents
rooted in place, who cannot exit the hazard---while the mobile young, who can,
do not respond. The event matters as a revelation of a permanent structural
condition rather than as a transient shock, which is why it leaves the
population mean near zero even as the cross-sectional association is large.
Section~\ref{sec:mechanisms} develops this reading, showing that identity tracks
the latent structural hazard rather than realized seismicity, and that
religiosity---which exhibits the same cross-sectional association---displays
neither an average nor an age-differential response to realized events.

\subsection{Identification and Robustness}\label{sec:results_robust}

This subsection collects the robustness and diagnostic evidence for both
designs, the cross-section first and the realized-event design second.  The
cross-sectional checks span outcome construction, specifications and controls,
hazard specificity, selection and influential observations, sample structure and
weighting, and spatial inference; the results are robust throughout and
consistent with the main identification assumption.

On outcome construction, Table~\ref{tab:alternative_indices} reports robustness to alternative
measures---the Lan--Li-style, economic/exclusionary, and country-over-world
constructions of Section~\ref{sec:data}: the Lan--Li measures are negative, the
narrower auxiliary measures less precise.  These are not additional primary
indices; the main analysis uses the national-core index and its three
components.

Table~\ref{tab:robustness_specs} examines specification
robustness across all Table~\ref{tab:main_results} outcomes (one per panel,
seven specifications): the baseline~(1); excluding high-risk-zone districts~(2);
and adding night-light intensity, a proxy for urbanization and economic
activity~(3); individual controls for generalized trust, unemployment, and
agricultural occupation~(4); a district-geography
battery~(5),\footnote{The battery comprises population density, arable share,
precipitation mean and variability, district area, terrain ruggedness,
elevation, and an indicator for positive disaster exposure.} earthquake-event
lags $t-2$ to $t-10$~(6); and distance to the national capital~(7).  Every coefficient is signed as predicted and all but
one are significant---the very-proud component when high-risk zones are
excluded.\footnote{Applying
multiple-testing corrections
jointly to the core index and its three components, the core index,
willingness to fight, and jobs priority survive the
Benjamini--Hochberg correction at the 5\% level; very proud survives at the
10\% level.  Under the more conservative Bonferroni correction, only the core
index survives at 5\%.  The four outcomes overlap by construction, so this is a
conservative family-wise diagnostic rather than four independent hypotheses
(see Table~\ref{tab:mht}).}  Columns~(3) and~(5), estimated on
essentially the full sample, reproduce the baseline core-index standardized
effect almost exactly ($-0.099$ and $-0.085$, against $-0.098$), so the result
does not depend on the baseline control set.  Column~(5) is especially
informative for confounding: terrain ruggedness is the only covariate imbalanced
across high- and low-risk regions after conditioning on country-year fixed
effects, latitude, and coastal distance (Table~\ref{tab:covariate_balance}, adjusted $p=0.005$ versus $p>0.17$ for the
rest), yet adding the full geography battery---ruggedness included---leaves the
national-core coefficient negative and significant, so this single imbalance does
not drive the main result.\footnote{The estimates are likewise insensitive to dropping WVS
survey weights (Table~\ref{tab:unweighted}) and stable under
the \citet{bentzen2019} specification checks, including a quadratic distance
term (Table~\ref{tab:bentzen_specs}).}

Two tables assess hazard specificity.  Table~\ref{tab:robustness_hazards} substitutes alternative seismic-risk measures
(each of the four outcomes in its own panel): relative to the baseline~(1),
distance to zone~4 alone, log distance, and fault-line distance~(2)--(4) are
negative throughout, with log distance significant for all four,\footnote{Two
cells are exceptions to statistical significance: the zone-4 estimate for jobs
priority and the fault-line estimate for very proud.  The mean earthquake-zone
indicator is small and insignificant for all outcomes.} while placebo distances
to tsunami zones and storm tracks are small and insignificant.  Table~\ref{tab:alternative_disasters} repeats this with non-seismic and
composite earthquake--tsunami measures and reaches the same verdict: earthquake
distance is the only exposure that consistently predicts national orientation;
tsunami, storm, and composite measures are null; and in both tables volcanic
proximity is significant only for willingness to fight.  This specificity is
consistent with the identification assumption: confounding by unobserved
geography correlated with seismic risk would produce similar associations for
geographically comparable hazards; it does not.

On selection and influential observations, absorbing ethnic-group composition
and restricting to lower-mobility respondents---older respondents and
agricultural workers---leaves the national-core coefficient essentially
unchanged (Table~\ref{tab:spatial_sorting}), so observable
composition does not point to contemporary self-selection as the main driver; a
leave-one-country-out exercise yields uniformly negative coefficients with no
reversals (Figure~\ref{fig:loo}).\footnote{\citet{oster2019} bounds further quantify robustness to selection on unobservables, including the historical sorting these sample restrictions cannot rule out directly. Relative to the rich control set (the geography battery and individual covariates), the degree of selection on unobservables relative to observables that would drive the coefficient to zero ranges from $\delta=5.4$ (jobs priority) to $\delta=28.7$ (willingness to fight)---far above the $\delta=1$ benchmark---and the bias-adjusted coefficient ($\delta=1$, $R_{\max}=1.3\,\tilde{R}$) remains negative for every baseline outcome.}

The estimation-sample structure is examined directly (Table~\ref{tab:aggregation_weighting}).  Collapsing to the roughly 1{,}200
region--country-year cells where the treatment varies reproduces the baseline
almost exactly, and restricting to the 136 country-year cells with identifying
variation changes nothing, so individual-level sample size plays no role.
Normalizing every country's total weight to one---so that no large-sample
country dominates---preserves the core index and willingness to fight, while
jobs priority and the religiosity benchmark weaken noticeably;
the composite is again the robust object, with component magnitudes partly
reflecting WVS country composition.\footnote{Collapsed to varying cells, the
core index is $-0.047$ against $-0.049$ at baseline; under equal country
weighting it is $-0.042$ (region-clustered $p=0.038$).}

Because the treatment is a spatially smooth distance observed at the 494 region
centroids, I also apply spatial-inference checks beyond region clustering.
Conley spatial-HAC standard errors (in brackets in
Table~\ref{tab:main_results}) are essentially identical to the region-clustered
ones at a 500~km Bartlett cutoff and at most eight percent larger at 1{,}000~km,
leaving every conclusion unchanged.\footnote{Two harsher diagnostics are collected in Tables~\ref{tab:spatial_inference} and~\ref{tab:kelly_noise}. Country-level clustering (63 clusters) raises standard errors by 25--45 percent, leaving the full-sample core index marginal ($p=0.065$; $0.092$ under a wild cluster bootstrap) and the component and \citet{bentzen2019} religiosity coefficients at comparable levels---so the religiosity and national coefficients face similar inference limits. Two patterns help interpret, but do not eliminate, this uncertainty. First, it is driven by recent waves: in the early period (1981--2005) the association is strong and survives country-level wild-bootstrap inference (Table~\ref{tab:spatial_inference}, Panel~B). Second, the precision loss does not occur for the expressive margin where the mechanism predicts a response: under high state--religion alignment the very-proud coefficient survives even country-level wild-bootstrap inference ($p_{\mathrm{wcb}}=0.046$; Section~\ref{sec:mech_evidence} and Table~\ref{tab:pride_regime}).}  Finally, a \citet{kelly2020} spatial-noise placebo replaces the treatment with
1{,}000 Gaussian random fields calibrated to its spatial covariance: it shows
that region-clustered inference over-rejects against spatially correlated noise
(9--17 percent at the nominal five percent), and provides a noise benchmark that
the core index clears at every calibration (noise-based $p=0.023$--$0.034$),
while the components and religiosity benchmark range $0.05$--$0.18$ (Online
Appendix Table~\ref{tab:kelly_noise}).

The remaining checks concern the realized-event design.  Its pooled conclusions
are stable across estimator choices---omitting individual controls, restricting
to transitions with at least 20 respondents per endpoint, or dropping WVS
weights leaves Table~\ref{tab:event_study_usgs} essentially unchanged (Online
Appendix Table~\ref{tab:event_study}).  The null is not an artifact of pride
coding: an indicator for being either very or quite proud shows, if anything, a
modest between-wave decline ($-0.027$, $p_{\mathrm{wcb}}=0.036$; Table~\ref{tab:event_study_pride_threshold})---no evidence of a positive
broad-pride response, though not proof that exposure itself causes a decline.  An
equal-length future-event placebo---earthquakes striking immediately after the
later fieldwork period---predicts no outcome except a marginal
willingness-to-fight coefficient ($p_{\mathrm{wcb}}=0.057$)---the same margin
whose between-wave coefficient is itself only marginal, from which I draw no
conclusion---while family importance and generalized trust, as outcome-side
placebos, show no response (Table~\ref{tab:event_study_placebos}).

\section{Mechanisms}\label{sec:mechanisms}

\subsection{A Chronic-Covariate-Risk Account}
\label{sec:mech_explanation}

Why should people who live closer to major seismic zones be more
nationalist---both in distributive and in expressive terms?  Two features of
the result discipline any answer.  First, as shown in Sections~\ref{sec:results_baseline}
and~\ref{sec:results_event}, the association is attached to the
place more clearly than to individual events: attitudes track the
modeled structural hazard---the deep-tectonic propensity encoded in the GSHAP
zones---whereas the realized-event estimates do not reveal a robust positive
between-wave post-event response.  What
needs explaining is therefore a chronic orientation: a standing hazard
apprehended through building codes, zoning, insurance, and the cultural
salience of inhabiting a dangerous place, absorbed into identity slowly
through socialization and accumulated experience.  An orientation of this kind
can behave as a stock rather than a discrete reaction to one or more events
during a survey interval.  This permanence is also why a realization adds
little: the binding hazard is the standing risk itself, which a household cannot
diversify away except by leaving, so what is internalized is the structural
propensity---and most by those rooted in place.  The account thus carries a
dynamic implication, tested in Section~\ref{sec:mech_evidence}: a realized
earthquake should move these immobile residents, not the mobile young who can
exit.

Crucially, this response works through salience under limited attention rather
than Bayesian updating: a local earthquake does not raise the objective
probability of an already-known chronic hazard so much as its attention weight.
In the model this is an increase in chronic-risk salience $S$ for the exposed,
and the comparative static in $S$ (Proposition~1) is positive---a higher
\emph{level} of national priority when the risk becomes salient. The static
framework therefore rationalizes a between-wave level shift without any temporal
dynamics, and the realized-event design is best read as a comparative-static
test in $S$, not as evidence of belief-updating about the hazard.

Second, national identity is not the only
attachment this geography predicts.  In the same sample, the same exposure
is also associated with higher religiosity, replicating \citet{bentzen2019}
(see column~(5) in Table~\ref{tab:main_results}); and
realized earthquakes leave religiosity just as unmoved as national attitudes
in the event design (see column~(5) in Table~\ref{tab:event_study_usgs}).
Thus, both identities display a long-run association without a detectable positive
event response, a pattern consistent with slowly accumulated orientations
attached to the same chronic hazard.  An account of the national result should
therefore also explain why these two large-scale identities are jointly
associated with chronic seismic exposure---and why comparable patterns do not
appear for other attachments.

My answer starts from the kind of risk earthquakes pose and the local
institutions that normally absorb risk.  Earthquake risk is highly
covariate: a major event would damage the assets, homes, and support
networks of many nearby households at once.  The institutions that buffer
idiosyncratic losses are not the ones that buffer shared losses
\citep{townsend1994, morduch1995}; in many settings, kinship and neighborhood
networks are important risk-sharing institutions that insure a
relative's illness or a household-specific income loss \citep{fafchamps2003}.
The hypothesis is that where risk is highly covariate, this local margin becomes
a less reliable hedge because the households that would normally help are
precisely those likely to be hit.  Larger communities of protection,
obligation, and meaning may therefore become relatively more salient---including the religious community and the
nation.
Appendix~\ref{app:theory} formalizes this as a social-interaction model in
which individuals allocate effort across local insurance, distributive
national priority, expressive national attachment, and religious coping, and
the equilibrium allocation depends on the covariance of risk and the
fragmentation of the religious field.  The logic immediately yields the
account's first and sharpest prediction: the response should be
selective.  The family and local margin should not strengthen---it is
the margin whose insurance capacity is most constrained---and out-group
boundaries need not harden generally.  This differs from a diffuse
threat-and-insecurity account under which attachment or exclusion would rise
across a broader set of social domains.  It is also the risk-based counterpart
to the family-versus-nation substitution documented by
\citet{alesinagiuliano2010}: strong local ties substitute for broader national
engagement, and covariate risk is what tips that substitution toward the
nation rather than away from it.

How, then, does the nation absorb the demand that local networks cannot meet?
It combines two roles that generate different predictions.  The first is
practical allocation: where jobs, aid, or public resources are scarce,
national membership defines who has priority \citep{miller1995}.  This
requires only a recognized boundary between members and non-members---not a
shared symbolic vocabulary (see Proposition~1 in the model)---which is why jobs
priority is the clearest direct measure of the practical margin in
Table~\ref{tab:main_results}: it can operate wherever scarcity meets a national
boundary without requiring the symbolic conditions developed below.  The 2023
Turkey--Syria earthquake is illustrative: shortages of shelter and aid made
the line between Turkish nationals and Syrian refugees politically salient as
an allocation rule \citep{syriadirect2023quake}.

Is all of nationalism distributive, then?  The second channel suggests not.
National pride and willingness to fight require the hazard to be narrated as
shared suffering, duty, and collective resilience
\citep{sabhlok2010, buchenau2009}---a symbolic operation closer to religious
identification, since both turn physical danger into a threat to a morally
meaningful collective ``we'' \citep{durkheim1912, anderson1983}.  After the
2011 T\={o}hoku earthquake, for instance, the \textit{Kizuna} (``bonds'')
campaign framed recovery as national resilience \citep{kantei2011kizuna}.
This expressive channel is institutionally conditional: it activates only
where a symbolic vocabulary is available to do the narrating, which is why
expressive outcomes move in the expected direction but vary more across
settings than jobs priority (see Propositions~2--3).

What supplies that vocabulary?  I propose that religious institutions do:
they hold repertoires of duty, sacrifice, and protection that can be attached
to the nation---not by making religion and nation identical, but by lending
symbolic infrastructure.  I hypothesize that this borrowing is easier where
state and religion are institutionally aligned
\citep{barromccleary2005, rubin2017}, so that religious and national idioms
already point to the same community, and where the religious field is less
fragmented, so that the vocabulary is more widely shared \citep{casanova1994}
(see the comparative statics of Propositions~2--3).
Ethnic diversity, by contrast, marks ancestry and language lines that need not
disrupt this shared religious vocabulary \citep{bellah1967}. 

The account therefore leaves three testable predictions, in increasing order
of specificity.  Selectivity: national and religious outcomes should display
stronger seismic-risk associations than family and general out-group
outcomes.  Asymmetric amplification: state--religion alignment bears
only on the symbolic coupling, so it should strengthen the association for the
expressive components more than for jobs priority.  Asymmetric
attenuation: religious fractionalization weakens the shared vocabulary and
so should attenuate the response most strongly for the expressive
components---equivalently, the expressive response to seismic risk should be
at full strength where the religious field is unified---while ethnic
fractionalization should show no comparable pattern, a
benchmark against a generic-diversity reading.  A complementary descriptive
diagnostic---conditioning on religiosity, which the coupling between
expressive and religious effort in the model makes natural even though it has
no formal comparative-static counterpart---is examined alongside the main
tests.  The next subsection takes these up in order.

\subsection{Testing the Two-Channel Asymmetry}
\label{sec:mech_evidence}

The first prediction concerns selectivity: is the association specific to
national and religious identity, or does seismic risk strengthen any close tie
and harden any group boundary?  Table~\ref{tab:scalable_institutions} shows
that family importance and family trust are small and insignificant, as are
attitudes toward ethnic diversity, rejection of different-language and
different-race neighbors, and generalized distrust.  Confidence in churches is
also null, so the religious response in Table~\ref{tab:main_results} appears to
reflect individual piety rather than institutional trust.  Taken together, the
two tables show national outcomes that are negative and mostly exclude zero,
religiosity between them, and family and out-group outcomes centered near
zero.  This ordering weighs against a diffuse insecurity response that
strengthens close attachments and hardens group boundaries generally.\footnote{A
broader battery of out-group and
cosmopolitan outcomes is also null, although several auxiliary items have
substantially smaller samples (Table~\ref{tab:outgroup}).}

The second prediction is asymmetric amplification: symbolic integration should
matter more for expressive national attachment than for practical national
priority.  Table~\ref{tab:expressive_alignment} interacts seismic-risk distance
with a standardized country-year index of state--religion alignment.  Alignment
significantly strengthens the distance association for the national-core index
and very proud ($p=0.036$ and $p=0.030$), marginally strengthens it for
willingness to fight ($p=0.068$), and does not significantly moderate jobs
priority ($p=0.325$).  Moving from one standard deviation below to one above
mean alignment changes the standardized very-proud coefficient from $-0.001$
to $-0.236$, compared with a change from $-0.057$ to $-0.143$ for jobs
priority.\footnote{The moderator is the RAS official-support variable (SBX,
0--13) from \citet{arda_ras3}, standardized to mean zero and unit variance; the
2014 value is carried forward for later WVS waves (Section~\ref{sec:data}).}

Split-sample estimates provide a complementary representation of this
conditionality.  The very-proud coefficient is negative in high-alignment
country-years, survives country-level wild-cluster-bootstrap inference and
equal-country weighting, and is close to zero under low alignment.  A
three-point pride scale produces the same qualitative pattern.  I treat these
splits as corroboration of the continuous interaction rather than as separate
tests: they support the prediction that expressive attachment is strongest
where religious and national idioms are institutionally connected, but they
do not imply that expressive attachment responds only in those
settings.\footnote{Table~\ref{tab:pride_regime} reports a very-proud coefficient of $-0.179$ above mean alignment ($p_{\mathrm{wcb}}=0.046$), against essentially zero below it, and robust to equal-country weighting. Table~\ref{tab:alignment_horserace} shows that the alignment interaction survives a competing interaction with country mean religiosity, although allowing separate distance slopes across nine macro-regions reduces its precision for very proud.}

The third prediction concerns fragmentation of the symbolic environment.
Table~\ref{tab:religious_fractionalization} interacts distance with religious
(Panel~A) and ethnic (Panel~B) fractionalization \citep{alesina2003}.
Religious fractionalization significantly attenuates the national-core index
($p=0.003$), very proud ($p=0.005$), and willingness to fight ($p=0.038$).
For very proud, the standardized coefficient moves from $-0.210$ at one
standard deviation below mean fractionalization to $-0.023$ at one standard
deviation above it.  The jobs-priority interaction is smaller and imprecise
($p=0.143$), while ethnic fractionalization shows no comparable pattern.
Thus, the relevant friction appears to be fragmentation of the religious
vocabulary that can be attached to the nation, rather than diversity in
general.\footnote{The result is unchanged when ethnic fractionalization is
measured with the time-varying Historical Index of Ethnic Fractionalization
\citep{drazanova2020}; none of its four interactions is significant (Online
Appendix Table~\ref{tab:ethnic_fractionalization_heterogeneity}).  Because the
two \citet{alesina2003} indices correlate at $r=0.80$, I also enter their
interactions jointly.  The religious interactions remain significant and the
ethnic interactions do not (Table~\ref{tab:religious_fractionalization},
Panel~C).}

The complementary descriptive diagnostic asks whether conditioning on religiosity
changes the seismic-risk coefficients.  Adding individual religiosity
attenuates the coefficient by 16.4\% for very proud, 8.9\% for willingness to
fight, and 8.5\% for the national-core index, compared with 2.9\% for jobs
priority (Table~\ref{tab:religious_coping_attenuation}).  This ordering is consistent
with symbolic coupling between religious and expressive effort, but it is not
a causal decomposition: religiosity itself responds to seismic risk and is
therefore a bad control for mediation.  Its limited implication is that the
jobs-priority association shares little of the variation captured by the
religiosity index.

The exposure logic of Section~\ref{sec:mech_explanation} leaves a dynamic
signature as well. As Section~\ref{sec:results_event} showed, a realized
earthquake moves the between-wave response among the immobile---older,
place-attached residents, not the mobile young
(Table~\ref{tab:place_attachment})---and does so on the expressive and core
margins rather than the religiosity benchmark. This places the dynamic evidence
where the account predicts: on those bound to the risk, and on the national
rather than the religious-coping margin.

The evidence therefore has a coherent hierarchy: national and religious
outcomes respond while the available family and boundary diagnostics do not;
alignment especially amplifies expressive attachment; and religious, but not
ethnic, fragmentation attenuates that response.  Multiple-testing adjustment
leaves the strongest evidence on the expressive components and the
national-core index rather than on jobs priority. These patterns are consistent with the proposed mechanism but
do not exclude every alternative; the next subsection considers three of them.

\subsection{Alternative Explanations}\label{sec:alternatives}

Three alternatives could generate similar patterns, and each carries a prediction
the proposed mechanism does not share.  The first---a diffuse insecurity or
ethnic-threat response that raises all in-group preferences at once---is already
answered by the specificity evidence above: the association is concentrated in the
national domain and is attenuated by religious, not ethnic, fragmentation,
neither of which a uniform threat response would predict.

A second, demand-side reading holds that earthquake risk raises demand for
state-provided goods, with national membership serving as the eligibility
criterion; if so, the response should be stronger where the state is more
capable.  Table~\ref{tab:state_capacity_heterogeneity} interacts seismic-risk
distance with a standardized state-capacity composite (log GDP per capita, V-Dem
rule of law, public-sector cleanliness, and Polity democracy).\footnote{The
composite averages standardized log GDP per capita \citep{worldbank2026wdi},
V-Dem rule of law and public-sector cleanliness \citep{vdem2026}, and Polity
democracy \citep{marshallpolity2018}; all components are predetermined relative
to WVS responses.}  The composite does not significantly moderate any outcome.\footnote{Harder measures of state authority and fiscal capacity amplify the jobs-priority coefficient more clearly than the core index, the expressive components, or religiosity (Table~\ref{tab:country_heterogeneity_blocks})---a scope condition for the \emph{distributive} interpretation (a credible state raises the value of membership as an allocation claim) rather than a force on the expressive margin, which these measures do not amplify. Among eight institutional-confidence items, only confidence in the armed forces has a significant distance coefficient, an exploratory pattern consistent with collective security rather than broad confidence in administrative delivery; it---like the headline gradients for every baseline outcome---is essentially unchanged after controlling for distance to the nearest international land border, so the results do not reflect proximity to a militarized border (Tables~\ref{tab:institutional_confidence} and~\ref{tab:armed_forces_border}).}

The third alternative attributes the effect to top-down authoritarian
mobilization, which should make it larger in low-democracy regimes where leaders
exploit crises through nationalist appeals.  A standardized
democracy--accountability composite does not significantly moderate the association,
and the estimated conditional effects remain similar across low- and
high-democracy environments (Table~\ref{tab:country_heterogeneity_selected}).\footnote{The index averages standardized V-Dem electoral and liberal democracy and judicial and legislative constraints \citep{vdem2026} with Polity democracy \citep{marshallpolity2018}; V-Dem electoral democracy is also reported separately. Two features reinforce this: country-year fixed effects absorb all country-level institutional variation, including regime type, so any residual concern would require mobilization targeted specifically at seismically exposed regions; and the expressive results track state--religion alignment rather than democratic quality, pointing to symbolic integration rather than political repression.}  

Residual concerns remain---local labor-market scarcity near active zones could
independently raise jobs priority, and selective sorting of nationally oriented
individuals into exposed regions cannot be fully excluded---but the
heterogeneity evidence helps delimit these alternatives.  Accountability,
civil-society strength, subnational political scale, unemployment, and several
other moderators do not systematically interact with exposure.  The
individual results are not uniformly null, however: low education,
agricultural work, and especially low confidence in state institutions are
associated with weaker jobs-priority effects.  These exploratory interactions
may reflect post-treatment characteristics and are therefore not causal tests.
The most systematic country-level moderators remain religious
fractionalization, state--religion alignment, and selected hard state-capacity
measures (Tables~\ref{tab:country_heterogeneity_selected},
\ref{tab:country_heterogeneity_blocks},
and~\ref{tab:micro_heterogeneity}).\footnote{A related alternative is the ``welfare chauvinism'' hypothesis of Section~\ref{sec:data}, which cannot be tested directly because region-level immigrant exposure is unavailable in the WVS. Own immigrant status, observed only from 2010, illustrates the limit: in that later-wave sample the jobs-priority coefficient is already near zero, and restricting to native-born respondents or controlling for immigrant status changes little (Table~\ref{tab:jobs_native})---so own status does not explain within-subsample differences, but this cannot establish full-sample robustness or separate allocation under scarcity from anti-immigrant sentiment. The contrast with out-group diagnostics is nonetheless informative: the seismic gradient appears for jobs priority (Table~\ref{tab:main_results}) but not for immigration-policy restrictiveness ($\hat\beta_{\mathrm{std}}=-0.027$, n.s.) or rejection of immigrant neighbors ($-0.005$, n.s.; Table~\ref{tab:outgroup}, columns~(7) and~(4)), consistent with allocation under scarcity rather than diffuse anti-immigrant sentiment. The item alone cannot establish this, but controlling directly for both attitudes leaves the jobs gradient essentially unchanged ($-0.102\to-0.099$ on that sample; Table~\ref{tab:jobs_distributive_validity}), indicating that its seismic component is largely orthogonal to general anti-immigrant sentiment.}  The alternatives remain possible,
but each requires a separate assumption to explain why the response is national
rather than generalized, why it is moderated by religious rather than ethnic
fragmentation, and why the expressive margin tracks state--religion alignment.
The proposed mechanism organizes all three with a single distinction: the nation
can work immediately as a practical membership rule, while its transformation
into an affective moral community depends on symbolic infrastructure.

\section{Conclusion}\label{sec:conclusion}

The evidence in this paper supports a specific account of how collective
environmental threat connects to national identity.  Earthquake risk is
associated with stronger national in-group orientation, but not through a
uniform process.  The distributive margin---treating national membership as a
practical allocation rule for scarce resources---appears in the baseline and
across many specification checks, although its precision weakens under equal
country weighting and country-level inference.  The
expressive margin---turning material destruction into duty, pride, and
collective sacrifice---is conditionally activated, requiring symbolic
infrastructure that links religious and national institutions.  These
margins are related and conceptually distinct, though formal tests do not sharply
separate them---their average magnitudes are statistically indistinguishable, and
the distinction rests mainly on their differential institutional moderation.  That
distinction is the paper's central claim.

The baseline pattern documents this separation.  Proximity to seismic risk
predicts the national-core index and all three components---very proud,
willingness to fight, and jobs priority---while family attachment, out-group
diagnostics, and generalized trust are near zero (family nulls partly reflect
ceiling effects in the underlying items) and religiosity rises in parallel.
The mechanism evidence locates the conditionality: state--religion alignment
amplifies the association for expressive outcomes more clearly than for jobs
priority, and religious fractionalization attenuates the response more
systematically than ethnic fractionalization.

The alternative explanations remain possible but fit the full pattern less
directly.  A diffuse-insecurity account predicts stronger effects on ethnic-threat
and out-group diagnostics; these are small and insignificant.  A state-delivery
account predicts larger responses where the state is more capable; aggregate
state-capacity interactions are null, though harder fiscal measures amplify
jobs priority, consistent with the distributive channel.  An
authoritarian-mobilization account predicts stronger effects under low democracy;
democracy interactions are insignificant. 

The design identifies a long-run within-country-year association between
seismic-risk geography and attitudes.  The diagnostics that carry the identification burden
are demanding: the exposure measure predicts realized shaking, the relationship
is specific to seismic hazard among geographically comparable hazards, and it
survives the geography battery, ethnic-composition controls, and
leave-one-out exercises.  The stacked realized-event design
is complementary: it finds no robust shift in the average respondent---for
religiosity as for nationalism---but the response concentrates among the
immobile, older residents who cannot leave the hazard a realized earthquake
makes salient. The event thus operates as a revelation of permanent risk rather
than a transient shock, reconciling the two designs.

The broader implication is that collective environmental threat can be
associated with nationalism in more than one way: it can make the nation
salient as an allocation boundary, and it can make the nation a moral
community---an expressive response that depends on symbolic infrastructure.
Where the nation-building literature explains national identity as something
states supply, the evidence here points to a demand-side, geographically
rooted origin that operates from below---and that the symbolic infrastructure
of nation-building can amplify.
The connection between religion and nationalism documented
here is not one of substitution: religion conditions the expressive content of
national attachment---turning material destruction into a shared moral
experience---while the distributive margin operates largely independently of
the religious-symbolic field.

%
%

\newpage
\bibliographystyle{apalike}
\bibliography{references}


\clearpage
\renewcommand{\thefigure}{\Roman{figure}}
\setcounter{figure}{0}
\section*{Main Figures}


\begin{figure}[H]
\centering
\captionsetup{font={normalsize,bf}}
\caption{WVS Region Coverage and Seismic Geography}\label{fig:map}
\vspace{0.3cm}
\includegraphics[width=0.95\textwidth]{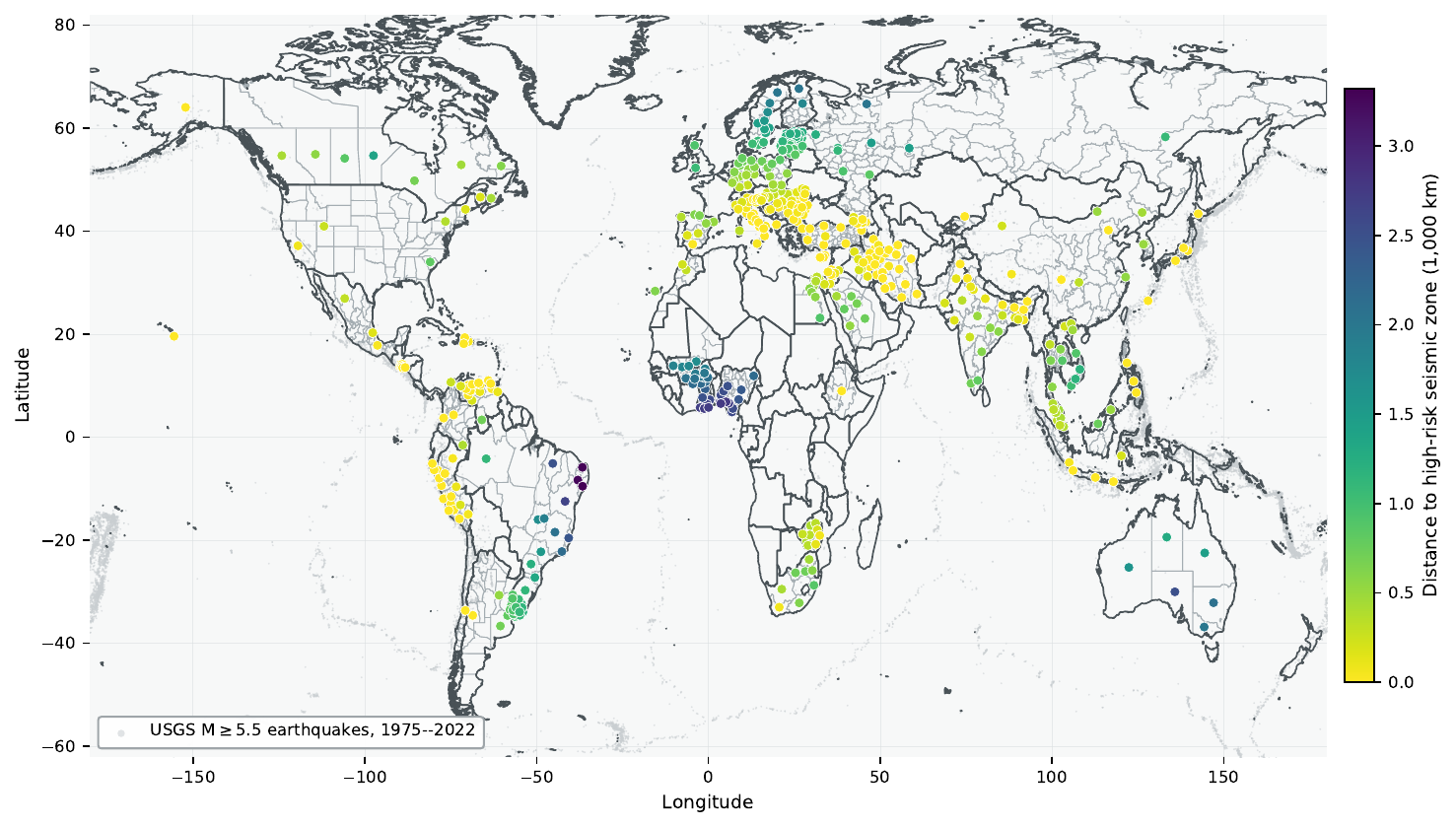}
\vspace{0.1cm}
\begin{minipage}{0.92\textwidth}
\scriptsize\textbf{Notes:} Each marker is one of the 497 matched WVS/Bentzen subnational-region centroids, colored by the treatment variable: geodesic distance from the region's administrative boundary to the nearest high-risk seismic zone (GSHAP earthquake-intensity zones~3--4), in 1,000~km; darker colors indicate regions farther from seismic risk. Thick lines mark national borders and thin lines mark first-order administrative divisions in sample countries. The light-grey background points are the epicenters of all 22,797 USGS magnitude-$\geq$5.5 earthquakes, 1975--2022, which trace the world's plate boundaries and provide the seismic-geography reference. Identification compares regions within the same country and survey year.
\end{minipage}
\end{figure}


\clearpage
\renewcommand{\thetable}{\Roman{table}}
\setcounter{table}{0}
\section*{Main Tables}

\begin{table}[H]
{
\renewcommand{\arraystretch}{0.78}
\setlength{\tabcolsep}{1.5pt}
\captionsetup{font={normalsize,bf}}
\caption{Sample Coverage and Variable Availability}\label{tab:coverage}
\vspace{-0.5cm}
\begin{center}
\begin{adjustbox}{max width=0.98\linewidth, max totalheight=0.86\textheight, keepaspectratio}%
\begin{tabular}{L{6.5cm}*{7}{C{1.50cm}}}
\toprule
 & (1) & (2) & (3) & (4) & (5) & (6) & (7) \\
Variable & Obs. & Ctys. & Ctry. yrs & First & Last & Mean & SD \\
\cmidrule(lr){1-1}\cmidrule(lr){2-2}\cmidrule(lr){3-3}\cmidrule(lr){4-4}\cmidrule(lr){5-5}\cmidrule(lr){6-6}\cmidrule(lr){7-7}\cmidrule(lr){8-8}
\multicolumn{8}{l}{\textit{Panel A. National identity and national-orientation measures}} \\
National core index &    164,312 &   63 &  172 & 1981 & 2022 &    0.669 &    0.318 \\
Very proud of nationality &    165,785 &   63 &  173 & 1981 & 2022 &    0.624 &    0.484 \\
Very or quite proud &    165,785 &   63 &  173 & 1981 & 2022 &    0.902 &    0.298 \\
Willing to fight for country &    145,382 &   62 &  167 & 1981 & 2022 &    0.694 &    0.461 \\
Jobs priority for nationals &    157,406 &   61 &  165 & 1989 & 2022 &    0.705 &    0.456 \\
Economic/exclusionary index &    160,352 &   61 &  165 & 1990 & 2022 &    0.502 &    0.256 \\
\hline
\multicolumn{8}{l}{\textit{Panel B. Family attachment and religiosity}} \\
Family important &    167,258 &   62 &  170 & 1989 & 2022 &    3.894 &    0.356 \\
Trust in family &    102,196 &   49 &   99 & 2005 & 2022 &    3.784 &    0.505 \\
Religiosity index &    169,636 &   62 &  171 & 1981 & 2022 &    0.695 &    0.305 \\
Religious person &    163,374 &   62 &  168 & 1981 & 2022 &    0.709 &    0.454 \\
Comfort from religion &     56,655 &   46 &   67 & 1981 & 2004 &    0.757 &    0.429 \\
\hline
\multicolumn{8}{l}{\textit{Panel C. Institutional confidence and out-group diagnostics}} \\
State confidence index &    164,513 &   61 &  170 & 1981 & 2022 &    0.486 &    0.241 \\
Confidence: churches &    166,510 &   62 &  172 & 1981 & 2022 &    0.643 &    0.340 \\
Confidence: United Nations &    135,918 &   62 &  158 & 1995 & 2022 &    0.488 &    0.316 \\
Ethnic diversity erodes unity &     29,315 &   33 &   33 & 2005 & 2009 &    0.403 &    0.315 \\
Reject immigrant neighbors &    159,337 &   61 &  163 & 1981 & 2022 &    0.233 &    0.422 \\
General distrust &    163,948 &   63 &  175 & 1981 & 2022 &    0.752 &    0.432 \\
\hline
\multicolumn{8}{l}{\textit{Panel D. Geographic exposure and district controls}} \\
Distance to EQ zones 3--4 &    172,518 &   63 &  176 & 1981 & 2022 &    0.498 &    0.629 \\
Distance to EQ zone 4 &    172,518 &   63 &  176 & 1981 & 2022 &    1.413 &    1.866 \\
Distance to fault line &     95,276 &   61 &  100 & 1981 & 2008 &    0.581 &    0.623 \\
Distance to tsunami zones &    172,518 &   63 &  176 & 1981 & 2022 &    0.500 &    0.494 \\
Distance to storm tracks &    172,518 &   63 &  176 & 1981 & 2022 &    1.656 &    1.458 \\
Night-light intensity ($\times 10^{6}$) &    171,618 &   63 &  176 & 1981 & 2022 &    7.795 &   32.979 \\
Population density &    171,618 &   63 &  176 & 1981 & 2022 &    0.694 &    2.081 \\
Terrain ruggedness &    170,221 &   63 &  174 & 1981 & 2022 &  1.3e05 &  1.2e05 \\
\hline
\multicolumn{8}{l}{\textit{Panel E. Country-level moderators}} \\
Log GDP per capita &    172,468 &   63 &  175 & 1981 & 2022 &    8.760 &    1.239 \\
V-Dem electoral democracy &    172,251 &   63 &  175 & 1981 & 2022 &    0.612 &    0.246 \\
Religious fractionalization &    171,662 &   62 &  173 & 1981 & 2022 &    0.457 &    0.248 \\
Ethnic fractionalization &    171,662 &   62 &  173 & 1981 & 2022 &    0.421 &    0.235 \\
WGI governance index &    145,972 &   61 &  148 & 1996 & 2022 &    0.534 &    0.162 \\
RAS official religion support &    169,047 &   63 &  171 & 1990 & 2022 &    6.990 &    2.825 \\
RAS any official religion &    169,047 &   63 &  171 & 1990 & 2022 &    0.170 &    0.375 \\
\bottomrule

\addlinespace[0.15cm]
\multicolumn{8}{p{17.5cm}}{\scriptsize{\textbf{Notes:} This table reports coverage statistics for selected variables in the matched WVS--geographic estimation sample. Column~(1) reports non-missing observations; columns~(2) and~(3) report the numbers of countries and country-year cells; columns~(4) and~(5) report the first and last survey years; and columns~(6) and~(7) report the mean and standard deviation. The national-core index averages very proud, willingness to fight, and jobs priority and requires at least two observed components. Geographic variables are assigned to respondents by region; country moderators are assigned by country-year when time-varying and by country otherwise.}} \\
\end{tabular}
\end{adjustbox}
\end{center}
}
\end{table}


\begin{table}[H]
{
\renewcommand{\arraystretch}{0.85}
\setlength{\tabcolsep}{2.5pt}
\captionsetup{font={normalsize,bf}}
\caption{Main National In-Group Orientation Results}\label{tab:main_results}
\vspace{-0.45cm}
\begin{center}
\small
\begin{adjustbox}{max width=0.98\linewidth, max totalheight=0.95\textheight, keepaspectratio}%
\begin{tabular}{L{4.0cm}*{5}{C{2.4cm}}}
\toprule
 & (1) & (2) & (3) & (4) & (5) \\
 & \shortstack[c]{Nat. core\\index} & \shortstack[c]{Very proud\\of nation} & \shortstack[c]{Willing to\\fight} & \shortstack[c]{Jobs priority\\for nationals} & \shortstack[c]{Religiosity\\(benchmark)} \\ \midrule
Dist. to EQ zone & -0.0488*** & -0.0398* & -0.0468** & -0.0534** & -0.0344** \\
 & (0.0180) & (0.0238) & (0.0210) & (0.0234) & (0.0140) \\
Conley SE (500 km) & [0.0181] & [0.0239] & [0.0218] & [0.0230] & [0.0139] \\
Conley SE (1{,}000 km) & [0.0194] & [0.0262] & [0.0231] & [0.0242] & [0.0138] \\
Std. $\beta$ & -0.098 & -0.052 & -0.066 & -0.075 & -0.071 \\ \midrule
Obs. &    162,643 &    164,037 &    143,835 &    156,662 &    167,847 \\
Adj. $R^2$ & 0.230 & 0.186 & 0.166 & 0.150 & 0.365 \\ \bottomrule

\addlinespace[0.15cm]
\multicolumn{6}{p{16.5cm}}{\scriptsize{\textbf{Notes:} This table presents the main estimated effects of proximity to high-risk seismic zones on national in-group orientation. All columns report estimates of equation~\eqref{eq:baseline} and include country-year fixed effects and controls for age, age$^2$, sex, marital status, absolute latitude, and distance to the coast. Column~(1) reports the national-core index; columns~(2)--(4) report being very proud of one's nationality, willingness to fight for the country, and jobs priority for nationals. The core index averages the three components and requires at least two non-missing items. Column~(5) reports the Bentzen-style religiosity index as a replication benchmark \citep{bentzen2019}. Distance to high-risk earthquake zones~3--4 is measured in 1,000~km. Standard errors in parentheses are clustered at the region level; Conley spatial-HAC standard errors (Bartlett kernel, 500~km and 1{,}000~km cutoffs, computed from region-centroid distances) are reported in brackets and are very close to the clustered ones. *, **, and *** indicate statistical significance at the 10\%, 5\%, and 1\% levels, respectively, based on region-clustered standard errors.}} \\
\end{tabular}
\end{adjustbox}
\end{center}
}
\end{table}


\begin{table}[H]
\centering
{
\renewcommand{\arraystretch}{0.8}
\setlength{\tabcolsep}{3.2pt}
\captionsetup{font={normalsize,bf}}
\caption{Realized Earthquakes Between WVS Fieldwork Periods}\label{tab:event_study_usgs}
\vspace{-0.5cm}
\begin{center}
\small
\begin{adjustbox}{max width=0.98\linewidth, max totalheight=0.85\textheight, keepaspectratio}%
\begin{tabular}{L{6cm}*{5}{C{2.0cm}}}
\toprule
 & (1) & (2) & (3) & (4) & (5) \\
  & National core & Very proud & Fight country & Jobs priority & Religiosity \\ \midrule
Post $\times$ earthquake & -0.0065 & -0.0301 & -0.0435* &  0.0391 & -0.0011 \\
  & (0.0175) & (0.0278) & (0.0215) & (0.0255) & (0.0141) \\
Wild-bootstrap $p$ & 0.728 & 0.320 & 0.053 & 0.190 & 0.941 \\ \midrule
Transition fixed effects & Yes & Yes & Yes & Yes & Yes \\
Country-wave-pair $\times$ endpoint FE & Yes & Yes & Yes & Yes & Yes \\
Individual controls & Yes & Yes & Yes & Yes & Yes \\
Observations &   184,033 &   182,774 &   162,159 &   177,541 &   191,247 \\
Country clusters & 43 & 43 & 43 & 43 & 43 \\ \bottomrule

\addlinespace[0.15cm]
\multicolumn{6}{p{17cm}}{\scriptsize{\textbf{Notes:} This table presents stacked repeated-cross-section estimates of equation~\eqref{eq:eventstudy}. Each transition stacks independent respondent samples from the earlier and later survey endpoints; it does not follow the same individuals over time. A transition is treated when a USGS magnitude-$\geq$5.5 earthquake occurred within 150~km of the district centroid strictly between the two WVS fieldwork periods. Survey-year gaps range from two to eleven years (median six), so the pooled coefficient measures a between-wave post-event shift rather than an immediate response. Columns~(1)--(4) report the national-core index and its three components: being very proud, willingness to fight, and jobs priority. Column~(5) reports the Bentzen-style religiosity index as a benchmark; its estimated post-event shift is no larger than those for the national outcomes. All specifications include transition fixed effects, country-wave-pair by endpoint fixed effects, age, age squared, sex, and marital status. Regressions use WVS weights; standard errors are clustered by country. Reported $p$-values use a wild cluster bootstrap with Rademacher weights, the null imposed, and 9{,}999 replications. *, **, and *** indicate significance according to those $p$-values at the 10\%, 5\%, and 1\% levels.}} \\
\end{tabular}
\end{adjustbox}
\end{center}
}
\end{table}



\begin{table}[H]
{
\renewcommand{\arraystretch}{0.85}
\setlength{\tabcolsep}{3.5pt}
\captionsetup{font={normalsize,bf}}
\caption{Alternative Group Attachments and Out-Group
Diagnostics}\label{tab:scalable_institutions}
\vspace{-0.45cm}
\begin{center}
\small
\begin{adjustbox}{max width=1\linewidth, max totalheight=0.84\textheight, keepaspectratio}%
\begin{tabular}{L{3.0cm}*{7}{C{2cm}}}
\toprule
 & (1) & (2) & (3) & (4) & (5) & (6) & (7) \\
 & \shortstack[c]{Family\\important} & \shortstack[c]{Trust:\\family} & \shortstack[c]{Conf.:\\churches} & \shortstack[c]{Ethnic div.\\threat} & \shortstack[c]{Reject diff.-\\lang.\ nbr.} & \shortstack[c]{Reject race\\nbr.} & \shortstack[c]{General\\distrust} \\ \midrule
Dist. to EQ zone & 0.0001 & -0.0015 & -0.0051 & -0.0221 & -0.0058 & -0.0110 & 0.0171 \\
 & (0.0110) & (0.0222) & (0.0127) & (0.0311) & (0.0181) & (0.0180) & (0.0169) \\
Std. $\beta$ &  0.000 & -0.002 & -0.009 & -0.052 & -0.011 & -0.019 &  0.025 \\ \midrule
Obs. &    166,485 &    101,621 &    164,716 &     29,115 &     98,859 &    160,153 &    162,222 \\
Adj. $R^2$ & 0.052 & 0.060 & 0.307 & 0.181 & 0.094 & 0.137 & 0.110 \\ \bottomrule

\addlinespace[0.15cm]
\multicolumn{8}{p{18.7cm}}{\scriptsize{\textbf{Notes:} This table presents estimated effects of proximity to high-risk seismic zones on alternative group-attachment and out-group diagnostic outcomes. All columns report estimates of equation~\eqref{eq:baseline} and include country-year fixed effects and controls for age, age$^2$, sex, marital status, absolute latitude, and distance to the coast. Columns~(1)--(3) report the family and religious-institution diagnostics: family importance, trust in one's family, and confidence in churches; the religiosity index itself is reported with the main results (Table~\ref{tab:main_results}, column~(5)). Columns~(4)--(7) report the ethnic-boundary and trust diagnostics: perceived ethnic-diversity threat, rejection of neighbors of a different language, rejection of neighbors of a different race, and generalized distrust. All seven coefficients are statistically indistinguishable from zero. The pattern is consistent with a selective response concentrated in national and religious dimensions rather than a generalized tightening of group boundaries. Distance to high-risk earthquake zones~3--4 is measured in 1,000~km. Standard errors are clustered at the region level. *, **, and *** indicate statistical significance at the 10\%, 5\%, and 1\% levels, respectively.}} \\
\end{tabular}
\end{adjustbox}
\end{center}
}
\end{table}


\begin{table}[H]
{
\renewcommand{\arraystretch}{0.85}
\setlength{\tabcolsep}{4pt}
\captionsetup{font={normalsize,bf}}
\caption{State-Religion Alignment and Expressive National
Responses}\label{tab:expressive_alignment}
\vspace{-0.45cm}
\begin{center}
\small
\begin{adjustbox}{max width=0.98\linewidth, max totalheight=0.82\textheight, keepaspectratio}%
\begin{tabular}{L{4.5cm}*{4}{C{1.9cm}}}
\toprule
 & (1) & (2) & (3) & (4) \\
 & Nat. core idx. & Very proud & Fight country & Jobs priority \\
\cmidrule(lr){2-2}\cmidrule(lr){3-3}\cmidrule(lr){4-4}\cmidrule(lr){5-5}
dist $\times$ z-RAS support & -0.0619** & -0.0892** & -0.0595* & -0.0306 \\
 & (0.0289) & (0.0401) & (0.0320) & (0.0308) \\
$\beta$ (low alignment) & -0.043 & -0.001 & -0.027 & -0.057 \\
$\beta$ (high alignment) & -0.293 & -0.236 & -0.196 & -0.143 \\ \midrule
Obs. &    160,976 &    161,753 &    142,097 &    156,351 \\ \bottomrule

\addlinespace[0.15cm]
\multicolumn{5}{p{13.4cm}}{\scriptsize{\textbf{Notes:} This table presents heterogeneity estimates for equation~\eqref{eq:baseline} augmented with an interaction between distance to earthquake risk and the standardized RAS measure of official state-religion alignment. All specifications include country-year fixed effects and the baseline controls. Column~(1) uses the national-core index; columns~(2)--(4) use being very proud, willingness to fight, and jobs priority. $\beta$~(low/high alignment) reports the standardized distance effect one~SD below and above the moderator mean. Alignment amplifies the core index and very proud; the willingness-to-fight interaction is similar in sign and marginally significant, while the jobs-priority interaction is smaller and insignificant. The sample is restricted to WVS waves from 1990 onward, when RAS data are available. Standard errors are clustered at the country level. *, **, and *** indicate statistical significance at the 10\%, 5\%, and 1\% levels, respectively.}} \\
\end{tabular}
\end{adjustbox}
\end{center}
}
\end{table}


\begin{table}[H]
{
\renewcommand{\arraystretch}{0.85}
\setlength{\tabcolsep}{4pt}
\captionsetup{font={normalsize,bf}}
\caption{Heterogeneity by Religious and Ethnic
Fractionalization}\label{tab:religious_fractionalization}
\vspace{-0.45cm}
\begin{center}
\small
\begin{adjustbox}{max width=0.98\linewidth, max totalheight=0.82\textheight, keepaspectratio}%
\begin{tabular}{L{4.5cm}*{4}{C{1.9cm}}}
\toprule
 & (1) & (2) & (3) & (4) \\
 & Nat. core idx. & Very proud & Fight country & Jobs priority \\ \midrule
\multicolumn{5}{l}{\textit{Panel A: Alesina religious fractionalization}} \\
dist $\times$ z-relig. frac. & 0.0593*** & 0.0713*** & 0.0505** & 0.0418 \\
 & (0.0194) & (0.0242) & (0.0238) & (0.0281) \\
$\beta$ (low frac.) & -0.304 & -0.210 & -0.190 & -0.179 \\
$\beta$ (high frac.) & -0.066 & -0.023 & -0.048 & -0.061 \\
Obs. &    161,850 &    163,322 &    143,183 &    155,835 \\
Adj. $R^2$ & 0.230 & 0.187 & 0.167 & 0.149 \\
\hline
\multicolumn{5}{l}{\textit{Panel B: Alesina ethnic fractionalization}} \\
dist $\times$ z-eth. frac. & 0.0267* & 0.0320 & 0.0335 & 0.0139 \\
 & (0.0158) & (0.0239) & (0.0220) & (0.0194) \\
$\beta$ (low frac.) & -0.161 & -0.102 & -0.120 & -0.099 \\
$\beta$ (high frac.) & -0.054 & -0.018 & -0.026 & -0.059 \\
Obs. &    161,850 &    163,322 &    143,183 &    155,835 \\
Adj. $R^2$ & 0.230 & 0.186 & 0.167 & 0.148 \\
\hline
\multicolumn{5}{l}{\textit{Panel C: both measures entered jointly}} \\
dist $\times$ z-relig. frac. & 0.0565*** & 0.0680*** & 0.0456* & 0.0409 \\
 & (0.0209) & (0.0250) & (0.0231) & (0.0304) \\
dist $\times$ z-eth. frac. & 0.0123 & 0.0163 & 0.0228 & 0.0038 \\
 & (0.0188) & (0.0230) & (0.0221) & (0.0254) \\
Obs. &    161,850 &    163,322 &    143,183 &    155,835 \\
Adj. $R^2$ & 0.230 & 0.187 & 0.167 & 0.149 \\ \bottomrule

\addlinespace[0.15cm]
\multicolumn{5}{p{13.4cm}}{\scriptsize{\textbf{Notes:} This table presents heterogeneity estimates for equation~\eqref{eq:baseline} augmented with interactions between distance to earthquake risk and country-level fractionalization measures. All specifications include country-year fixed effects and the baseline controls. Panel~A uses the Alesina et~al.\ (2003) religious fractionalization index; Panel~B uses the corresponding ethnic fractionalization index; Panel~C enters both interactions jointly, a horse-race motivated by the $r=0.80$ correlation between the two indices. Column~(1) uses the national-core index; columns~(2)--(4) use being very proud, willingness to fight, and jobs priority. $\beta$~(low/high frac.) reports the standardized distance effect one~SD below and above the mean. Religious fractionalization significantly attenuates the core index and the two expressive components, while its interaction with jobs priority is smaller and imprecise. Ethnic fractionalization has no comparable 5-percent-significant moderating effect, and in the joint specification of Panel~C the religious interaction remains significant while the ethnic interaction does not. Both fractionalization measures are time-invariant (circa 2000) and standardized. Standard errors are clustered at the country level. *, **, and *** indicate statistical significance at the 10\%, 5\%, and 1\% levels, respectively.}} \\
\end{tabular}
\end{adjustbox}
\end{center}
}
\end{table}


\begin{table}[h!]
\centering
{
\renewcommand{\arraystretch}{0.8}
\setlength{\tabcolsep}{3.2pt}
\captionsetup{font={normalsize,bf}}
\caption{Realized Earthquakes and Place Attachment}\label{tab:place_attachment}
\vspace{-0.5cm}
\begin{center}
\small
\begin{adjustbox}{max width=0.98\linewidth, max totalheight=0.85\textheight, keepaspectratio}%
\begin{tabular}{L{6cm}*{5}{C{2.0cm}}}
\toprule
 & (1) & (2) & (3) & (4) & (5) \\
 & National core & Very proud & Fight country & Jobs priority & Religiosity \\ \midrule
\multicolumn{6}{l}{\textit{Panel A. High-hazard districts (distance to zone $\leq$ median)}} \\
\quad Event $\times$ 1[Age $\geq$ 50] &  0.0382*** &  0.0248 &  0.0445** &  0.0400*** &  0.0015 \\
  & (0.0134) & (0.0155) & (0.0203) & (0.0111) & (0.0070) \\
\quad Event $\times$ Age (per decade) &  0.0137*** &  0.0087* &  0.0170*** &  0.0132*** &  0.0020 \\
  & (0.0042) & (0.0047) & (0.0064) & (0.0035) & (0.0025) \\
\quad Observations &    91,007 &    91,962 &    77,854 &    85,719 &    96,812 \\
\quad Region clusters & 174 & 175 & 174 & 170 & 175 \\ \midrule
\multicolumn{6}{l}{\textit{Panel B. Full sample}} \\
\quad Event $\times$ 1[Age $\geq$ 50] &  0.0415** &  0.0251 &  0.0508** &  0.0388*** & -0.0123 \\
  & (0.0171) & (0.0180) & (0.0249) & (0.0136) & (0.0078) \\
\quad Event $\times$ Age (per decade) &  0.0142*** &  0.0097* &  0.0178** &  0.0117*** & -0.0020 \\
  & (0.0053) & (0.0054) & (0.0082) & (0.0041) & (0.0028) \\
\quad Observations &   184,033 &   182,774 &   162,159 &   177,541 &   191,247 \\
\quad Region clusters & 311 & 312 & 311 & 307 & 312 \\ \midrule
Transition fixed effects & Yes & Yes & Yes & Yes & Yes \\
Country-wave-pair $\times$ endpoint FE & Yes & Yes & Yes & Yes & Yes \\
Individual controls & Yes & Yes & Yes & Yes & Yes \\ \bottomrule

\addlinespace[0.15cm]
\multicolumn{6}{p{17cm}}{\scriptsize{\textbf{Notes:} This table re-estimates equation~\eqref{eq:eventstudy}, interacting the treatment $\text{Post}_{s}\times\text{EQ}_{r\tau}$ with respondent age, operationalized as an indicator for age~$\geq$~50 and, alternatively, as continuous age in decades. The reported coefficient is the \emph{within-district} age differential in the between-wave response---how much more a realized earthquake moves older than younger respondents in the same transition. Panel~A restricts the sample to structurally high-hazard districts (distance to a high-risk zone at or below the median), where qualifying earthquakes occur: 119 of the 120 treated transitions fall in this group. Panel~B uses the full sample; the two panels nearly coincide because low-hazard districts contain essentially no events and enter only as controls. Columns~(1)--(4) report the national-core index and its three components; column~(5) reports the Bentzen-style religiosity benchmark. Because the estimand nets out the common district-level shock, standard errors (in parentheses) are clustered at the region (district) level, the appropriate level for this within-district interaction; under the more conservative country-level clustering the jobs-priority interaction remains significant, while the national-core and fight interactions are significant only at the region level, reflecting the small number of treated countries; I therefore read the pattern as suggestive corroboration rather than independent proof. All specifications retain the transition and country-wave-pair-by-endpoint fixed effects and individual controls of equation~\eqref{eq:eventstudy} and use WVS weights. *, **, and *** indicate significance at the 10\%, 5\%, and 1\% levels.}} \\
\end{tabular}
\end{adjustbox}
\end{center}
}
\end{table}


\clearpage
\appendix
\section*{Online Appendix for the Paper ``Environmental Threat and the Nation: Earthquake Risk, Distributive Priority, and Expressive Attachment''}
\begin{center}
\normalsize Hector Galindo-Silva\\
\small Department of Economics, Pontificia Universidad Javeriana
\end{center}
\addcontentsline{toc}{section}{Online Appendix for the Paper ``Environmental Threat and the Nation''}

\noindent This online appendix contains the supplementary evidence and formal
model for the paper.  Section~\ref{app:tables_figures} reports tables and
figures in the order in which they are first cited in the main text.  The
material covers measurement and identification diagnostics, robustness of the
cross-sectional and realized-event designs, and additional mechanism and
heterogeneity exercises.  Section~\ref{app:theory} presents the model that
organizes the distinction between distributive national priority and expressive
national attachment and states the limits of its empirical interpretation.




\section{Additional Tables and Figures}
\label{app:tables_figures}
\renewcommand{\thetable}{A\arabic{table}}
\setcounter{table}{0}
\renewcommand{\thefigure}{A\arabic{figure}}
\setcounter{figure}{0}


\begin{table}[H]
{
\renewcommand{\arraystretch}{0.8}
\setlength{\tabcolsep}{2pt}
\captionsetup{font={normalsize,bf}}
\caption{Covariate Balance by Earthquake-Risk Exposure}\label{tab:covariate_balance}
\vspace{-0.5cm}
\begin{center}
\small
\begin{adjustbox}{max width=0.98\linewidth, max totalheight=0.85\textheight, keepaspectratio}%
\begin{tabular}{L{5.2cm}*{5}{C{2.0cm}}}
\toprule
 & (1) & (2) & (3) & (4) & (5) \\
Covariate & Obs. & High-risk mean & Low-risk mean & Adj. diff. & Adj. p \\
\midrule
Age &    164,018 &  40.319 &  42.007 &  -0.304 & 0.468 \\
Male &    163,258 &   0.496 &   0.492 &   0.002 & 0.761 \\
Generalized trust &    156,438 &   0.214 &   0.283 &   0.007 & 0.592 \\
Unemployed &    158,533 &   0.096 &   0.114 &  -0.019 & 0.218 \\
Agricultural worker &     79,596 &   0.094 &   0.098 &   0.007 & 0.678 \\
Night lights ($\times 10^{6}$) &    163,412 &   9.920 &   5.022 &  0.6756 & 0.683 \\
Population density &    163,412 &   0.552 &   0.722 &  -0.083 & 0.658 \\
Elevation &    162,779 & 706.753 & 401.668 &   6.139 & 0.941 \\
Terrain ruggedness &    162,724 &   179,164 &    80,731 &    33,669 & 0.005 \\
Arable land &    164,312 &   638,936 &   739,904 &   -91,495 & 0.252 \\
Temperature &    163,412 &  16.780 &  15.746 &   0.274 & 0.502 \\
Ethnic minority &    129,203 &   0.263 &   0.206 &   0.092 & 0.201 \\
Dist. to capital (1,000 km) &     85,628 &   0.744 &   2.444 &  -0.143 & 0.174 \\
\bottomrule

\addlinespace[0.15cm]
\multicolumn{6}{p{15.8cm}}{\scriptsize{\textbf{Notes:} This table reports covariate balance by earthquake-risk exposure. High-risk regions are those below the weighted median distance to high-risk earthquake zones in the national-core estimation sample. Column~(1) reports the number of non-missing observations; columns~(2) and~(3) report weighted means in high- and low-risk regions; column~(4) reports the adjusted high-minus-low difference; and column~(5) reports its $p$-value. Adjusted differences come from regressions of each covariate on the high-risk indicator, controlling for country-year fixed effects, absolute latitude, and distance to the coast, with standard errors clustered by subnational region.}} \\
\end{tabular}
\end{adjustbox}
\end{center}
}
\end{table}


\begin{table}[H]
{
\renewcommand{\arraystretch}{0.9}
\setlength{\tabcolsep}{3.0pt}
\captionsetup{font={normalsize,bf}}
\caption{Realized Earthquake Exposure Validation}\label{tab:treatment_validation}
\vspace{-0.45cm}
\begin{center}
\small
\begin{adjustbox}{max width=1\linewidth, max totalheight=0.85\textheight, keepaspectratio}%
\begin{tabular}{L{4cm}*{5}{C{2.5cm}}}
\toprule
 & (1) & (2) & (3) & (4) & (5) \\
 & \shortstack[c]{Any damaging\\quake} & \shortstack[c]{No. damaging\\quake events} & \shortstack[c]{ln(1+allocated\\deaths)} & \shortstack[c]{ln(1+allocated\\affected)} & \shortstack[c]{ln(1+allocated\\damages)} \\ \midrule
Dist. to EQ zone & -0.1133** & -0.2793** & -0.3902 & -1.2001** & -1.0172 \\
  & (0.0501) & (0.1378) & (0.2394) & (0.5298) & (0.6258) \\
Std. $\beta$ & -0.201 & -0.162 & -0.140 & -0.220 & -0.164 \\
Mean dep. var. &  0.174 &  0.393 &  0.628 &  1.534 &  1.534 \\ \midrule
Obs. &      1,115 &      1,115 &      1,115 &      1,115 &      1,115 \\
Country-year FE & Yes & Yes & Yes & Yes & Yes \\
Geographic controls & Yes & Yes & Yes & Yes & Yes \\
SE cluster & Region & Region & Region & Region & Region \\
\bottomrule

\addlinespace[0.15cm]
\multicolumn{6}{p{17.5cm}}{\scriptsize{\textbf{Notes:} This table validates the baseline distance-to-earthquake-risk measure against realized damaging-earthquake exposure at the region--survey-year level. Damaging earthquake locations from GDIS are matched to EM-DAT country-event severity measures; GDIS locations are assigned to the nearest WVS/Bentzen region in the same country within 150~km, deaths, affected population, and damages are allocated equally across the event's GDIS locations, and exposure is aggregated over the 20 years before the survey year. Column~(1) reports an indicator for any damaging earthquake; column~(2) the number of damaging earthquake events; columns~(3)--(5) the logs of one plus allocated deaths, affected population, and damages. All specifications include country-year fixed effects and controls for absolute latitude and distance to the coast. Standard errors are clustered by subnational region. Distance to high-risk earthquake zones~3--4 is measured in 1,000~km, so negative coefficients indicate greater realized exposure closer to high-risk seismic zones. *, **, and *** indicate statistical significance at the 10\%, 5\%, and 1\% levels, respectively.}} \\
\end{tabular}
\end{adjustbox}
\end{center}
}
\end{table}


\begin{table}[H]
{
\renewcommand{\arraystretch}{0.9}
\setlength{\tabcolsep}{3.2pt}
\captionsetup{font={normalsize,bf}}
\caption{Validation of the National-Core Index}\label{tab:core_validation}
\vspace{-0.5cm}
\begin{center}
\small
\begin{adjustbox}{max width=1\linewidth, max totalheight=0.85\textheight, keepaspectratio}%
\begin{tabular}{L{3.2cm}*{8}{C{2.0cm}}}
\toprule
 & (1) & (2) & (3) & (4) & (5) & (6) & (7) & (8) \\
 & Core ($\geq$2) & Core (strict) & Core (PCA) & Very proud + jobs & Fight + jobs & Very proud & Fight & Jobs priority \\ \midrule
Dist. to EQ zone & -0.0488*** & -0.0447** & -0.1621** & -0.0470** & -0.0501*** & -0.0398* & -0.0468** & -0.0534** \\
 & (0.0180) & (0.0183) & (0.0677) & (0.0198) & (0.0181) & (0.0238) & (0.0210) & (0.0234) \\
Std. $\beta$ & -0.098 & -0.097 & -0.094 & -0.087 & -0.096 & -0.052 & -0.066 & -0.075 \\ \midrule
Country-year FE & Yes & Yes & Yes & Yes & Yes & Yes & Yes & Yes \\
Controls & Yes & Yes & Yes & Yes & Yes & Yes & Yes & Yes \\
Obs. &    162,643 &    131,519 &    131,519 &    150,681 &    136,179 &    164,037 &    143,835 &    156,662 \\
Countries &  63 &  60 &  60 &  61 &  60 &  63 &  62 &  61 \\ \bottomrule

\addlinespace[0.15cm]
\multicolumn{9}{p{21cm}}{\scriptsize{\textbf{Notes:} This table validates the national-core index through alternative aggregation methods. All columns report estimates of equation~\eqref{eq:baseline} and include country-year fixed effects and controls for age, age$^2$, sex, marital status, absolute latitude, and distance to the coast. Each column reports the baseline distance coefficient for one index construction; see column headers for specific definitions. The primary composite result replicates across strict, PCA, component-pair, and item-specific definitions, consistent with robustness to the choice of aggregation method. The preferred index averages available components; the strict version requires all three. Cronbach's alpha for the three components is 0.335; the first principal component has positive loadings on all items. Distance to high-risk earthquake zones~3--4 is measured in 1,000~km. Standard errors are clustered at the region level. *, **, and *** indicate statistical significance at the 10\%, 5\%, and 1\% levels, respectively.}} \\
\end{tabular}
\end{adjustbox}
\end{center}
}
\end{table}


\begin{table}[H]
{
\renewcommand{\arraystretch}{0.85}
\setlength{\tabcolsep}{4pt}
\captionsetup{font={normalsize,bf}}
\caption{Temporal Heterogeneity of the Seismic-Risk Gradient}\label{tab:temporal_heterogeneity}
\vspace{-0.45cm}
\begin{center}
\small
\begin{adjustbox}{max width=0.98\linewidth, max totalheight=0.82\textheight, keepaspectratio}%
\begin{tabular}{L{6.2cm}*{5}{C{2.3cm}}}
\toprule
 & (1) & (2) & (3) & (4) & (5) \\
 & \shortstack[c]{Nat. core\\index} & \shortstack[c]{Very proud\\of nation} & \shortstack[c]{Willing to\\fight} & \shortstack[c]{Jobs priority\\for nationals} & \shortstack[c]{Religiosity\\(benchmark)} \\
\midrule
Distance $\times$ year, full sample & 0.0053*** & 0.0056** & 0.0025 & 0.0078*** & 0.0037*** \\
Distance $\times$ year, balanced panel & 0.0055*** & 0.0055* & 0.0034 & 0.0072*** & 0.0036*** \\
Distance $\times$ year, within-country slopes & 0.0036* & 0.0026 & 0.0016 & 0.0076*** & 0.0029** \\
\midrule
Countries (balanced / full) & 34 / 63 & 34 / 63 & 34 / 62 & 34 / 61 & 34 / 62 \\
Observations (full sample) &   162,643 &   164,037 &   143,835 &   156,662 &   167,847 \\
\bottomrule

\addlinespace[0.15cm]
\multicolumn{6}{p{19cm}}{\scriptsize{\textbf{Notes:} This table presents temporal change in the seismic-risk association. Each cell reports the distance-by-year interaction---the per-year change in the distance coefficient---from a single specification of equation~\eqref{eq:baseline} that interacts distance to high-risk zones with a continuous survey-year trend (centered at 2000); no period dummies or sample splits enter the estimating equation. A positive coefficient means the negative distance effect weakens toward zero over time. Column~(1) reports the national-core index; columns~(2)--(4) report being very proud of one's nationality, willingness to fight for the country, and jobs priority for nationals; column~(5) reports the Bentzen-style religiosity index. The first row uses the full sample; the second restricts to the balanced set of 34 countries surveyed in both an early ($\leq$2005) and a late ($\geq$2010) period, holding country coverage fixed; the third absorbs country-specific distance slopes, so the trend is identified only from within-country changes over time. The decline is statistically clear in columns~(1), (4), and~(5) and weaker for the expressive components in columns~(2)--(3) taken individually; for the core index and religiosity it reproduces in the balanced panel and persists with within-country slopes, indicating a within-country secular change rather than compositional turnover in WVS country coverage, and should be read as a scope condition on external validity over time. All models include country-year fixed effects, baseline controls, and WVS weights. Standard errors are clustered at the region level. *, **, and *** indicate statistical significance at the 10\%, 5\%, and 1\% levels, respectively.}} \\
\end{tabular}
\end{adjustbox}
\end{center}
}
\end{table}


\begin{table}[H]
{
\renewcommand{\arraystretch}{0.85}
\setlength{\tabcolsep}{3.2pt}
\captionsetup{font={normalsize,bf}}
\caption{Event-Time Dynamics of the Realized-Earthquake Response}\label{tab:event_study_dynamics}
\vspace{-0.45cm}
\begin{center}
\small
\begin{adjustbox}{max width=0.98\linewidth, max totalheight=0.82\textheight, keepaspectratio}%
\begin{tabular}{L{5.8cm}*{5}{C{1.8cm}}}
\toprule
 & (1) & (2) & (3) & (4) & (5) \\
Time since earthquake & National core & Very proud & Fight country & Jobs priority & Religiosity \\ \midrule
0--2 years &  0.0040 & -0.0326 & -0.0462** &  0.0681* &  0.0070 \\
  & (0.0214) & (0.0327) & (0.0236) & (0.0288) & (0.0176) \\
Wild-bootstrap $p$ & 0.866 & 0.362 & 0.043 & 0.065 & 0.708 \\
More than 2--4 years & -0.0155 & -0.0641 & -0.0477 &  0.0375 &  0.0077 \\
  & (0.0376) & (0.0738) & (0.0327) & (0.0367) & (0.0162) \\
Wild-bootstrap $p$ & 0.723 & 0.432 & 0.203 & 0.389 & 0.666 \\
More than 4--6 years & -0.0216 & -0.0395 & -0.0055 & -0.0056 &  0.0127 \\
  & (0.0126) & (0.0262) & (0.0340) & (0.0185) & (0.0156) \\
Wild-bootstrap $p$ & 0.273 & 0.254 & 0.913 & 0.780 & 0.427 \\
\midrule
Transition fixed effects & Yes & Yes & Yes & Yes & Yes \\
Country-wave-pair $\times$ endpoint FE & Yes & Yes & Yes & Yes & Yes \\
Observations &   125,171 &   121,834 &   111,533 &   121,514 &   127,758 \\ \bottomrule

\addlinespace[0.15cm]
\multicolumn{6}{p{16cm}}{\scriptsize{\textbf{Notes:} This table presents the event-time profile of the realized-earthquake response. All columns report stacked repeated-cross-section estimates of equation~\eqref{eq:eventstudy} in which the treatment interaction $\text{Post}_{s}\times\text{EQ}_{r\tau}$ is split into three indicators for the time elapsed between the most recent qualifying earthquake and the start of the later fieldwork period: 0--2 years, more than 2--4 years, and more than 4--6 years. Column~(1) reports the national-core index; columns~(2)--(4) report being very proud of one's nationality, willingness to fight for the country, and jobs priority for nationals; column~(5) reports the Bentzen-style religiosity index. The sample is restricted to adjacent WVS fieldwork periods no more than six years apart.  All specifications include transition fixed effects, country-wave-pair by endpoint fixed effects, individual controls (age, age$^2$, sex, marital status), and WVS weights. Standard errors, in parentheses, are clustered by country; reported $p$-values use a wild cluster bootstrap with Rademacher weights, the null imposed, and 9{,}999 replications. *, **, and *** indicate significance according to those $p$-values at the 10\%, 5\%, and 1\% levels.}} \\
\end{tabular}
\end{adjustbox}
\end{center}
}
\end{table}


\begin{table}[H]
{
\renewcommand{\arraystretch}{0.8}
\setlength{\tabcolsep}{3.2pt}
\captionsetup{font={normalsize,bf}}
\caption{Auxiliary National-Orientation Indices}\label{tab:alternative_indices}
\vspace{-0.5cm}
\begin{center}
\small
\begin{adjustbox}{max width=0.98\linewidth, max totalheight=0.85\textheight, keepaspectratio}%
\begin{tabular}{L{4.2cm}*{5}{C{2.0cm}}}
\toprule
 & (1) & (2) & (3) & (4) & (5) \\
 & Lan-Li strict & Lan-Li geq2 comp. & Lan-Li PCA & Econ. excl. idx. & Country vs. world \\ \midrule
Dist. to EQ zone & -0.0246* & -0.0320** & -0.1217* & -0.0206 & -0.0201 \\
 & (0.0149) & (0.0153) & (0.0721) & (0.0142) & (0.0189) \\
Std. $\beta$ & -0.067 & -0.080 & -0.069 & -0.052 & -0.046 \\ \midrule
Obs. &    139,712 &    167,446 &    139,712 &    159,584 &     67,977 \\
Adj. $R^2$ & 0.261 & 0.270 & 0.264 & 0.219 & 0.113 \\ \bottomrule

\addlinespace[0.15cm]
\multicolumn{6}{p{15.3cm}}{\scriptsize{\textbf{Notes:} This table reports auxiliary composite national-orientation measures. All columns estimate equation~\eqref{eq:baseline} with country-year fixed effects and the baseline controls. Columns~(1)--(3) are Lan--Li-style constructions; column~(4) is an economic/exclusionary index; and column~(5) is country-over-world identity. The Lan--Li strict index requires all three components (pride, defense, and armed-forces confidence); the $\geq 2$ version allows any two. Country-over-world identity is country-citizen minus world-citizen identity. These are validation measures, not additional primary outcomes. Distance to high-risk earthquake zones~3--4 is measured in 1,000~km. Standard errors are clustered at the region level. *, **, and *** indicate statistical significance at the 10\%, 5\%, and 1\% levels, respectively.}} \\
\end{tabular}
\end{adjustbox}
\end{center}
}
\end{table}


\begin{table}[H]
{
\renewcommand{\arraystretch}{0.7}
\setlength{\tabcolsep}{0.2pt}
\captionsetup{font={normalsize,bf}}
\caption{Robustness Specifications}\label{tab:robustness_specs}
\vspace{-0.5cm}
\begin{center}
\begin{adjustbox}{max width=0.98\linewidth, max totalheight=0.9\textheight, keepaspectratio}%
\begin{tabular}{L{3.1cm}*{7}{C{2.3cm}}}
\toprule
 & (1) & (2) & (3) & (4) & (5) & (6) & (7) \\
 & Baseline & Excl.\ risk & Night lights & Extra indiv. & Extra geog. & EQ lags & Dist.\ capital \\ \midrule
\multicolumn{8}{l}{\textit{Panel A. National core index}} \\
Dist.\ to EQ zone & -0.0488*** & -0.0427** & -0.0491*** & -0.0824*** & -0.0424** & -0.0851*** & -0.0802*** \\
  & (0.0180) & (0.0181) & (0.0180) & (0.0242) & (0.0164) & (0.0235) & (0.0232) \\
Std.\ $\beta$ & -0.098 & -0.089 & -0.099 & -0.155 & -0.085 & -0.166 & -0.157 \\
Obs. &    162,643 &    133,812 &    161,743 &     70,951 &    160,907 &     86,704 &     84,293 \\ \addlinespace[0.25cm]
\multicolumn{8}{l}{\textit{Panel B. Very proud of nationality}} \\
Dist.\ to EQ zone & -0.0398* & -0.0381 & -0.0404* & -0.0607* & -0.0373* & -0.0771** & -0.0649* \\
  & (0.0238) & (0.0246) & (0.0237) & (0.0347) & (0.0225) & (0.0358) & (0.0357) \\
Std.\ $\beta$ & -0.052 & -0.051 & -0.053 & -0.074 & -0.049 & -0.097 & -0.082 \\
Obs. &    164,037 &    135,070 &    163,138 &     72,448 &    161,624 &     89,745 &     85,043 \\ \addlinespace[0.25cm]
\multicolumn{8}{l}{\textit{Panel C. Willing to fight for country}} \\
Dist.\ to EQ zone & -0.0468** & -0.0401* & -0.0468** & -0.0537* & -0.0447** & -0.0515* & -0.0599** \\
  & (0.0210) & (0.0208) & (0.0210) & (0.0273) & (0.0184) & (0.0279) & (0.0290) \\
Std.\ $\beta$ & -0.066 & -0.058 & -0.066 & -0.072 & -0.063 & -0.072 & -0.083 \\
Obs. &    143,835 &    120,545 &    142,954 &     60,868 &    142,177 &     73,449 &     71,503 \\ \addlinespace[0.25cm]
\multicolumn{8}{l}{\textit{Panel D. Jobs priority for nationals}} \\
Dist.\ to EQ zone & -0.0534** & -0.0465* & -0.0537** & -0.1270*** & -0.0412* & -0.1204*** & -0.1153*** \\
  & (0.0234) & (0.0237) & (0.0234) & (0.0246) & (0.0219) & (0.0236) & (0.0229) \\
Std.\ $\beta$ & -0.075 & -0.066 & -0.076 & -0.166 & -0.058 & -0.165 & -0.159 \\
Obs. &    156,662 &    129,454 &    155,763 &     67,326 &    154,922 &     82,446 &     79,711 \\
\midrule
SE cluster & Region & Region & Region & Region & Region & Region & Region \\
Country-year FE & Yes & Yes & Yes & Yes & Yes & Yes & Yes \\ \bottomrule

\addlinespace[0.15cm]
\multicolumn{8}{p{19cm}}{\scriptsize{\textbf{Notes:} This table presents robustness checks for all Table~\ref{tab:main_results} outcomes across seven specifications. Each panel reports one outcome with the dependent variable indicated in the panel header. All specifications estimate equation~\eqref{eq:baseline} with country-year fixed effects and region-clustered standard errors. Column~(1) is the baseline; column~(2) excludes respondents in high-risk earthquake zones; column~(3) adds night-light intensity; column~(4) adds individual controls for generalized trust, unemployment, and agricultural occupation; column~(5) adds district geography (population density, arable share, precipitation mean and variability, district area, terrain ruggedness, elevation, and an indicator for positive disaster exposure); column~(6) adds earthquake-event lags from $t-2$ to $t-10$; column~(7) adds distance to the national capital as a control for state centrality. Standardized $\beta$ is the coefficient scaled by the ratio of the standard deviation of the exposure to the standard deviation of the outcome. Distance to high-risk earthquake zones~3--4 is measured in 1,000~km. *, **, and *** indicate statistical significance at the 10\%, 5\%, and 1\% levels, respectively.}} \\
\end{tabular}
\end{adjustbox}
\end{center}
}
\end{table}


\begin{table}[H]
{
\renewcommand{\arraystretch}{0.85}
\setlength{\tabcolsep}{1.5pt}
\captionsetup{font={normalsize,bf}}
\caption{Multiple Testing Corrections}\label{tab:mht}
\vspace{-0.45cm}
\begin{center}
\small
\begin{adjustbox}{max width=0.98\linewidth, max totalheight=0.9\textheight, keepaspectratio}%
\begin{tabular}{L{4.8cm}C{1.8cm}C{2.0cm}C{2.0cm}C{2.0cm}C{2.0cm}C{1.8cm}}
\toprule
 & (1) & (2) & (3) & (4) & (5) & (6) \\
Outcome & $b$ & (se) & $p$ unadj. & $p$ Bonf. & $q$ BH & $N$ \\
\midrule
\multicolumn{7}{l}{\textit{Panel A. National in-group outcomes (family size $k=4$)}} \\
\addlinespace[0.15cm]
National core index & $-$0.0488 & (0.0180) & 0.0071\textbf{***} & 0.0285\textbf{**} & 0.0285\textbf{**} &   162,643 \\
Very proud of nationality & $-$0.0398 & (0.0238) & 0.0955\textbf{*} & 0.3820 & 0.0955\textbf{*} &   164,037 \\
Willing to fight & $-$0.0468 & (0.0210) & 0.0265\textbf{**} & 0.1060 & 0.0353\textbf{**} &   143,835 \\
Jobs priority for nationals & $-$0.0534 & (0.0234) & 0.0231\textbf{**} & 0.0926\textbf{*} & 0.0353\textbf{**} &   156,662 \\
\addlinespace[0.2cm]\multicolumn{7}{l}{\textit{Panel B. Placebo outcomes (family size $k=3$)}} \\
\addlinespace[0.15cm]
Family importance & $+$0.0001 & (0.0110) & 0.9957 & 1.0000 & 0.9957 &   166,485 \\
Family trust & $-$0.0015 & (0.0222) & 0.9453 & 1.0000 & 0.9957 &   101,621 \\
Generalized trust & $-$0.0171 & (0.0169) & 0.3125 & 0.9375 & 0.9375 &   162,222 \\
\bottomrule

\addlinespace[0.15cm]
\multicolumn{7}{p{17cm}}{\scriptsize{\textbf{Notes:} This table reports multiple-testing corrections for the main national-identity and placebo regressions. Column~(1) reports the coefficient on distance to high-risk seismic zones; column~(2) its region-clustered standard error; column~(3) the unadjusted $p$-value; column~(4) the Bonferroni-adjusted $p$-value; column~(5) the Benjamini--Hochberg (BH) $q$-value; and column~(6) the number of observations. All specifications are identical to equation~\eqref{eq:baseline}. Panel~A defines the national family as the national-core index and its three components ($k=4$); Panel~B uses three placebo outcomes ($k=3$).  Corrections are computed within each displayed family \citep{benjaminihochberg1995}. *, **, and *** attached to columns~(3)--(5) indicate that the corresponding probability is below 10\%, 5\%, and 1\%, respectively.}} \\
\end{tabular}
\end{adjustbox}
\end{center}
}
\end{table}


\begin{table}[H]
{
\renewcommand{\arraystretch}{0.85}
\setlength{\tabcolsep}{4pt}
\captionsetup{font={normalsize,bf}}
\caption{Survey-Weighted vs.\ Unweighted Estimates}\label{tab:unweighted}
\vspace{-0.45cm}
\begin{center}
\small
\begin{adjustbox}{max width=0.80\linewidth, max totalheight=0.85\textheight, keepaspectratio}%
\begin{tabular}{L{4.2cm}C{2.0cm}C{1.8cm}C{2.0cm}C{1.8cm}}
\toprule
 & \multicolumn{2}{c}{Weighted} & \multicolumn{2}{c}{Unweighted} \\
 & (1) & (2) & (3) & (4) \\
 & Coef. & N & Coef. & N \\
\cmidrule(lr){2-3}\cmidrule(lr){4-5}
National core index & -0.0488*** & \multirow{2}{*}{  162,643} & -0.0480*** & \multirow{2}{*}{  162,643} \\
 & (0.0180) & & (0.0182) & \\
Very proud & -0.0398* & \multirow{2}{*}{  164,037} & -0.0400* & \multirow{2}{*}{  164,037} \\
 & (0.0238) & & (0.0230) & \\
Willing to fight & -0.0468** & \multirow{2}{*}{  143,835} & -0.0445** & \multirow{2}{*}{  143,835} \\
 & (0.0210) & & (0.0214) & \\
Jobs priority & -0.0534** & \multirow{2}{*}{  156,662} & -0.0541** & \multirow{2}{*}{  156,662} \\
 & (0.0234) & & (0.0231) & \\
\midrule
Country-year FE & \multicolumn{2}{c}{Yes} & \multicolumn{2}{c}{Yes} \\
Individual controls & \multicolumn{2}{c}{Yes} & \multicolumn{2}{c}{Yes} \\
SE cluster & \multicolumn{2}{c}{Region} & \multicolumn{2}{c}{Region} \\
\bottomrule

\addlinespace[0.15cm]
\multicolumn{5}{p{13.5cm}}{\scriptsize{\textbf{Notes:} This table compares the main results with and without WVS survey weights. Columns~(1)--(2) reproduce the baseline weighted estimates; columns~(3)--(4) use the same sample and specification but drop survey weights. All regressions include country-year fixed effects, individual controls (age, age$^2$, sex, marital status), and geographic controls (absolute latitude, distance to coast), with standard errors clustered by subnational region. *, **, *** indicate significance at the 10\%, 5\%, and 1\% levels, respectively.}} \\
\end{tabular}
\end{adjustbox}
\end{center}
}
\end{table}


\begin{table}[H]
{
\renewcommand{\arraystretch}{0.8}
\setlength{\tabcolsep}{3.2pt}
\captionsetup{font={normalsize,bf}}
\caption{Bentzen Specification Checks}\label{tab:bentzen_specs}
\vspace{-0.5cm}
\begin{center}
\small
\begin{adjustbox}{max width=0.98\linewidth, max totalheight=0.85\textheight, keepaspectratio}%
\begin{tabular}{L{4.2cm}*{4}{C{2.0cm}}}
\toprule
 & (1) & (2) & (3) & (4) \\
 & Core index & Very proud & Willing fight & Jobs priority \\ \midrule
\multicolumn{5}{l}{\textit{Panel A. Baseline}} \\
Dist. to EQ zone & -0.0488*** & -0.0398* & -0.0468** & -0.0534** \\
  & (0.0180) & (0.0238) & (0.0210) & (0.0234) \\
Std. $\beta$ & -0.098 & -0.052 & -0.066 & -0.075 \\
Obs. &    162,643 &    164,037 &    143,835 &    156,662 \\
\hline
\multicolumn{5}{l}{\textit{Panel B. Excluding districts in earthquake risk zones 3--4}} \\
Dist. to EQ zone & -0.0427** & -0.0381 & -0.0401* & -0.0465* \\
  & (0.0181) & (0.0246) & (0.0208) & (0.0237) \\
Std. $\beta$ & -0.089 & -0.051 & -0.058 & -0.066 \\
Obs. &    133,812 &    135,070 &    120,545 &    129,454 \\
\hline
\multicolumn{5}{l}{\textit{Panel C. Quadratic distance term (Bentzen Table 2 Panel D)}} \\
Dist. to EQ zone & -0.1023*** & -0.0977** & -0.0674* & -0.1318*** \\
  & (0.0328) & (0.0469) & (0.0405) & (0.0318) \\
Distance\textsuperscript{2} &  0.0248** &  0.0273** &  0.0094 &  0.0360*** \\
  & (0.0101) & (0.0129) & (0.0118) & (0.0122) \\
Std. $\beta$ (linear term) & -0.206 & -0.128 & -0.095 & -0.186 \\
Obs. &    162,643 &    164,037 &    143,835 &    156,662 \\
\midrule
Country-year FE & Yes & Yes & Yes & Yes \\
Controls & Yes & Yes & Yes & Yes \\ \bottomrule

\addlinespace[0.15cm]
\multicolumn{5}{p{13cm}}{\scriptsize{\textbf{Notes:} This table presents specification checks following \citet{bentzen2019}. Column~(1) reports the national-core index; columns~(2)--(4) report being very proud, willingness to fight, and jobs priority. Panel~A replicates equation~\eqref{eq:baseline}; Panel~B excludes districts located within high-risk seismic zones; and Panel~C adds a quadratic distance term. Panel~B shows that the main estimate is stable when districts located inside the risk zone are excluded. Panel~C allows the distance relationship to be nonlinear. All specifications include country-year fixed effects and baseline controls; standard errors are clustered by region. *, **, and *** indicate statistical significance at the 10\%, 5\%, and 1\% levels, respectively.}} \\
\end{tabular}
\end{adjustbox}
\end{center}
}
\end{table}


\begin{table}[H]
{
\renewcommand{\arraystretch}{0.8}
\setlength{\tabcolsep}{3.2pt}
\captionsetup{font={normalsize,bf}}
\caption{Hazard Specificity: Alternative Seismic Risk Measures}\label{tab:robustness_hazards}
\vspace{-0.5cm}
\begin{center}
\small
\begin{adjustbox}{max width=0.98\linewidth, max totalheight=0.85\textheight, keepaspectratio}%
\begin{tabular}{L{3.0cm}*{8}{C{2.0cm}}}
\toprule
 & (1) & (2) & (3) & (4) & (5) & (6) & (7) & (8) \\
 & Dist. earthquake MM3.4+ & Dist. earthquake MM4+ & Log distance (MM3.4+) & Fault-line dist. & Mean EQ zone & Dist. tsunami & Dist. volcano MM2.6+ & Dist. storm wind MM2.5+ \\ \midrule
\multicolumn{9}{l}{\textit{Panel A. National core index}} \\
Risk measure & -0.0488*** & -0.0277** & -0.1031*** & -0.0308*** & 0.0022 & 0.0143 & -0.0181 & 0.0120 \\
 & (0.0180) & (0.0108) & (0.0326) & (0.0115) & (0.0173) & (0.0124) & (0.0116) & (0.0157) \\
Std. $\beta$ & -0.098 & -0.164 & -0.114 & -0.062 & 0.003 & 0.023 & -0.045 & 0.055 \\
Obs. & 162,643 & 162,643 & 162,643 & 87,356 & 87,398 & 162,643 & 162,643 & 162,643 \\ \addlinespace[0.25cm]
\multicolumn{9}{l}{\textit{Panel B. Very proud of nationality}} \\
Risk measure & -0.0398* & -0.0379** & -0.0914** & -0.0195 & -0.0366 & 0.0250 & -0.0119 & 0.0089 \\
 & (0.0238) & (0.0170) & (0.0463) & (0.0196) & (0.0259) & (0.0206) & (0.0197) & (0.0228) \\
Std. $\beta$ & -0.052 & -0.147 & -0.066 & -0.025 & -0.030 & 0.026 & -0.020 & 0.027 \\
Obs. & 164,037 & 164,037 & 164,037 & 90,132 & 90,859 & 164,037 & 164,037 & 164,037 \\ \addlinespace[0.25cm]
\multicolumn{9}{l}{\textit{Panel C. Willing to fight for country}} \\
Risk measure & -0.0468** & -0.0467*** & -0.0895** & -0.0385** & 0.0186 & 0.0207 & -0.0431*** & 0.0291 \\
 & (0.0210) & (0.0166) & (0.0404) & (0.0173) & (0.0166) & (0.0190) & (0.0165) & (0.0269) \\
Std. $\beta$ & -0.066 & -0.190 & -0.069 & -0.055 & 0.016 & 0.023 & -0.075 & 0.092 \\
Obs. & 143,835 & 143,835 & 143,835 & 74,152 & 74,188 & 143,835 & 143,835 & 143,835 \\ \addlinespace[0.25cm]
\multicolumn{9}{l}{\textit{Panel D. Jobs priority for nationals}} \\
Risk measure & -0.0534** & 0.0021 & -0.1160*** & -0.0333** & 0.0177 & -0.0015 & -0.0009 & -0.0059 \\
 & (0.0234) & (0.0078) & (0.0348) & (0.0162) & (0.0129) & (0.0175) & (0.0080) & (0.0114) \\
Std. $\beta$ & -0.075 & 0.009 & -0.090 & -0.047 & 0.015 & -0.002 & -0.002 & -0.019 \\
Obs. & 156,662 & 156,662 & 156,662 & 82,407 & 82,446 & 156,662 & 156,662 & 156,662 \\ \bottomrule

\addlinespace[0.15cm]
\multicolumn{9}{p{20.8cm}}{\scriptsize{\textbf{Notes:} This table re-estimates the baseline specification replacing the earthquake-risk measure with alternative hazard-proximity variables. Panels~A--D report the national-core index and its three components (very proud, willingness to fight, jobs priority); each cell is a separate regression of the panel outcome on the column exposure measure. All columns report estimates of equation~\eqref{eq:baseline} and include country-year fixed effects and controls for age, age$^2$, sex, marital status, absolute latitude, and distance to the coast. Column~(1) reproduces the Table~\ref{tab:main_results} baseline (distance to earthquake zones~3--4, MM3.4+) and serves as the reference point; column~(2) uses distance to zone~4 alone (MM4+); column~(3) uses log distance (MM3.4+); column~(4) uses distance to the nearest fault line; and column~(5) uses the mean earthquake-zone indicator. Columns~(6)--(8) are three non-seismic placebo hazards: distance to tsunami zones, volcanic fields, and tropical storm tracks. Standard errors are clustered at the region level. *, **, and *** indicate statistical significance at the 10\%, 5\%, and 1\% levels, respectively.}} \\
\end{tabular}
\end{adjustbox}
\end{center}
}
\end{table}


\begin{table}[H]
{
\renewcommand{\arraystretch}{0.8}
\setlength{\tabcolsep}{3.2pt}
\captionsetup{font={normalsize,bf}}
\caption{Alternative Natural-Hazard Proximity Checks}\label{tab:alternative_disasters}
\vspace{-0.5cm}
\begin{center}
\small
\begin{adjustbox}{max width=0.98\linewidth, max totalheight=0.85\textheight, keepaspectratio}%
\begin{tabular}{L{4.0cm}*{6}{C{1.8cm}}}
\toprule
 & (1) & (2) & (3) & (4) & (5) & (6) \\
 & Earthquake & Tsunami & Avg. eq.-tsu. & Min. eq.-tsu. & Volcano & Storm \\ \midrule
\multicolumn{7}{l}{\textit{Panel A. Dep. var.: Core nationalism index}} \\
Distance & -0.0488*** &  0.0143 & -0.0164 &  0.0092 & -0.0181 &  0.0120 \\
  & (0.0180) & (0.0124) & (0.0191) & (0.0156) & (0.0116) & (0.0157) \\
Std. $\beta$ & -0.098 &  0.023 & -0.025 &  0.012 & -0.045 &  0.055 \\
Obs. &    162,643 &    162,643 &    162,643 &    162,643 &    162,643 &    162,643 \\
\hline
\multicolumn{7}{l}{\textit{Panel B. Dep. var.: Very proud of nationality}} \\
Distance & -0.0398* &  0.0250 & -0.0022 & -0.0008 & -0.0119 &  0.0089 \\
  & (0.0238) & (0.0206) & (0.0283) & (0.0230) & (0.0197) & (0.0228) \\
Std. $\beta$ & -0.052 &  0.026 & -0.002 & -0.001 & -0.020 &  0.027 \\
Obs. &    164,037 &    164,037 &    164,037 &    164,037 &    164,037 &    164,037 \\
\hline
\multicolumn{7}{l}{\textit{Panel C. Dep. var.: Willing to fight for country}} \\
Distance & -0.0468** &  0.0207 & -0.0100 & -0.0040 & -0.0431*** &  0.0291 \\
  & (0.0210) & (0.0190) & (0.0259) & (0.0195) & (0.0165) & (0.0269) \\
Std. $\beta$ & -0.066 &  0.023 & -0.011 & -0.004 & -0.075 &  0.092 \\
Obs. &    143,835 &    143,835 &    143,835 &    143,835 &    143,835 &    143,835 \\
\hline
\multicolumn{7}{l}{\textit{Panel D. Dep. var.: Jobs priority for nationals}} \\
Distance & -0.0534** & -0.0015 & -0.0324 &  0.0318 & -0.0009 & -0.0059 \\
  & (0.0234) & (0.0175) & (0.0244) & (0.0196) & (0.0080) & (0.0114) \\
Std. $\beta$ & -0.075 & -0.002 & -0.035 &  0.028 & -0.002 & -0.019 \\
Obs. &    156,662 &    156,662 &    156,662 &    156,662 &    156,662 &    156,662 \\
\bottomrule

\addlinespace[0.15cm]
\multicolumn{7}{p{16cm}}{\scriptsize{\textbf{Notes:} This table tests whether the baseline result is specific to seismic risk. Panels~A--D use the national-core index and its three components (very proud, willingness to fight, jobs priority). Each column replaces the baseline earthquake-risk measure with an alternative natural-hazard proximity variable: column~(1) is the baseline (distance to earthquake zones~3--4); column~(2) uses distance to tsunami zones; column~(3) uses the average of earthquake and tsunami distances; column~(4) uses the minimum; column~(5) uses distance to volcanic fields; column~(6) uses distance to tropical storm tracks. All columns report estimates of equation~\eqref{eq:baseline} with country-year fixed effects and controls for age, age$^2$, sex, marital status, absolute latitude, and distance to the coast. Standard errors are clustered at the region level. *, **, and *** indicate statistical significance at the 10\%, 5\%, and 1\% levels, respectively.}} \\
\end{tabular}
\end{adjustbox}
\end{center}
}
\end{table}


\begin{table}[H]
{
\renewcommand{\arraystretch}{0.85}
\setlength{\tabcolsep}{4pt}
\captionsetup{font={normalsize,bf}}
\caption{Spatial-Sorting Robustness}\label{tab:spatial_sorting}
\vspace{-0.45cm}
\begin{center}
\small
\begin{adjustbox}{max width=0.9\linewidth, max totalheight=0.82\textheight, keepaspectratio}%
\begin{tabular}{L{5.0cm}*{4}{C{2.6cm}}}
\toprule
 & (1) & (2) & (3) & (4) \\
Specification (pre-2009) & Nat. core index & Very proud & Willing to fight & Jobs priority \\
\midrule
Baseline (ethnicity recorded) & -0.140*** & -0.072 & -0.044 & -0.165*** \\
  & [N=  72329] & [N=  71132] & [N=  62715] & [N=  68964] \\ \addlinespace[0.1cm]
+ Ethnic-group fixed effects & -0.155*** & -0.100** & -0.046 & -0.164*** \\
  & [N=  72283] & [N=  71087] & [N=  62676] & [N=  68919] \\ \addlinespace[0.1cm]
Older respondents (age 50+) & -0.137*** & -0.090** & -0.049 & -0.122*** \\
  & [N=  27446] & [N=  27752] & [N=  23553] & [N=  26171] \\ \addlinespace[0.1cm]
Agricultural workers (rural) & -0.098 & -0.075 & -0.044 & -0.115 \\
  & [N=   7188] & [N=   7285] & [N=   6361] & [N=   6798] \\ \addlinespace[0.1cm]
\bottomrule

\addlinespace[0.15cm]
\multicolumn{5}{p{16.5cm}}{\scriptsize{\textbf{Notes:} This table reports standardized seismic-risk coefficients on the pre-2009 sample. Column~(1) reports the national-core index; columns~(2)--(4) report being very proud, willingness to fight, and jobs priority. The baseline row restricts to respondents with ethnic group recorded (WVS item X051); the second row adds ethnic-group fixed effects; the last two rows restrict to low-mobility subsamples (respondents aged~50 and over, and agricultural workers). X051 is a country-specific ethnic-group \emph{membership} classifier used here as a composition control, not a measure of ethnic-identity salience (see Section~\ref{sec:data}); it is distinct from the country-level ethnic-fractionalization moderator of Table~\ref{tab:religious_fractionalization}. All specifications include country-year fixed effects, baseline controls, WVS weights, and region-clustered standard errors. *, **, and *** indicate significance at the 10\%, 5\%, and 1\% levels.}} \\
\end{tabular}
\end{adjustbox}
\end{center}
}
\end{table}


\begin{figure}[H]
\centering
\captionsetup{font={normalsize,bf}}
\caption{Leave-One-Country-Out Robustness}\label{fig:loo}
\vspace{0.3cm}
\includegraphics[width=0.90\textwidth,height=0.78\textheight,keepaspectratio]{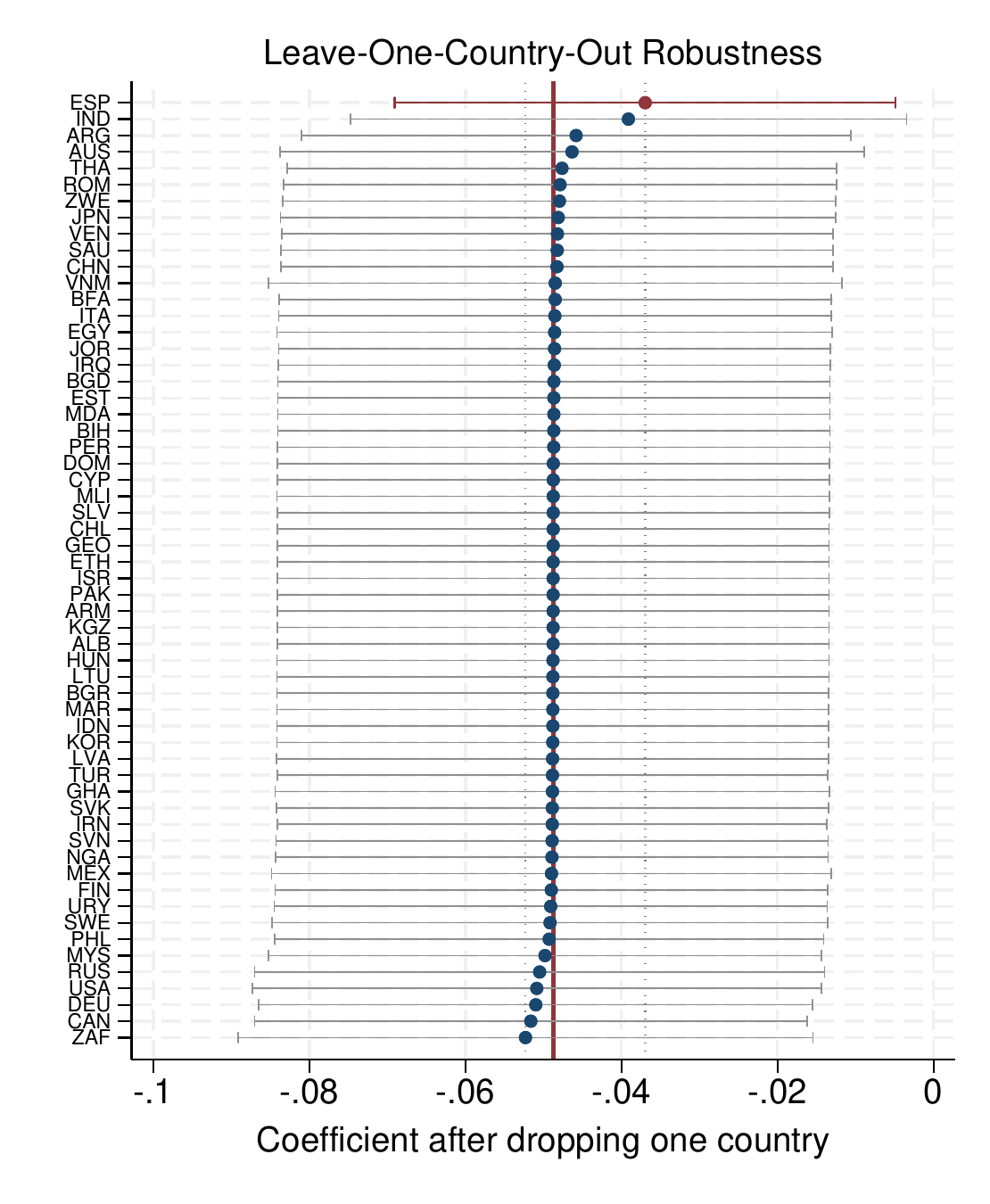}
\vspace{0.1cm}
\begin{minipage}{0.90\textwidth}
\scriptsize\textbf{Notes:} Each row reports the baseline coefficient on distance to high-risk seismic zones for the national-core index, estimated after dropping all observations from one country. The vertical line marks the full-sample estimate; dotted vertical lines mark the range of leave-one-out estimates. The country highlighted in red produces the largest absolute shift; all 63 estimates remain negative. Standard errors are clustered by subnational region.
\end{minipage}
\end{figure}


\begin{table}[H]
{
\renewcommand{\arraystretch}{0.9}
\setlength{\tabcolsep}{3.2pt}
\captionsetup{font={normalsize,bf}}
\caption{Aggregation, Weighting, and Identifying-Sample Diagnostics}\label{tab:aggregation_weighting}
\vspace{-0.45cm}
\begin{center}
\small
\begin{adjustbox}{max width=0.95\linewidth, max totalheight=0.85\textheight, keepaspectratio}%
\begin{tabular}{L{4.6cm}*{5}{C{2.2cm}}}
\toprule
 & (1) & (2) & (3) & (4) & (5) \\
 & \shortstack[c]{Nat. core\\index} & \shortstack[c]{Very proud\\of nation} & \shortstack[c]{Willing to\\fight} & \shortstack[c]{Jobs priority\\for nationals} & \shortstack[c]{Religiosity\\(benchmark)} \\ \midrule
\multicolumn{6}{l}{\textit{Panel A: baseline (Table II reference)}} \\
Dist. to EQ zone & -0.0488*** & -0.0398* & -0.0468** & -0.0534** & -0.0344** \\
 & (0.0180) & (0.0238) & (0.0210) & (0.0234) & (0.0140) \\
$p$, country cluster & 0.065 & 0.243 & 0.076 & 0.097 & 0.057 \\
Obs. &    162,643 &    164,037 &    143,835 &    156,662 &    167,847 \\ \addlinespace[0.12cm]
\multicolumn{6}{l}{\textit{Panel B: region--country-year cells (weighted means)}} \\
Dist. to EQ zone & -0.0472** & -0.0377 & -0.0472** & -0.0542** & -0.0326** \\
 & (0.0202) & (0.0269) & (0.0229) & (0.0258) & (0.0154) \\
$p$, country cluster & 0.089 & 0.307 & 0.086 & 0.097 & 0.076 \\
Obs. &      1,182 &      1,171 &      1,144 &      1,149 &      1,179 \\ \addlinespace[0.12cm]
\multicolumn{6}{l}{\textit{Panel C: equal country weights}} \\
Dist. to EQ zone & -0.0420** & -0.0429* & -0.0436** & -0.0352 & -0.0166 \\
 & (0.0201) & (0.0252) & (0.0206) & (0.0247) & (0.0130) \\
$p$, country cluster & 0.168 & 0.246 & 0.124 & 0.331 & 0.325 \\
Obs. &    162,643 &    164,037 &    143,835 &    156,662 &    167,847 \\ \addlinespace[0.12cm]
\multicolumn{6}{l}{\textit{Panel D: country-years with identifying variation}} \\
Dist. to EQ zone & -0.0480*** & -0.0393* & -0.0458** & -0.0533** & -0.0342** \\
 & (0.0182) & (0.0238) & (0.0211) & (0.0236) & (0.0141) \\
$p$, country cluster & 0.072 & 0.247 & 0.088 & 0.099 & 0.060 \\
Obs. &    146,971 &    147,275 &    131,261 &    141,915 &    151,081 \\ \bottomrule

\addlinespace[0.15cm]
\multicolumn{6}{p{16.6cm}}{\scriptsize{\textbf{Notes:} This table examines how the Table~\ref{tab:main_results} baseline depends on the structure of the estimation sample. Column~(1) reports the national-core index; columns~(2)--(4) report being very proud, willingness to fight, and jobs priority; column~(5) reports the religiosity benchmark. Panel~A reproduces the baseline. Panel~B collapses the data to region--country-year cells (weighted means of outcomes and controls; cells weighted by their WVS weight sums). Panel~C normalizes each country's total WVS weight to one. Panel~D restricts estimation to country-year cells with within-cell variation in the treatment (136 of 171 cells, 55 of 63 countries with any within-country treatment variation). Standard errors in parentheses are clustered by subnational region; the $p$, country cluster row reports $p$-values with clustering at the country level. *, **, and *** indicate significance at the 10\%, 5\%, and 1\% levels based on region-clustered standard errors.}} \\
\end{tabular}
\end{adjustbox}
\end{center}
}
\end{table}


\begin{table}[H]
{
\renewcommand{\arraystretch}{0.85}
\setlength{\tabcolsep}{3.2pt}
\captionsetup{font={normalsize,bf}}
\caption{Country-Level Clustering and Wild Cluster Bootstrap}\label{tab:spatial_inference}
\vspace{-0.45cm}
\begin{center}
\small
\begin{adjustbox}{max width=0.95\linewidth, max totalheight=0.85\textheight, keepaspectratio}%
\begin{tabular}{L{4.6cm}*{5}{C{2.2cm}}}
\toprule
 & (1) & (2) & (3) & (4) & (5) \\
 & \shortstack[c]{Nat. core\\index} & \shortstack[c]{Very proud\\of nation} & \shortstack[c]{Willing to\\fight} & \shortstack[c]{Jobs priority\\for nationals} & \shortstack[c]{Religiosity\\(benchmark)} \\ \midrule
\multicolumn{6}{l}{\textit{Panel A: full sample}} \\
Dist. to EQ zone & -0.0488* & -0.0398 & -0.0468* & -0.0534 & -0.0344* \\
SE, region cluster & (0.0180) & (0.0238) & (0.0210) & (0.0234) & (0.0140) \\
SE, country cluster & [0.0259] & [0.0337] & [0.0259] & [0.0316] & [0.0178] \\
Std. $\beta$ & -0.098 & -0.052 & -0.066 & -0.075 & -0.071 \\
$p$, region cluster & 0.007 & 0.096 & 0.027 & 0.023 & 0.015 \\
$p$, country cluster & 0.065 & 0.243 & 0.076 & 0.097 & 0.057 \\
$p$, wild cluster bootstrap & 0.092 & 0.283 & 0.094 & 0.182 & 0.074 \\
Obs. &    162,643 &    164,037 &    143,835 &    156,662 &    167,847 \\
Countries & 63 & 63 & 62 & 61 & 62 \\ \midrule
\multicolumn{6}{l}{\textit{Panel B: early waves (1981--2005)}} \\
Dist. to EQ zone & -0.0927* & -0.0768 & -0.0649 & -0.1261** & -0.0632* \\
SE, region cluster & (0.0276) & (0.0387) & (0.0326) & (0.0274) & (0.0178) \\
SE, country cluster & [0.0388] & [0.0595] & [0.0387] & [0.0250] & [0.0234] \\
Std. $\beta$ & -0.156 & -0.082 & -0.078 & -0.152 & -0.108 \\
$p$, region cluster & 0.001 & 0.048 & 0.047 & 0.000 & 0.000 \\
$p$, country cluster & 0.021 & 0.203 & 0.101 & 0.000 & 0.009 \\
$p$, wild cluster bootstrap & 0.070 & 0.341 & 0.162 & 0.013 & 0.070 \\
Obs. &     65,143 &     68,703 &     54,022 &     61,236 &     70,133 \\
Countries & 51 & 52 & 46 & 48 & 51 \\ \bottomrule

\addlinespace[0.15cm]
\multicolumn{6}{p{16.6cm}}{\scriptsize{\textbf{Notes:} This table re-estimates the Table~\ref{tab:main_results} baseline and reports, for each outcome, the region-clustered standard error in parentheses, the country-clustered standard error in brackets, and $p$-values under region clustering ($t$ with $G-1$ degrees of freedom), country clustering, and a wild cluster bootstrap at the country level (Rademacher weights, null imposed, 9{,}999 replications; \citealp{roodman2019}).  Panel~A uses the full sample; Panel~B restricts to WVS waves through 2005, when the seismic-risk association is strongest (Section~\ref{sec:results_baseline} and Table~\ref{tab:temporal_heterogeneity}).  Country-level clustering treats all within-country regional variation---the variation that identifies the design---as a single cluster and is therefore a conservative benchmark; by that standard the \citet{bentzen2019} religiosity coefficient (column~5) and the national coefficients behave alike.  Stars are based on the wild-bootstrap $p$-values. *, **, and *** indicate significance at the 10\%, 5\%, and 1\% levels.}} \\
\end{tabular}
\end{adjustbox}
\end{center}
}
\end{table}


\begin{table}[H]
{
\renewcommand{\arraystretch}{0.9}
\setlength{\tabcolsep}{4pt}
\captionsetup{font={normalsize,bf}}
\caption{Spatial-Noise Placebo for the Baseline Association}\label{tab:kelly_noise}
\vspace{-0.45cm}
\begin{center}
\small
\begin{adjustbox}{max width=0.95\linewidth, max totalheight=0.85\textheight, keepaspectratio}%
\begin{tabular}{L{4.6cm}*{6}{C{2cm}}}
\toprule
 & (1) & (2) & (3) & (4) & (5) & (6) \\
 & & & \multicolumn{4}{c}{Noise-based $p$-value, by spatial range of the noise} \\
\cmidrule(lr){4-7}
Outcome & Observed $t$ & \shortstack[c]{95th pct.\\$|t|$, noise} & Fitted & 250 km & 500 km & 1{,}000 km \\ \midrule
National core index &  -2.70 &  2.32 & 0.023 & 0.026 & 0.034 & 0.031 \\
Very proud &  -1.67 &  2.36 & 0.158 & 0.178 & 0.182 & 0.161 \\
Willing to fight &  -2.23 &  2.24 & 0.052 & 0.056 & 0.067 & 0.065 \\
Jobs priority &  -2.28 &  2.58 & 0.088 & 0.085 & 0.092 & 0.103 \\
Religiosity (benchmark) &  -2.45 &  2.83 & 0.089 & 0.071 & 0.080 & 0.075 \\
\midrule
\multicolumn{7}{l}{\textit{Mean rejection rate of pure noise at the nominal 5\% level: 0.12}}\\
\bottomrule

\addlinespace[0.15cm]
\multicolumn{7}{p{18.2cm}}{\scriptsize{\textbf{Notes:} Following \citet{kelly2020}, this table benchmarks the observed baseline $t$-statistics against pure spatial noise. Column~(1) reports the observed region-clustered $t$-statistic; column~(2) the 95th percentile of the absolute simulated $t$-statistics under the fitted covariance; and columns~(3)--(6) the noise-based $p$-values under the fitted, 250~km, 500~km, and 1{,}000~km spatial ranges. For each calibration, 1{,}000 Gaussian random fields are simulated at the 494 region centroids; each field replaces the treatment in equation~\eqref{eq:baseline}. The noise-based $p$-value is the share of simulations whose $|t|$ exceeds the observed $|t|$.}} \\
\end{tabular}
\end{adjustbox}
\end{center}
}
\end{table}


\begin{table}[H]
{
\renewcommand{\arraystretch}{0.9}
\setlength{\tabcolsep}{4pt}
\captionsetup{font={normalsize,bf}}
\caption{The Expressive Association by Symbolic Environment: Sample Splits}\label{tab:pride_regime}
\vspace{-0.45cm}
\begin{center}
\small
\begin{adjustbox}{max width=0.98\linewidth, max totalheight=0.85\textheight, keepaspectratio}%
\begin{tabular}{L{6.2cm}*{7}{C{1.9cm}}}
\toprule
 & (1) & (2) & (3) & (4) & (5) & (6) & (7) \\
Subsample & Coef. & \shortstack[c]{Std.\\$\beta$} & \shortstack[c]{$p$, region\\cluster} & \shortstack[c]{$p$, country\\cluster} & \shortstack[c]{$p$, wild\\bootstrap} & Obs. & Countries \\ \midrule
\multicolumn{8}{l}{\textit{Panel A: WVS weights}} \\
Alignment above mean & -0.1787** & -0.122 & 0.000 & 0.013 & 0.046 &     86,496 & 35 \\
Alignment below mean &  0.0168 &  0.026 & 0.372 & 0.169 & 0.193 &     75,257 & 31 \\
Relig.\ fractionalization below mean & -0.1283* & -0.105 & 0.004 & 0.004 & 0.069 &     75,347 & 31 \\
Relig.\ fractionalization above mean & -0.0142 & -0.021 & 0.554 & 0.662 & 0.864 &     87,975 & 31 \\
Activated regime (both) & -0.1608* & -0.094 & 0.002 & 0.005 & 0.080 &     62,581 & 24 \\
Complement of activated regime & -0.0112 & -0.016 & 0.637 & 0.724 & 0.899 &     98,457 & 39 \\
\addlinespace[0.12cm]
Activated regime, 3-point pride scale & -0.2099* & -0.095 & 0.001 & 0.003 & 0.077 &     62,581 & 24 \\
\addlinespace[0.12cm]
\multicolumn{8}{l}{\textit{Panel B: equal country weights}} \\
Alignment above mean & -0.1741* & -0.149 & 0.003 & 0.042 & 0.073 &     86,496 & 35 \\
Alignment below mean &  0.0035 &  0.006 & 0.873 & 0.883 & 0.874 &     75,257 & 31 \\
Relig.\ fractionalization below mean & -0.0970* & -0.110 & 0.014 & 0.024 & 0.088 &     75,347 & 31 \\
Relig.\ fractionalization above mean & -0.0310 & -0.049 & 0.305 & 0.471 & 0.600 &     87,975 & 31 \\
Activated regime (both) & -0.1023 & -0.094 & 0.071 & 0.135 & 0.164 &     62,581 & 24 \\
Complement of activated regime & -0.0321 & -0.050 & 0.238 & 0.408 & 0.485 &     98,457 & 39 \\
\addlinespace[0.12cm]
Activated regime, 3-point pride scale & -0.1605* & -0.114 & 0.015 & 0.046 & 0.085 &     62,581 & 24 \\
\bottomrule

\addlinespace[0.15cm]
\multicolumn{8}{p{21.4cm}}{\scriptsize{\textbf{Notes:} Each row re-estimates the baseline very-proud specification of equation~\eqref{eq:baseline} on the indicated subsample. Column~(1) reports the distance coefficient; column~(2) its standardized beta; columns~(3)--(5) report $p$-values under region clustering, country clustering, and a country-level wild cluster bootstrap; and columns~(6)--(7) report observations and countries. Splits use the standardized RAS official-support index and the standardized \citet{alesina2003} religious-fractionalization index, each split at its mean; the ``activated regime'' requires alignment above the mean and fractionalization below it. The final row of each panel replaces the very-proud dummy with a three-point pride scale. Panel~A uses WVS weights; Panel~B normalizes each country's total weight to one. All regressions include country-year fixed effects and baseline controls. The sample begins in 1990, when RAS coverage starts. Stars are based on the wild-bootstrap $p$-values: *, **, and *** indicate significance at the 10\%, 5\%, and 1\% levels.}} \\
\end{tabular}
\end{adjustbox}
\end{center}
}
\end{table}


\begin{table}[H]
{
\renewcommand{\arraystretch}{0.8}
\setlength{\tabcolsep}{3.2pt}
\captionsetup{font={normalsize,bf}}
\caption{Realized-Event Estimator and Sample Robustness}\label{tab:event_study}
\vspace{-0.5cm}
\begin{center}
\small
\begin{adjustbox}{max width=0.98\linewidth, max totalheight=0.85\textheight, keepaspectratio}%
\begin{tabular}{L{5.8cm}*{4}{C{2.0cm}}}
\toprule
 & (1) & (2) & (3) & (4) \\
Specification & National core & Very proud & Fight country & Jobs priority \\ \midrule
No individual controls & -0.0102 & -0.0335 & -0.0454** &  0.0385 \\
  & (0.0157) & (0.0252) & (0.0189) & (0.0258) \\
Wild-bootstrap $p$ & 0.544 & 0.218 & 0.032 & 0.198 \\
Observations &   185,710 &   184,536 &   163,698 &   178,305 \\
\addlinespace[0.12cm]
Baseline & -0.0065 & -0.0301 & -0.0435* &  0.0391 \\
  & (0.0175) & (0.0278) & (0.0215) & (0.0255) \\
Wild-bootstrap $p$ & 0.728 & 0.320 & 0.053 & 0.190 \\
Observations &   184,033 &   182,774 &   162,159 &   177,541 \\
\addlinespace[0.12cm]
Minimum cell size $\geq 20$ & -0.0041 & -0.0272 & -0.0371* &  0.0384 \\
  & (0.0165) & (0.0276) & (0.0184) & (0.0261) \\
Wild-bootstrap $p$ & 0.810 & 0.365 & 0.057 & 0.210 \\
Observations &   182,172 &   180,942 &   160,480 &   175,717 \\
\addlinespace[0.12cm]
Unweighted & -0.0088 & -0.0316 & -0.0434* &  0.0339 \\
  & (0.0186) & (0.0285) & (0.0241) & (0.0260) \\
Wild-bootstrap $p$ & 0.670 & 0.308 & 0.073 & 0.256 \\
Observations &   184,033 &   182,774 &   162,159 &   177,541 \\
\midrule
Transition fixed effects & Yes & Yes & Yes & Yes \\
Country-wave-pair $\times$ endpoint FE & Yes & Yes & Yes & Yes \\
\bottomrule

\addlinespace[0.15cm]
\multicolumn{5}{p{14.6cm}}{\scriptsize{\textbf{Notes:} This table checks the stacked repeated-cross-section estimates across four specifications. Column~(1) reports the national-core index; columns~(2)--(4) report being very proud, willingness to fight, and jobs priority. The rows sequentially omit individual controls, report the weighted baseline, restrict to transitions with at least 20 respondents at both endpoints, and omit WVS weights. All specifications include transition fixed effects and country-wave-pair by endpoint fixed effects. Standard errors in parentheses are clustered by country; reported $p$-values use a wild cluster bootstrap with Rademacher weights, the null imposed, and 9{,}999 replications. *, **, and *** indicate significance according to those bootstrap $p$-values at the 10\%, 5\%, and 1\% levels.}} \\
\end{tabular}
\end{adjustbox}
\end{center}
}
\end{table}


\begin{table}[H]
{
\renewcommand{\arraystretch}{0.85}
\setlength{\tabcolsep}{4pt}
\captionsetup{font={normalsize,bf}}
\caption{Event-Study Sensitivity to the Pride Threshold}\label{tab:event_study_pride_threshold}
\vspace{-0.45cm}
\begin{center}
\small
\begin{adjustbox}{max width=0.85\linewidth, max totalheight=0.82\textheight, keepaspectratio}%
\begin{tabular}{L{6.0cm}C{3.0cm}C{3.0cm}}
\toprule
 & (1) & (2) \\
Time since earthquake & Very proud & Very or quite proud \\ \midrule
Any event between waves & -0.0301 & -0.0274** \\
 & (0.0278) & (0.0121) \\
Wild-bootstrap $p$ & 0.320 & 0.036 \\
0--2 years & -0.0326 & -0.0206 \\
 & (0.0327) & (0.0136) \\
Wild-bootstrap $p$ & 0.362 & 0.143 \\
More than 2--4 years & -0.0641 & -0.0389 \\
 & (0.0738) & (0.0267) \\
Wild-bootstrap $p$ & 0.432 & 0.218 \\
More than 4--6 years & -0.0395 & -0.0300 \\
 & (0.0262) & (0.0171) \\
Wild-bootstrap $p$ & 0.254 & 0.142 \\
\midrule
Transition fixed effects & Yes & Yes \\
Country-wave-pair $\times$ endpoint FE & Yes & Yes \\
Observations &   182,774 &   182,774 \\ \bottomrule

\addlinespace[0.15cm]
\multicolumn{3}{p{12.5cm}}{\scriptsize{\textbf{Notes:} This table compares stacked repeated-cross-section responses for two codings of WVS national pride. Column~(1) equals one only for respondents who are ``very proud'' of their nationality; column~(2) equals one for respondents who are either ``very'' or ``quite'' proud. The first row uses any qualifying earthquake between fieldwork periods; the remaining rows split treatment by years since the most recent event and restrict the survey gap to at most six years. The broader outcome has an approximately 90-percent ceiling. All specifications include transition fixed effects, country-wave-pair by endpoint fixed effects, individual controls, and WVS weights. Standard errors in parentheses are clustered by country; reported $p$-values use a wild cluster bootstrap with Rademacher weights, the null imposed, and 9{,}999 replications. *, **, and *** indicate significance according to those bootstrap $p$-values at the 10\%, 5\%, and 1\% levels.}} \\
\end{tabular}
\end{adjustbox}
\end{center}
}
\end{table}


\begin{table}[H]
{
\renewcommand{\arraystretch}{0.85}
\setlength{\tabcolsep}{4pt}
\captionsetup{font={normalsize,bf}}
\caption{Equal-Length Future-Earthquake Placebo}\label{tab:event_study_placebos}
\vspace{-0.45cm}
\begin{center}
\small
\begin{adjustbox}{max width=0.98\linewidth, max totalheight=0.82\textheight, keepaspectratio}%
\begin{tabular}{L{5.8cm}*{6}{C{1.5cm}}}
\toprule
 & (1) & (2) & (3) & (4) & (5) & (6) \\
  & National core & Very proud & Fight country & Jobs priority & Family importance & Generalized trust \\ \midrule
Post $\times$ earthquake between waves &  0.0118 & -0.0005 & -0.0426* &  0.0591 &  0.0180 &  0.0157 \\
  & (0.0152) & (0.0217) & (0.0196) & (0.0365) & (0.0208) & (0.0190) \\
Wild-bootstrap $p$ & 0.462 & 0.985 & 0.078 & 0.190 & 0.625 & 0.424 \\
Post $\times$ future earthquake &  0.0235 &  0.0311 &  0.0581* & -0.0271 & -0.0088 & -0.0034 \\
  & (0.0129) & (0.0211) & (0.0198) & (0.0313) & (0.0155) & (0.0183) \\
Wild-bootstrap $p$ & 0.122 & 0.189 & 0.057 & 0.437 & 0.630 & 0.869 \\
\midrule
Transition fixed effects & Yes & Yes & Yes & Yes & Yes & Yes \\
Country-wave-pair $\times$ endpoint FE & Yes & Yes & Yes & Yes & Yes & Yes \\
Observations &   158,649 &   157,622 &   138,859 &   152,237 &   161,856 &   154,043 \\ \bottomrule

\addlinespace[0.15cm]
\multicolumn{7}{p{16.5cm}}{\scriptsize{\textbf{Notes:} This table jointly includes the earthquake occurring between the two WVS fieldwork periods and an equal-length future-event placebo defined immediately after the later fieldwork period. Columns~(1)--(6) report the national-core index, very proud, willingness to fight, jobs priority, family importance, and generalized trust, respectively, in each outcome's own available sample. The sample is limited to transitions whose full future window is observed before the end of the USGS catalog. All specifications include transition fixed effects, country-wave-pair by endpoint fixed effects, individual controls, WVS weights, and a post-period interaction with chronic seismic-risk distance. Standard errors in parentheses are clustered by country; reported $p$-values use a wild cluster bootstrap with Rademacher weights, the null imposed, and 9{,}999 replications. *, **, and *** indicate significance according to those bootstrap $p$-values at the 10\%, 5\%, and 1\% levels.}} \\
\end{tabular}
\end{adjustbox}
\end{center}
}
\end{table}


\begin{table}[H]
{
\renewcommand{\arraystretch}{0.8}
\setlength{\tabcolsep}{3.2pt}
\captionsetup{font={normalsize,bf}}
\caption{Out-Group and Specificity Outcomes}\label{tab:outgroup}
\vspace{-0.5cm}
\begin{center}
\small
\begin{adjustbox}{max width=0.98\linewidth, max totalheight=0.85\textheight, keepaspectratio}%
\begin{tabular}{L{4.3cm}*{7}{C{1.5cm}}}
\toprule
 & (1) & (2) & (3) & (4) & (5) & (6) & (7) \\
 & Country citizen & World citizen & Trust foreigner & Rej. imm. nbr. & Rej. race nbr. & Rej. foreign nbr. & Immig. policy \\ \midrule
Dist. to EQ zone & -0.0597 & 0.0753 & -0.0234 & -0.0030 & -0.0110 & -0.3729 & -0.0374 \\
 & (0.0863) & (0.0600) & (0.0553) & (0.0163) & (0.0180) & (0.2501) & (0.0468) \\
Std. $\beta$ & -0.031 &  0.064 & -0.019 & -0.005 & -0.019 & -0.329 & -0.027 \\ \midrule
Country-year FE & Yes & Yes & Yes & Yes & Yes & Yes & Yes \\
Controls & Yes & Yes & Yes & Yes & Yes & Yes & Yes \\
Obs. &      5,327 &     68,371 &     95,400 &    157,534 &    160,153 &      2,381 &    109,996 \\ \bottomrule

\addlinespace[0.15cm]
\multicolumn{8}{p{16.4cm}}{\scriptsize{\textbf{Notes:} This table presents estimated effects of proximity to high-risk seismic zones on out-group exclusion and cosmopolitan outcomes. All columns report estimates of equation~\eqref{eq:baseline} and include country-year fixed effects and controls for age, age$^2$, sex, marital status, absolute latitude, and distance to the coast. Each column uses a different outcome measuring rejection of or attitudes toward out-groups and foreigners; see column headers for specific definitions. Distance to high-risk earthquake zones~3--4 is measured in 1,000~km. Standard errors are clustered at the region level. *, **, and *** indicate statistical significance at the 10\%, 5\%, and 1\% levels, respectively.}} \\
\end{tabular}
\end{adjustbox}
\end{center}
}
\end{table}

\begin{table}[H]
{
\renewcommand{\arraystretch}{0.9}
\setlength{\tabcolsep}{4pt}
\captionsetup{font={normalsize,bf}}
\caption{Alignment Moderation: Mean-Religiosity and Macro-Region Horse-Races}\label{tab:alignment_horserace}
\vspace{-0.45cm}
\begin{center}
\small
\begin{adjustbox}{max width=0.9\linewidth, max totalheight=0.85\textheight, keepaspectratio}%
\begin{tabular}{L{7.2cm}*{4}{C{2.2cm}}}
\toprule
 & (1) & (2) & (3) & (4) \\ \midrule
\multicolumn{5}{l}{\textit{Panel A: very proud}} \\
Dist. $\times$ state--religion alignment & -0.0892** & -0.1058** & -0.0509 & -0.0576 \\
  & (0.0401) & (0.0438) & (0.0611) & (0.0556) \\
Dist. $\times$ country mean religiosity & -- &  0.0392 & -- &  0.0167 \\
  &  & (0.0314) &  & (0.0492) \\
Dist. $\times$ macro-region & No & No & Yes & Yes \\
Obs. & 161,753 & 161,660 & 161,753 & 161,660 \\
Countries & 63 & 62 & 63 & 62 \\ \midrule
\multicolumn{5}{l}{\textit{Panel B: national core index}} \\
Dist. $\times$ state--religion alignment & -0.0619** & -0.0773** & -0.0768 & -0.1010** \\
  & (0.0289) & (0.0305) & (0.0591) & (0.0472) \\
Dist. $\times$ country mean religiosity & -- &  0.0376 & -- &  0.0598 \\
  &  & (0.0228) &  & (0.0387) \\
Dist. $\times$ macro-region & No & No & Yes & Yes \\
Obs. & 160,976 & 160,887 & 160,976 & 160,887 \\
Countries & 63 & 62 & 63 & 62 \\ \bottomrule

\addlinespace[0.15cm]
\multicolumn{5}{p{17cm}}{\scriptsize{\textbf{Notes:} This table augments equation~\eqref{eq:baseline} with an interaction between distance to earthquake risk and standardized state--religion alignment, and asks whether the moderation survives competing interactions.  Column~(1) is the baseline interaction; column~(2) adds an interaction with standardized country mean religiosity (weighted country mean of the Bentzen religiosity index); column~(3) instead allows fully separate distance slopes across nine macro-regions; column~(4) includes both.  Moderator main effects are absorbed by the country-year fixed effects.  The alignment moderation strengthens when mean religiosity competes, so it is not general religiosity in disguise; macro-region slopes absorb much of the cross-country variation that identifies the moderator for very proud (Panel~A), while the core-index moderation survives the full specification (Panel~B).  Standard errors, in parentheses, are clustered at the country level. *, **, and *** indicate statistical significance at the 10\%, 5\%, and 1\% levels, respectively.}} \\
\end{tabular}
\end{adjustbox}
\end{center}
}
\end{table}


\begin{table}[H]
{
\renewcommand{\arraystretch}{0.85}
\setlength{\tabcolsep}{4pt}
\captionsetup{font={normalsize,bf}}
\caption{Heterogeneity by Ethnic Fractionalization}\label{tab:ethnic_fractionalization_heterogeneity}
\vspace{-0.45cm}
\begin{center}
\small
\begin{adjustbox}{max width=0.98\linewidth, max totalheight=0.82\textheight, keepaspectratio}%
\begin{tabular}{L{4.5cm}*{4}{C{1.9cm}}}
\toprule
 & (1) & (2) & (3) & (4) \\
 & Nat. core idx. & Very proud & Fight country & Jobs priority \\ \midrule
dist $\times$ z-HIEF & 0.0304 & 0.0249 & 0.0398 & 0.0241 \\
 & (0.0195) & (0.0253) & (0.0278) & (0.0200) \\
$\beta$ (low ethnic frac.) & -0.168 & -0.074 & -0.141 & -0.120 \\
$\beta$ (high ethnic frac.) & -0.044 & -0.007 & -0.026 & -0.050 \\
$p$ (interaction) & 0.124 & 0.330 & 0.158 & 0.233 \\ \midrule
Obs. &    151,978 &    153,326 &    133,997 &    146,430 \\
Adj. $R^2$ & 0.230 & 0.194 & 0.163 & 0.152 \\ \bottomrule

\addlinespace[0.15cm]
\multicolumn{5}{p{13cm}}{\scriptsize{\textbf{Notes:} This table interacts distance to earthquake risk with the standardized Historical Index of Ethnic Fractionalization (HIEF), an annual country panel covering 1945--2013 \citep{drazanova2020}. Columns~(1)--(4) report the national-core index, very proud, willingness to fight, and jobs priority, respectively. None of the four interactions is statistically significant, in contrast to the religious-fractionalization pattern in Table~\ref{tab:religious_fractionalization}. $\beta$~(low/high frac.) reports the standardized distance effect one~SD below and above the moderator mean. All specifications include country-year fixed effects, baseline controls, WVS weights, and country-clustered standard errors. *, **, and *** indicate statistical significance at the 10\%, 5\%, and 1\% levels, respectively.}} \\
\end{tabular}
\end{adjustbox}
\end{center}
}
\end{table}


\begin{table}[H]
{
\renewcommand{\arraystretch}{0.85}
\setlength{\tabcolsep}{4pt}
\captionsetup{font={normalsize,bf}}
\caption{Religiosity Attenuation by Outcome}\label{tab:religious_coping_attenuation}
\vspace{-0.45cm}
\begin{center}
\small
\begin{adjustbox}{max width=0.98\linewidth, max totalheight=0.82\textheight, keepaspectratio}%
\begin{tabular}{L{5.0cm}*{6}{C{1.42cm}}}
\toprule
 & (1) & (2) & (3) & (4) & (5) & (6) \\
Outcome & Baseline coef. & + Relig. coef. & Atten. \% & Base beta & + Relig. beta & N \\
\midrule
National core index & -0.049 & -0.045 &   8.5 & -0.099 & -0.091 &    156,983 \\
Very proud & -0.038 & -0.032 &  16.4 & -0.050 & -0.042 &    158,666 \\
Willing to fight & -0.045 & -0.041 &   8.9 & -0.065 & -0.059 &    138,899 \\
Jobs priority nationals & -0.057 & -0.055 &   2.9 & -0.081 & -0.079 &    151,067 \\
\bottomrule

\addlinespace[0.15cm]
\multicolumn{7}{p{15cm}}{\scriptsize{\textbf{Notes:} This table reports attenuation diagnostics for equation~\eqref{eq:baseline} when the individual religiosity index is added as a control. Column~(1) reports the baseline distance coefficient; column~(2) the coefficient after adding religiosity; column~(3) the percentage attenuation; columns~(4)--(5) the corresponding standardized betas; and column~(6) the common-sample observation count. Rows use the national-core index and its three components, with a common sample across specifications within each row. All specifications include country-year fixed effects, WVS weights, baseline controls, and region-clustered standard errors. Statistical significance is intentionally not marked because the table compares descriptive attenuation magnitudes rather than presenting a mediation test.}} \\
\end{tabular}
\end{adjustbox}
\end{center}
}
\end{table}


\begin{table}[H]
{
\renewcommand{\arraystretch}{0.85}
\setlength{\tabcolsep}{4pt}
\captionsetup{font={normalsize,bf}}
\caption{State Capacity as a Scope Condition}\label{tab:state_capacity_heterogeneity}
\vspace{-0.45cm}
\begin{center}
\small
\begin{adjustbox}{max width=0.98\linewidth, max totalheight=0.82\textheight, keepaspectratio}%
\begin{tabular}{L{4.8cm}*{5}{C{1.85cm}}}
\toprule
 & (1) & (2) & (3) & (4) & (5) \\
 & Nat. core index & Jobs priority & Religiosity idx. & Fam. important & Trust: family \\ \midrule
dist $\times$ z-state cap. & 0.0133 & -0.0063 & 0.0168 & 0.0130 & 0.0302 \\
 & (0.0263) & (0.0234) & (0.0207) & (0.0134) & (0.0331) \\
$\beta$ (low state cap.) & -0.135 & -0.063 & -0.120 & -0.032 & -0.061 \\
$\beta$ (high state cap.) & -0.081 & -0.081 & -0.050 &  0.014 &  0.022 \\
$p$ (interaction) & 0.614 & 0.788 & 0.421 & 0.336 & 0.366 \\ \midrule
Obs. &    162,378 &    156,398 &    167,584 &    166,222 &    101,621 \\
Adj. $R^2$ & 0.230 & 0.149 & 0.366 & 0.052 & 0.060 \\ \bottomrule

\addlinespace[0.15cm]
\multicolumn{6}{p{15.3cm}}{\scriptsize{\textbf{Notes:} This table presents heterogeneity estimates for equation~\eqref{eq:baseline} augmented with an interaction between distance to earthquake risk and a standardized country-year state-capacity index combining GDP per capita, V-Dem rule of law, public-sector cleanliness, and Polity5 democracy. Column~(1) uses the national-core index; column~(2) jobs priority; column~(3) the religiosity index; columns~(4)--(5) family importance and family trust. $\beta$~(low/high state cap.) reports the standardized distance effect one~SD below and above the moderator mean. All specifications include country-year fixed effects, baseline controls, WVS weights, and country-clustered standard errors. Distance is measured in 1,000~km. *, **, and *** indicate significance at the 10\%, 5\%, and 1\% levels.}} \\
\end{tabular}
\end{adjustbox}
\end{center}
}
\end{table}


\begin{table}[H]
{
\renewcommand{\arraystretch}{0.78}
\setlength{\tabcolsep}{1.7pt}
\captionsetup{font={normalsize,bf}}
\caption{Country-Level Heterogeneity by Mechanism
Block}\label{tab:country_heterogeneity_blocks}
\vspace{-0.45cm}
\begin{center}
\small
\begin{adjustbox}{max width=0.98\linewidth, max totalheight=0.84\textheight, keepaspectratio}%
\begin{tabular}{L{3.6cm}*{5}{C{2.0cm}}}
\toprule
Moderator & (1) & (2) & (3) & (4) & (5) \\
 & Nat. core index & Very proud & Fight country & Jobs priority & Religiosity \\ \midrule
\multicolumn{6}{l}{\textit{Block 1. Hard state capacity}} \\
\multicolumn{6}{l}{\small Hard state-cap. index} \\
dist $\times$ z-mod & -0.0980 & -0.0902 & -0.0494 & -0.1700** & -0.0320 \\
 & (0.0610) & (0.0721) & (0.0604) & (0.0682) & (0.0303) \\
\addlinespace[0.04cm]
\multicolumn{6}{l}{\small State authority} \\
dist $\times$ z-mod & -0.0620 & -0.0520 & -0.0451 & -0.0985* & -0.0168 \\
 & (0.0383) & (0.0377) & (0.0367) & (0.0510) & (0.0177) \\
\addlinespace[0.04cm]
\multicolumn{6}{l}{\small Fiscal capacity} \\
dist $\times$ z-mod & -0.0214 & -0.0171 & 0.0107 & -0.0572** & -0.0102 \\
 & (0.0366) & (0.0578) & (0.0451) & (0.0280) & (0.0219) \\
\addlinespace[0.05cm]
\multicolumn{6}{l}{\textit{Block 2. Revealed disaster response}} \\
\multicolumn{6}{l}{\small Revealed EQ response} \\
dist $\times$ z-mod & 0.0141 & 0.0166 & -0.0075 & 0.0245 & 0.0190 \\
 & (0.0253) & (0.0284) & (0.0230) & (0.0387) & (0.0145) \\
\addlinespace[0.05cm]
\multicolumn{6}{l}{\textit{Block 3. Subnational scale and civil society}} \\
\multicolumn{6}{l}{\small Subnational elec. index} \\
dist $\times$ z-mod & 0.0344 & 0.0274 & 0.0480 & 0.0197 & 0.0063 \\
 & (0.0292) & (0.0362) & (0.0328) & (0.0355) & (0.0199) \\
\addlinespace[0.04cm]
\multicolumn{6}{l}{\small Civil society index} \\
dist $\times$ z-mod & 0.0079 & 0.0062 & 0.0286 & -0.0153 & 0.0022 \\
 & (0.0205) & (0.0228) & (0.0262) & (0.0261) & (0.0138) \\
\midrule
Country-year FE & Yes & Yes & Yes & Yes & Yes \\
Controls & Yes & Yes & Yes & Yes & Yes \\
\bottomrule

\addlinespace[0.15cm]
\multicolumn{6}{p{15.6cm}}{\scriptsize{\textbf{Notes:} This table reports heterogeneity estimates for equation~\eqref{eq:baseline} augmented with interactions between distance to earthquake risk and country-year moderators, grouped by conceptual block. Each row is a separate regression; columns~(1)--(5) report the national-core index, very proud, willingness to fight, jobs priority, and religiosity. The hard state-capacity moderators of Block~1 amplify the jobs-priority gradient more clearly than the expressive components or religiosity, consistent with a scope condition for the distributive margin; the remaining blocks are null across outcomes. All specifications include country-year fixed effects, WVS weights, baseline controls, and country-clustered standard errors. *, **, and *** indicate significance at the 10\%, 5\%, and 1\% levels.}} \\
\end{tabular}
\end{adjustbox}
\end{center}
}
\end{table}


\begin{table}[H]
{
\renewcommand{\arraystretch}{0.85}
\setlength{\tabcolsep}{4pt}
\captionsetup{font={normalsize,bf}}
\caption{Institutional Confidence Outcomes}\label{tab:institutional_confidence}
\vspace{-0.45cm}
\begin{center}
\small
\begin{adjustbox}{max width=0.98\linewidth, max totalheight=0.85\textheight, keepaspectratio}%
\begin{tabular}{L{3.6cm}*{8}{C{1.65cm}}}
\toprule
 & (1) & (2) & (3) & (4) & (5) & (6) & (7) & (8) \\
 & State idx. & Government & Civil serv. & Social sec. & Police & Armed forces & Churches & UN \\ \midrule
Dist. to EQ zone & 0.0167 & 0.0156 & 0.0124 & -0.0305 & 0.0179 & -0.0395** & -0.0051 & 0.0159 \\
 & (0.0188) & (0.0267) & (0.0151) & (0.0464) & (0.0172) & (0.0178) & (0.0127) & (0.0160) \\
Std. $\beta$ &  0.044 &  0.032 &  0.027 & -0.056 &  0.036 & -0.081 & -0.009 &  0.033 \\ \midrule
Country-year FE & Yes & Yes & Yes & Yes & Yes & Yes & Yes & Yes \\
Controls & Yes & Yes & Yes & Yes & Yes & Yes & Yes & Yes \\
Obs. &    162,713 &    154,652 &    157,354 &      7,856 &    161,623 &    159,516 &    164,716 &    135,282 \\ \bottomrule

\addlinespace[0.15cm]
\multicolumn{9}{p{19cm}}{\scriptsize{\textbf{Notes:} This table presents estimated seismic-risk coefficients for eight institutional-confidence outcomes. All columns estimate equation~\eqref{eq:baseline} with country-year fixed effects and baseline controls. Column~(1) uses a composite state-institution confidence index; columns~(2)--(8) use confidence in government, civil service, social security, police, armed forces, churches, and the United Nations. Distance is measured in 1,000~km; standard errors are clustered by region. *, **, and *** indicate significance at the 10\%, 5\%, and 1\% levels.}} \\
\end{tabular}
\end{adjustbox}
\end{center}
}
\end{table}


\begin{table}[H]
{
\renewcommand{\arraystretch}{0.85}
\setlength{\tabcolsep}{4pt}
\captionsetup{font={normalsize,bf}}
\caption{Seismic-Risk Gradient Controlling for Distance to the International Border}\label{tab:armed_forces_border}
\vspace{-0.45cm}
\begin{center}
\small
\begin{adjustbox}{max width=0.98\linewidth, max totalheight=0.82\textheight, keepaspectratio}%
\begin{tabular}{L{2.7cm}*{6}{C{2.0cm}}}
\toprule
 & (1) & (2) & (3) & (4) & (5) & (6) \\
 & \shortstack[c]{Nat.\\core} & \shortstack[c]{Very\\proud} & \shortstack[c]{Fight} & \shortstack[c]{Jobs} & \shortstack[c]{Relig-\\iosity} & \shortstack[c]{Armed\\forces} \\ \midrule
Dist. to EQ zone & -0.0787*** & -0.0624* & -0.0505* & -0.1204*** & -0.0460** & -0.0462** \\
 & (0.0232) & (0.0355) & (0.0285) & (0.0235) & (0.0196) & (0.0219) \\
Std. $\beta$ & -0.154 & -0.079 & -0.071 & -0.167 & -0.090 & -0.094 \\ \addlinespace[0.10cm]
Dist. to int'l border & -0.0287 & -0.0610 & -0.0356 &  0.0075 & -0.0319** & -0.0289 \\
 & (0.0216) & (0.0394) & (0.0291) & (0.0233) & (0.0155) & (0.0199) \\
Std. $\beta$ & -0.047 & -0.065 & -0.042 &  0.009 & -0.052 & -0.049 \\ \midrule
Obs. &    87,790 &    88,452 &    74,500 &    83,206 &    90,684 &    84,184 \\
\bottomrule

\addlinespace[0.15cm]
\multicolumn{7}{p{15.8cm}}{\scriptsize{\textbf{Notes:} This table addresses the concern that the seismic-risk gradient reflects proximity to a (often militarized) international border rather than seismic risk. Every column adds distance to the nearest international land border as a control. Columns~(1)--(5) report the baseline outcomes (national-core index, very proud, willingness to fight, jobs priority, and religiosity); column~(6) reports confidence in the armed forces, the only statistically significant item in the institutional-confidence battery of Table~\ref{tab:institutional_confidence}. For each regressor the table reports the coefficient, its region-clustered standard error in parentheses, and a standardized $\beta$. The seismic coefficient remains negative and significant for every outcome, and the border coefficient is itself small and insignificant except for religiosity. The estimation sample is restricted to observations with non-missing border distance, which is concentrated in earlier waves where the gradient is somewhat larger than in the full sample. All columns include country-year fixed effects, baseline controls, and WVS weights. *, **, and *** indicate significance at the 10\%, 5\%, and 1\% levels.}} \\
\end{tabular}
\end{adjustbox}
\end{center}
}
\end{table}


\begin{table}[H]
{
\renewcommand{\arraystretch}{0.88}
\setlength{\tabcolsep}{5pt}
\captionsetup{font={normalsize,bf}}
\caption{Authoritarian-Mobilization Check: Democracy and Accountability}\label{tab:country_heterogeneity_selected}
\vspace{-0.45cm}
\begin{center}
\small
\begin{adjustbox}{max width=0.98\linewidth, max totalheight=0.82\textheight, keepaspectratio}%
\begin{tabular}{L{3.6cm}*{5}{C{2.0cm}}}
\toprule
Moderator & (1) & (2) & (3) & (4) & (5) \\
 & Nat. core index & Very proud & Fight country & Jobs priority & Religiosity \\ \midrule
\multicolumn{6}{l}{\textit{Democracy/accountability}} \\
dist $\times$ z-mod & 0.0098 & 0.0045 & 0.0191 & 0.0042 & 0.0033 \\
 & (0.0196) & (0.0241) & (0.0226) & (0.0268) & (0.0143) \\
$\beta$ (low democracy) & -0.127 & -0.061 & -0.105 & -0.084 & -0.082 \\
$\beta$ (high democracy) & -0.087 & -0.049 & -0.052 & -0.072 & -0.068 \\
$p$ (interaction) & 0.619 & 0.853 & 0.401 & 0.875 & 0.819 \\
Obs. &    162,378 &    163,772 &    143,572 &    156,398 &    167,584 \\
\addlinespace[0.05cm]
\multicolumn{6}{l}{\textit{V-Dem electoral democracy}} \\
dist $\times$ z-mod & 0.0109 & 0.0024 & 0.0192 & 0.0085 & 0.0017 \\
 & (0.0230) & (0.0280) & (0.0253) & (0.0300) & (0.0148) \\
$\beta$ (low democracy) & -0.130 & -0.057 & -0.106 & -0.093 & -0.077 \\
$\beta$ (high democracy) & -0.086 & -0.050 & -0.051 & -0.069 & -0.070 \\
$p$ (interaction) & 0.636 & 0.931 & 0.449 & 0.779 & 0.907 \\
Obs. &    162,378 &    163,772 &    143,572 &    156,398 &    167,584 \\
\midrule
Country-year FE & Yes & Yes & Yes & Yes & Yes \\
Controls & Yes & Yes & Yes & Yes & Yes \\
\bottomrule

\addlinespace[0.15cm]
\multicolumn{6}{p{15.6cm}}{\scriptsize{\textbf{Notes:} This table reports heterogeneity estimates for equation~\eqref{eq:baseline} augmented with interactions between distance to earthquake risk and country-year measures of democracy and executive accountability. Each block is a separate regression using the moderator named in italics; columns~(1)--(5) report the national-core index, very proud, willingness to fight, jobs priority, and religiosity. The table reports the interaction coefficient and standard error, standardized distance effects one standard deviation below and above the moderator mean, the interaction $p$-value, and observations. No interaction is significant for any outcome, consistent with the seismic-risk gradient being orthogonal to democratic regime type. All specifications include country-year fixed effects, WVS weights, baseline controls, and country-clustered standard errors. *, **, and *** indicate significance at the 10\%, 5\%, and 1\% levels.}} \\
\end{tabular}
\end{adjustbox}
\end{center}
}
\end{table}


\begin{table}[H]
{
\renewcommand{\arraystretch}{0.85}
\setlength{\tabcolsep}{4pt}
\captionsetup{font={normalsize,bf}}
\caption{Heterogeneity by Individual
Vulnerability}\label{tab:micro_heterogeneity}
\vspace{-0.45cm}
\begin{center}
\small
\begin{adjustbox}{max width=0.98\linewidth, max totalheight=0.82\textheight, keepaspectratio}%
\begin{tabular}{L{3.6cm}*{5}{C{2.0cm}}}
\toprule
Moderator & (1) & (2) & (3) & (4) & (5) \\
 & Nat. core index & Very proud & Fight country & Jobs priority & Religiosity \\ \midrule
\multicolumn{6}{l}{\textit{Low education}} \\
dist $\times$ group-1 & 0.0109* & 0.0094 & 0.0094 & 0.0156** & -0.0064 \\
 & (0.0065) & (0.0076) & (0.0078) & (0.0078) & (0.0043) \\
$\beta$ (group 0) & -0.113 & -0.060 & -0.053 & -0.110 & -0.083 \\
$\beta$ (group 1) & -0.091 & -0.048 & -0.039 & -0.088 & -0.097 \\
$p$ (interaction) & 0.094 & 0.216 & 0.229 & 0.045 & 0.138 \\
Obs. &    117,292 &    118,455 &    102,421 &    113,270 &    121,384 \\
\addlinespace[0.05cm]
\multicolumn{6}{l}{\textit{Unemployed}} \\
dist $\times$ group-1 & 0.0024 & 0.0026 & -0.0038 & 0.0070 & -0.0099** \\
 & (0.0043) & (0.0060) & (0.0069) & (0.0064) & (0.0041) \\
$\beta$ (group 0) & -0.096 & -0.053 & -0.063 & -0.074 & -0.068 \\
$\beta$ (group 1) & -0.091 & -0.049 & -0.068 & -0.064 & -0.089 \\
$p$ (interaction) & 0.573 & 0.663 & 0.580 & 0.272 & 0.015 \\
Obs. &    156,961 &    158,419 &    138,963 &    151,070 &    162,471 \\
\addlinespace[0.05cm]
\multicolumn{6}{l}{\textit{Agricultural worker}} \\
dist $\times$ group-1 & 0.0077 & 0.0000 & 0.0045 & 0.0227** & -0.0150 \\
 & (0.0084) & (0.0130) & (0.0115) & (0.0113) & (0.0099) \\
$\beta$ (group 0) & -0.148 & -0.071 & -0.068 & -0.158 & -0.084 \\
$\beta$ (group 1) & -0.134 & -0.071 & -0.062 & -0.129 & -0.112 \\
$p$ (interaction) & 0.362 & 0.999 & 0.694 & 0.045 & 0.133 \\
Obs. &     78,856 &     80,401 &     67,289 &     74,934 &     82,126 \\
\addlinespace[0.05cm]
\multicolumn{6}{l}{\textit{Low state confidence}} \\
dist $\times$ group-1 & 0.0157*** & 0.0129** & 0.0017 & 0.0282*** & 0.0071* \\
 & (0.0047) & (0.0060) & (0.0069) & (0.0077) & (0.0037) \\
$\beta$ (group 0) & -0.113 & -0.060 & -0.068 & -0.096 & -0.077 \\
$\beta$ (group 1) & -0.081 & -0.043 & -0.066 & -0.056 & -0.062 \\
$p$ (interaction) & $<$0.001 & 0.032 & 0.804 & $<$0.001 & 0.052 \\
Obs. &    155,887 &    156,716 &    138,564 &    151,956 &    160,354 \\
\midrule
Country-year FE & Yes & Yes & Yes & Yes & Yes \\
Controls & Yes & Yes & Yes & Yes & Yes \\
\bottomrule

\addlinespace[0.15cm]
\multicolumn{6}{p{15.6cm}}{\scriptsize{\textbf{Notes:} This table presents heterogeneity estimates for equation~\eqref{eq:baseline} augmented with interactions between distance to earthquake risk and individual-level vulnerability characteristics. Each row group is a separate interaction with a different individual moderator; columns~(1)--(5) report the national-core index, very proud, willingness to fight, jobs priority, and religiosity. Group~0 and Group~1 betas report standardized distance effects for respondents without and with the characteristic, respectively. Because these moderators may be post-treatment, the interactions are exploratory rather than causal. All specifications include country-year fixed effects, WVS weights, baseline controls, and region-clustered standard errors. *, **, and *** indicate statistical significance at the 10\%, 5\%, and 1\% levels, respectively.}} \\
\end{tabular}
\end{adjustbox}
\end{center}
}
\end{table}


\begin{table}[H]
{
\renewcommand{\arraystretch}{0.85}
\setlength{\tabcolsep}{3pt}
\captionsetup{font={normalsize,bf}}
\caption{National Orientation and Respondents' Immigrant Status}\label{tab:jobs_native}
\vspace{-0.45cm}
\begin{center}
\small
\begin{adjustbox}{max width=0.98\linewidth, max totalheight=0.82\textheight, keepaspectratio}%
\begin{tabular}{L{4.0cm}*{4}{C{2.5cm}}}
\toprule
 & (1) & (2) & (3) & (4) \\
 & Full baseline & Status observed & Native-born only & + Status control \\
\midrule
\multicolumn{5}{l}{\textit{Panel A. National-core index}} \\
Dist. to EQ zone & -0.0488*** & -0.0094 & -0.0070 & -0.0094 \\
  & (0.0180) & (0.0241) & (0.0243) & (0.0242) \\
Std. $\beta$ & -0.098 & -0.020 & -0.015 & -0.020 \\
Obs. &   162,643 &    62,053 &    59,717 &    62,053 \\
Countries &  63 &  40 &  40 &  40 \\
Survey years & 1981--2022 & 2010--2022 & 2010--2022 & 2010--2022 \\
\hline
\multicolumn{5}{l}{\textit{Panel B. Very proud}} \\
Dist. to EQ zone & -0.0398* & -0.0085 & -0.0077 & -0.0084 \\
  & (0.0238) & (0.0237) & (0.0240) & (0.0237) \\
Std. $\beta$ & -0.052 & -0.012 & -0.011 & -0.012 \\
Obs. &   164,037 &    61,371 &    59,502 &    61,371 \\
Countries &  63 &  40 &  40 &  40 \\
Survey years & 1981--2022 & 2010--2022 & 2010--2022 & 2010--2022 \\
\hline
\multicolumn{5}{l}{\textit{Panel C. Willingness to fight}} \\
Dist. to EQ zone & -0.0468** & -0.0238 & -0.0166 & -0.0236 \\
  & (0.0210) & (0.0290) & (0.0293) & (0.0290) \\
Std. $\beta$ & -0.066 & -0.034 & -0.024 & -0.034 \\
Obs. &   143,835 &    57,957 &    55,711 &    57,957 \\
Countries &  62 &  40 &  40 &  40 \\
Survey years & 1981--2022 & 2010--2022 & 2010--2022 & 2010--2022 \\
\hline
\multicolumn{5}{l}{\textit{Panel D. Jobs priority}} \\
Dist. to EQ zone & -0.0534** &  0.0061 &  0.0026 &  0.0059 \\
  & (0.0234) & (0.0354) & (0.0363) & (0.0355) \\
Std. $\beta$ & -0.075 &  0.009 &  0.004 &  0.009 \\
Obs. &   156,662 &    61,076 &    58,696 &    61,076 \\
Countries &  61 &  40 &  40 &  40 \\
Survey years & 1989--2022 & 2010--2022 & 2010--2022 & 2010--2022 \\
\hline
\multicolumn{5}{l}{\textit{Panel E. Religiosity (benchmark)}} \\
Dist. to EQ zone & -0.0344** & -0.0172 & -0.0185 & -0.0170 \\
  & (0.0140) & (0.0137) & (0.0136) & (0.0137) \\
Std. $\beta$ & -0.071 & -0.040 & -0.044 & -0.040 \\
Obs. &   167,847 &    61,742 &    59,355 &    61,742 \\
Countries &  62 &  40 &  40 &  40 \\
Survey years & 1981--2022 & 2010--2022 & 2010--2022 & 2010--2022 \\
\bottomrule

\addlinespace[0.15cm]
\multicolumn{5}{p{15.6cm}}{\scriptsize{\textbf{Notes:} This table reports the seismic-risk coefficient for each baseline outcome: the national-core index (Panel~A), very proud (Panel~B), willingness to fight (Panel~C), jobs priority (Panel~D), and religiosity (Panel~E). Column~(1) reproduces the full baseline. Column~(2) restricts the sample to respondents for whom immigrant status is observed (WVS item G027A), which limits coverage to 2010--2022. Column~(3) further restricts that later-wave sample to native-born respondents. Column~(4) retains all respondents with observed status and controls for immigrant status. The seismic gradient for every outcome is already attenuated in the later-wave observed-status sample of columns~(2)--(4)---consistent with the temporal pattern of Section~\ref{sec:results_baseline}---and changes little between the native-born restriction and the immigrant-status control, so respondents' own immigrant status does not account for the within-subsample differences. All columns include country-year fixed effects, baseline controls, and WVS weights; inference uses region-clustered standard errors. Standardized $\beta$ scales the distance coefficient by the ratio of the standard deviation of exposure to that of the outcome. *, **, and *** indicate significance at the 10\%, 5\%, and 1\% levels.}} \\
\end{tabular}
\end{adjustbox}
\end{center}
}
\end{table}

\begin{table}[H]
{
\renewcommand{\arraystretch}{0.85}
\setlength{\tabcolsep}{3pt}
\captionsetup{font={normalsize,bf}}
\caption{Distributive Priority versus General Anti-Immigrant Sentiment}\label{tab:jobs_distributive_validity}
\vspace{-0.45cm}
\begin{center}
\small
\begin{adjustbox}{max width=0.98\linewidth, max totalheight=0.82\textheight, keepaspectratio}%
\begin{tabular}{L{5.0cm}*{5}{C{2.1cm}}}
\toprule
 & (1) & (2) & (3) & (4) & (5) \\
 & Jobs & Immig. & Immig. & Jobs & Jobs \\
 & priority & policy & neighbor & priority & priority \\ \midrule
Distance to risk zone ($\beta$) & -0.0753** & -0.0274 & -0.0046 & -0.1022*** & -0.0985*** \\
  & (0.0331) & (0.0343) & (0.0247) & (0.0308) & (0.0286) \\ \midrule
Anti-immigrant controls & No & No & No & No & Yes \\
Sample & Full & Full & Full & Restr. & Restr. \\
Country-year FE & Yes & Yes & Yes & Yes & Yes \\
Observations &   156,662 &   109,996 &   157,534 &   103,367 &   103,367 \\ \bottomrule

\addlinespace[0.15cm]
\multicolumn{6}{p{15.6cm}}{\scriptsize{\textbf{Notes:} This table probes whether the seismic-risk gradient in jobs priority reflects a demand for membership-based allocation or generalized anti-immigrant sentiment. Columns~(1)--(3) report the standardized distance coefficient for jobs priority (WVS C002), immigration-policy restrictiveness (E143), and rejection of immigrant neighbors (A124\_06): seismic risk predicts the economic-allocation item but neither general anti-immigrant attitude. Columns~(4)--(5) re-estimate the jobs-priority gradient on the sample where both attitudes are observed, without and with them as controls; the coefficient is essentially unchanged, so the seismic component of jobs priority is largely orthogonal to general anti-immigrant sentiment. All columns include country-year fixed effects, baseline controls (age, age squared, sex, marital status, latitude, coastal distance), and WVS weights; region-clustered standard errors in parentheses. The reported coefficient is the standardized $\beta$, which scales the distance coefficient by the ratio of the standard deviation of exposure to that of the outcome. *, **, and *** indicate significance at the 10\%, 5\%, and 1\% levels.}} \\
\end{tabular}
\end{adjustbox}
\end{center}
}
\end{table}


\newpage


\section{A Model of National Priority and Symbolic Complementarity}
\label{app:theory}

This appendix presents a compact formal model of social interaction across
family, nation, and religion under chronic covariate environmental risk.  The model
is not a substitute for identification and its parameters are not structurally
estimated.  Its role is to discipline the mechanism in
Section~\ref{sec:mechanisms}: facing chronic covariate risk, individuals allocate
costly effort across several social margins, which are produced by different
inputs.  Family insurance requires unaffected relatives.  Practical national
priority requires a membership boundary.  Religious coping provides meaning
inside religious communities.  Expressive national attachment requires a shared
moral interpretation at national scale.

The empirical design in the main text estimates long-run associations with
respect to seismic-risk geography.  Accordingly, the model treats the relevant
state variable as the sustained salience of a covariate-risk environment rather
than the occurrence of one particular earthquake.  Regions closer to structural
seismic risk are interpreted as places where the relative usefulness of local,
national, and religious responses is more persistently salient.
The main insight is therefore not simply that disasters increase ``identity.''
It is that the two national responses use different technologies.  The
distributive margin is a \emph{boundary technology}: it converts scarcity into a
claim over jobs, aid, or public resources by distinguishing members from
non-members.  The expressive margin is a \emph{meaning technology}: it converts
material destruction into pride, duty, sacrifice, and willingness to defend the
nation.  The latter is harder to produce because a disaster must be narrated as
a national trial, obligation, or call to reconstruction.  Religion matters in
the model because religious repertoires can lower this symbolic cost when they
can be attached to the nation.

\subsection{Environment}

Individuals choose costly effort in four margins: family or local insurance
$e_{iF}\ge0$, distributive national priority $e_{iD}\ge0$, expressive national
attachment $e_{iE}\ge0$, and religious coping $e_{iR}\ge0$.  Let $S>0$ denote
the salience or expected intensity of the chronic seismic-risk environment, and
let $\rho\in[0,1]$ denote the covariance of losses were that risk to materialize.
Thus $S$ is not an event-time shock indicator; it summarizes the standing
importance of earthquake risk generated by structural hazard, accumulated
experience, institutions, and socialization.  Costs are quadratic:
$C(e)=\tfrac{1}{2}e^2$.  The chosen effort levels are equilibrium allocations
under that prevailing risk environment, not transient responses to one
realized earthquake.

\paragraph{Family or local insurance.}
The family margin represents local mutual insurance.  Its key input is an
unaffected network: relatives and neighbors must retain capacity to transfer
resources or provide support.  Let the marginal return to this effort be
$M_F(\rho)\ge0$, with $M_F'(\rho)<0$ and $M_F(1)=0$.  The utility contribution
of local investment is $M_F(\rho)e_{iF}-\tfrac{1}{2}e_{iF}^2$, so
$e_{iF}^*=M_F(\rho)$.  At $\rho=1$, $e_{iF}^*=0$: the model's family null is
not a claim that families become irrelevant, but a claim that local insurance is
least useful where risk is covariate enough that a disaster would damage the very
households that would normally insure one another.

\paragraph{Religious denominations and social networks.}
Social interaction matters because the return to effort on the scalable
margins depends on whether others invest in the same margin.  Interaction
takes place on a fixed baseline graph $\mathbf{G}_0$: a symmetric,
non-negative $N\times N$ matrix with zero diagonal, where $g_{0,ij}>0$ means
$i$ and $j$ interact.  Each individual belongs to one of $R$ religious
denominations, assigned independently across individuals with population
shares $(s_1,\ldots,s_R)$, so that any pair belongs to different
denominations with probability $\phi=1-\sum_r s_r^2$---the Herfindahl
fractionalization index.  Denominational boundaries discount the usefulness
of a link: links across denominations carry coordination friction
$c\in(0,1)$ on the national margins, and religious coordination operates
only within denominations.

\begin{assumption}[Expected Interaction Networks]\label{ass:networks}
Effort returns depend on the \emph{expected} networks induced by the random
denomination assignment on $\mathbf{G}_0$.  Since a link is
inter-denominational with probability $\phi$, the expected link weight on the
national margins is $(1-\phi)\cdot 1+\phi\,(1-c)=1-c\phi$, and on the
religious margin it is $1-\phi$.  Hence
\[
\mathbf{G}^{D}(\phi)=\mathbf{G}^{E}(\phi)=(1-c\phi)\,\mathbf{G}_0,
\qquad
\mathbf{G}^{R}(\phi)=(1-\phi)\,\mathbf{G}_0,
\]
which are smooth in $\phi$ with
$\partial\mathbf{G}^{D}/\partial\phi=\partial\mathbf{G}^{E}/\partial\phi
=-c\,\mathbf{G}_0\le0$ and
$\partial\mathbf{G}^{R}/\partial\phi=-\mathbf{G}_0\le0$ elementwise, strictly
wherever $g_{0,ij}>0$.
\end{assumption}

Working with the expected rather than the realized networks is a standard
mean-field device: it replaces the random graph by its mean, makes the
comparative statics in $\phi$ well defined, and is exact for the
expected-network economy studied here.  The common-friction structure
deliberately gives religious fragmentation a
possible effect on both national margins.  This is conservative for the paper's
main mechanism.  The distinctive prediction is not that the distributive margin
is entirely insulated from fragmentation, but that only the expressive margin
has an additional symbolic-complementarity channel with religion, through the
cross-term defined below.

\paragraph{Utility: boundary technology and meaning technology.}
Under $\rho=1$ the relevant utility is:
\begin{align}
  U_i &= S\kappa_D e_{iD}+\beta_D\sum_{j\ne i}g_{ij}^D e_{iD}e_{jD}
      +S\kappa_E e_{iE}+\beta_E\sum_{j\ne i}g_{ij}^E e_{iE}e_{jE}
      +S\Omega e_{iR}+\beta_R\sum_{j\ne i}g_{ij}^R e_{iR}e_{jR}
      \nonumber\\
      &\quad
      -\frac{1}{2}e_{iD}^2-\frac{1}{2}e_{iE}^2-\frac{1}{2}e_{iR}^2
      +\gamma_i(a,\phi)e_{iE}e_{iR},
      \label{eq:utility_app}
\end{align}
where $\kappa_D>0$ is the return to the national membership boundary under
scarcity, $\kappa_E>0$ is the secular return to expressive national attachment,
$\Omega>0$ is the intrinsic return to religious coping, and
$\beta_D,\beta_E,\beta_R\ge0$ are within-margin strategic complementarities.

The cross-term $\gamma_i(a,\phi)e_{iE}e_{iR}$ is the symbolic-cost channel.
It captures the idea that religious effort can raise the marginal return to
expressive national effort: rituals, ceremonies, languages of sacrifice, and
images of a wounded community can make material destruction easier to experience
as a national trial or call to reconstruction rather than only as material loss.
The parameter $a$ denotes state-religion alignment.  I assume
$\partial\gamma_i/\partial a\ge0$ and $\partial\gamma_i/\partial\phi\le0$.
Alignment raises $\gamma_i$ because religious and national symbols point more
readily toward the same collective ``we.''  Religious fractionalization lowers
$\gamma_i$ because meanings activated within separate religious groups aggregate
less easily into a common national narrative.  The distributive margin
$e_{iD}$ is structurally absent from this cross-term: priority for nationals can
be produced by a membership boundary and scarcity alone.

\subsection{Equilibrium Conditions}

The distributive first-order condition is independent of $e_{iR}$ and $e_{iE}$:
\begin{equation}
  S\kappa_D+\beta_D\sum_{j\ne i}g_{ij}^D e_{jD}-e_{iD}=0.
  \label{eq:foc_D}
\end{equation}
In matrix form,
\[
(\mathbf{I}-\beta_D\mathbf{G}^D(\phi))\mathbf{e}_D
=S\kappa_D\mathbf{1}.
\]
Thus the distributive margin is a network game within one block.  It may be
amplified by other people making similar distributive claims, but it is not
formally coupled to religiosity.

\begin{assumption}[Distributive Spectral Condition]\label{ass:spectralD}
$\beta_D<1/\lambda_{\max}(\mathbf{G}_0)$.  Because
$\mathbf{G}^D(\phi)=(1-c\phi)\mathbf{G}_0\le\mathbf{G}_0$ elementwise and the
spectral radius is monotone on non-negative matrices, the condition implies
$\beta_D\lambda_{\max}(\mathbf{G}^D(\phi))<1$ for every $\phi\in[0,1]$.
\end{assumption}

Under this condition, $(\mathbf{I}-\beta_D\mathbf{G}^D)^{-1}$ exists and is
non-negative (by the Neumann series, since $\mathbf{G}^D\ge0$), giving:
\begin{equation}
  \mathbf{e}_D^*=S\kappa_D(\mathbf{I}-\beta_D\mathbf{G}^D(\phi))^{-1}\mathbf{1}.
  \label{eq:eD_star}
\end{equation}

The expressive and religious margins are different.  Their first-order
conditions form a coupled block system:
\[
  \underbrace{\begin{pmatrix}
    \mathbf{I}-\beta_E\mathbf{G}^E(\phi) & -\boldsymbol{\Gamma}(a,\phi)\\
    -\boldsymbol{\Gamma}(a,\phi) & \mathbf{I}-\beta_R\mathbf{G}^R(\phi)
  \end{pmatrix}}_{\displaystyle\equiv\,\mathbf{M}_{ER}(a,\phi)}
  \begin{pmatrix}\mathbf{e}_E\\\mathbf{e}_R\end{pmatrix}
  =S\begin{pmatrix}\kappa_E\mathbf{1}\\\Omega\mathbf{1}\end{pmatrix},
\]
where
$\boldsymbol{\Gamma}(a,\phi)=\mathrm{diag}(\gamma_1(a,\phi),\ldots,\gamma_N(a,\phi))$.
The off-diagonal block is what makes expressive nationalism different from
jobs priority: part of the return to expressive national attachment comes from
the religious coping block.

\begin{assumption}[Expressive-Religious Spectral Condition]\label{ass:spectralER}
$\rho(\bar{\mathbf{B}}_{ER})<1$, where
$\bar{\mathbf{B}}_{ER}=\bigl(\begin{smallmatrix}\beta_E\mathbf{G}_0 &
\bar{\boldsymbol{\Gamma}}\\ \bar{\boldsymbol{\Gamma}} & \beta_R\mathbf{G}_0
\end{smallmatrix}\bigr)$ and
$\bar{\boldsymbol{\Gamma}}=\mathrm{diag}(\bar\gamma_1,\ldots,\bar\gamma_N)$
with $\bar\gamma_i=\sup_{a,\phi}\gamma_i(a,\phi)<\infty$.  Because
$\mathbf{B}_{ER}(a,\phi)=\bigl(\begin{smallmatrix}\beta_E\mathbf{G}^E(\phi) &
\boldsymbol{\Gamma}(a,\phi)\\ \boldsymbol{\Gamma}(a,\phi) & \beta_R\mathbf{G}^R(\phi)
\end{smallmatrix}\bigr)\le\bar{\mathbf{B}}_{ER}$ elementwise for every
$(a,\phi)$, the condition implies $\rho(\mathbf{B}_{ER}(a,\phi))<1$
uniformly.
\end{assumption}

Under this condition, $\mathbf{M}_{ER}(a,\phi)=\mathbf{I}_{2N}-\mathbf{B}_{ER}$
is a non-singular M-matrix with $\mathbf{M}_{ER}^{-1}\ge0$, so:
\begin{equation}
  \begin{pmatrix}\mathbf{e}_E^*\\\mathbf{e}_R^*\end{pmatrix}
  =S\,\mathbf{M}_{ER}(a,\phi)^{-1}
  \begin{pmatrix}\kappa_E\mathbf{1}\\\Omega\mathbf{1}\end{pmatrix}
  \gg\mathbf{0}.
  \label{eq:eE_eR_star}
\end{equation}

\subsection{Comparative Statics}

\begin{proposition}[Chronic-Risk Salience and Scalable Responses]
Under Assumption~\ref{ass:spectralD}, $\partial e_{iD}^*/\partial S
=\kappa_D[(\mathbf{I}-\beta_D\mathbf{G}^D)^{-1}\mathbf{1}]_i>0$ for all
$i$.  Under Assumption~\ref{ass:spectralER},
$\partial e_{iE}^*/\partial S>0$ and $\partial e_{iR}^*/\partial S>0$ for
all $i$.
\end{proposition}

\begin{proof}
Differentiate \eqref{eq:eD_star} and \eqref{eq:eE_eR_star} with respect
to $S$.  By the Neumann series, each inverse is non-negative with every
diagonal entry at least one, so each row sum of the inverse is at least one.
Together with $\kappa_D,\kappa_E,\Omega>0$ this gives strict positivity for
every $i$, with no further connectivity requirement.
\end{proof}

\begin{proposition}[Religious Fractionalization and National Coordination]
Under Assumptions~\ref{ass:networks}--\ref{ass:spectralER} and
$\partial\gamma_i/\partial\phi\le0$, an increase in fractionalization $\phi$
weakly attenuates the disaster response on both national margins:
$\partial^2 e_{iD}^*/\partial S\,\partial\phi\le0$ and
$\partial^2 e_{iE}^*/\partial S\,\partial\phi\le0$ for all $i$.  The
distributive inequality is strict for every $i$ with at least one neighbor
in $\mathbf{G}_0$ (row $i$ of $\mathbf{G}_0$ non-zero) and $\beta_D>0$.  The
expressive margin has an additional attenuation channel through
$\partial\gamma_i/\partial\phi\le0$ that is absent from the distributive
margin; the expressive inequality is strict for every $i$ with a
$\mathbf{G}_0$-neighbor and $\beta_E>0$, or with
$\partial\gamma_i/\partial\phi<0$.
\end{proposition}

\begin{proof}
For the distributive margin, $\partial\mathbf{e}_D^*/\partial S=
\kappa_D(\mathbf{I}-\beta_D\mathbf{G}^D(\phi))^{-1}\mathbf{1}$ and, by
Assumption~\ref{ass:networks},
$\partial\mathbf{G}^D/\partial\phi=-c\,\mathbf{G}_0$.  The matrix derivative
identity gives
\[
  \frac{\partial}{\partial\phi}
  \frac{\partial\mathbf{e}_D^*}{\partial S}
  =-c\,\beta_D\,\kappa_D\,
  (\mathbf{I}-\beta_D\mathbf{G}^D)^{-1}\mathbf{G}_0
  (\mathbf{I}-\beta_D\mathbf{G}^D)^{-1}\mathbf{1}\le0.
\]
Strictness: by the Neumann series
$(\mathbf{I}-\beta_D\mathbf{G}^D)^{-1}\mathbf{1}\ge\mathbf{1}$, so
$[\mathbf{G}_0(\mathbf{I}-\beta_D\mathbf{G}^D)^{-1}\mathbf{1}]_i\ge
\sum_j g_{0,ij}>0$ whenever row $i$ of $\mathbf{G}_0$ is non-zero, and
premultiplying by $(\mathbf{I}-\beta_D\mathbf{G}^D)^{-1}\ge\mathbf{I}$
preserves strict positivity of entry $i$.

For the expressive margin,
\[
\frac{\partial\mathbf{M}_{ER}}{\partial\phi}
=\begin{pmatrix}
c\,\beta_E\,\mathbf{G}_0 &
-\partial\boldsymbol{\Gamma}/\partial\phi\\
-\partial\boldsymbol{\Gamma}/\partial\phi & \beta_R\,\mathbf{G}_0
\end{pmatrix}\ge0,
\]
because $\partial\mathbf{G}^E/\partial\phi=-c\,\mathbf{G}_0$,
$\partial\mathbf{G}^R/\partial\phi=-\mathbf{G}_0$, and
$\partial\boldsymbol{\Gamma}/\partial\phi\le0$.  Applying
$\partial\mathbf{M}_{ER}^{-1}/\partial\phi=-\mathbf{M}_{ER}^{-1}
(\partial\mathbf{M}_{ER}/\partial\phi)\mathbf{M}_{ER}^{-1}\le0$ to
\eqref{eq:eE_eR_star} gives
$\partial^2\mathbf{e}_E^*/\partial S\,\partial\phi\le0$ and
$\partial^2\mathbf{e}_R^*/\partial S\,\partial\phi\le0$.  The term involving
$\partial\boldsymbol{\Gamma}/\partial\phi$ is the additional symbolic
channel; strictness follows from the same row-sum argument, using
$\mathbf{M}_{ER}^{-1}\ge\mathbf{I}_{2N}$ and the strict positivity of
$S(\kappa_E\mathbf{1}',\Omega\mathbf{1}')'$.
\end{proof}

Note that under Assumption~\ref{ass:networks} the religious margin also
weakly declines with $\phi$ through its own network channel
($\partial\mathbf{G}^R/\partial\phi=-\mathbf{G}_0$).  This implication is
secondary to the paper's empirical focus on national outcomes.

\begin{proposition}[Alignment Amplifies the Expressive Margin]
Let $a$ parameterize state-religion alignment with
$\partial\gamma_i/\partial a\ge0$ for all $i$.  Then
$\partial^2 e_{iE}^*/\partial S\,\partial a\ge0$ and
$\partial^2 e_{iR}^*/\partial S\,\partial a\ge0$ for all $i$, with strict
inequality for every $i$ with $\partial\gamma_i/\partial a>0$.
There is no direct alignment effect on $e_{iD}^*$ because $a$ enters neither
\eqref{eq:foc_D} nor $\mathbf{G}^D$.
\end{proposition}

\begin{proof}
$\partial\mathbf{M}_{ER}/\partial a=\bigl(\begin{smallmatrix}0 &
-\partial\boldsymbol{\Gamma}/\partial a\\
-\partial\boldsymbol{\Gamma}/\partial a & 0\end{smallmatrix}\bigr)\le0$.
By the matrix derivative identity,
$\partial\mathbf{M}_{ER}^{-1}/\partial a\ge0$, and differentiating
\eqref{eq:eE_eR_star} gives
$\partial^2(\mathbf{e}_E^*,\mathbf{e}_R^*)'/\partial S\,\partial a\ge0$.
For strictness, the off-diagonal block of
$-\partial\mathbf{M}_{ER}/\partial a$ has diagonal entry
$\partial\gamma_i/\partial a>0$; sandwiching by
$\mathbf{M}_{ER}^{-1}\ge\mathbf{I}_{2N}$ and multiplying the strictly
positive vector $S(\kappa_E\mathbf{1}',\Omega\mathbf{1}')'$ keeps the
corresponding entries of the derivative strictly positive.
Equation~\eqref{eq:eD_star} contains neither $a$ nor
$\boldsymbol{\Gamma}$, so the distributive response has no direct alignment
effect.
\end{proof}

The equilibrium structure also clarifies, without a formal proposition, why
the descriptive religiosity-attenuation diagnostics in
Section~\ref{sec:mech_evidence} are natural: in
\eqref{eq:eE_eR_star} the expressive margin is the only national margin
coupled to religious effort (through $\boldsymbol{\Gamma}$), while the
distributive margin \eqref{eq:eD_star} is not.  I deliberately do not state
this as a comparative-static proposition: conditioning on an equilibrium
object ($e_{iR}^*$) is a regression exercise, not a primitive of the model,
and its empirical counterpart conditions on an outcome that is itself
affected by the treatment---the bad-control problem flagged in
Section~\ref{sec:mech_evidence}.  Those diagnostics are therefore read as
descriptive patterns consistent with the coupling structure, not as a
derived prediction.

\subsection{Empirical Mapping and Limits}

The model maps to the mechanism tests in Section~\ref{sec:mechanisms}.

\begin{enumerate}
  \item \emph{Family null.}  The implication $e_{iF}^*=M_F(1)=0$ does not
    say that families are unimportant.  It says that a survey item such as
    family importance should not be the margin most responsive to covariate
    earthquake risk, because the same shock affects many family members at once.
  \item \emph{Jobs priority as a boundary response.}  Proposition~1 predicts
    that greater chronic-risk salience raises distributive national priority through
    $S\kappa_D$ and the national coordination multiplier
    $(I-\beta_DG^D)^{-1}$.  This margin does not require religious coping.
  \item \emph{Religious coping as meaning inside religious communities.}
    Proposition~1 also predicts a positive response to chronic-risk salience through
    $S\Omega$ and the intra-denominational religious network $G^R$.
  \item \emph{State-religion alignment.}  Proposition~3 predicts that alignment
    amplifies expressive nationalism and religious coping by raising
    $\gamma_i$, while leaving the distributive equation directly unchanged.
  \item \emph{Religious fractionalization.}  Proposition~2 allows religious
    fractionalization to weaken both national margins through the common
    coordination-friction channel ($\partial G^D/\partial\phi=
    \partial G^E/\partial\phi=-c\,G_0$), but the expressive margin has an
    additional channel through $\partial\gamma_i/\partial\phi\le0$.  This is
    why the empirical test focuses on whether fragmentation matters
    especially for expressive national attachment and religious
    identification.
  \item \emph{Religiosity-attenuation diagnostics.}  The model offers no
    formal counterpart to the regression exercise of conditioning on
    religiosity: $e_{iR}^*$ is an equilibrium outcome, and its empirical
    analogue is a bad control.  The coupling structure---only the expressive
    margin is tied to religious effort through $\boldsymbol{\Gamma}$---makes
    the descriptive attenuation patterns in Section~\ref{sec:mech_evidence}
    natural, but they are not derived comparative statics.
\end{enumerate}

\medskip
\noindent\emph{Limits.}  Two limits bear emphasis.  First, the model is
static in chronic-risk salience $S$: it characterizes the long-run, cross-sectional
allocation under chronic covariate risk (selectivity and the alignment and
fractionalization moderation) and generates \emph{no} temporal
predictions: it does not model how $S$ evolves after a realized earthquake or
the secular weakening of the association across survey waves (Table~\ref{tab:temporal_heterogeneity}).  The realized-event response is best
read as a comparative static in $S$---a level shift when a local event makes the
chronic risk salient (Section~\ref{sec:mech_explanation})---rather than
model-implied dynamics, and I do not use it to validate the expressive
(alignment and fractionalization) comparative statics.  Second, the comparative statics on alignment and fractionalization follow
from the sign restrictions $\partial\gamma_i/\partial a\ge0$ and
$\partial\gamma_i/\partial\phi\le0$, which are disciplined by the institutional
literature but assumed rather than derived; the expected-network mapping
$\mathbf{G}(\phi)$ of Assumption~\ref{ass:networks} is likewise a modeling
choice---a mean-field reduction of the random denomination
assignment---rather than an estimated object.  
The model is therefore best read as a consistency device for the cross-sectional
mechanism, not as a structural account of dynamics.


\clearpage
\end{document}